\documentclass[oneside,english]{amsart}
\usepackage[T1]{fontenc}
\usepackage[latin9]{inputenc}
\usepackage{geometry}
\usepackage{verbatim}
\usepackage{amsthm}
\usepackage{amstext}
\usepackage{amssymb}
\usepackage{graphicx}
\usepackage{esint}
\usepackage{url}
\usepackage{yhmath}
\usepackage[all]{xy}
\usepackage{color}
\makeatletter
\numberwithin{equation}{section}
\numberwithin{figure}{section}
\theoremstyle{plain}
\newtheorem*{thm*}{\protect\theoremname}
\theoremstyle{plain}
\newtheorem{thm}{\protect\theoremname}[section]
\theoremstyle{plain}
\newtheorem{lem}[thm]{\protect\lemmaname}
\theoremstyle{remark}
\newtheorem{rem}[thm]{\protect\remarkname}
\theoremstyle{plain}
\newtheorem{prop}[thm]{\protect\propositionname}
\theoremstyle{plain}
\newtheorem{cor}[thm]{\protect\corollaryname}
\theoremstyle{plain}

\theoremstyle{plain}
\newtheorem{proc}[thm]{\protect\procname}

\usepackage{mathrsfs}
\usepackage{subfigure}

\AtBeginDocument{

}

\makeatother

\usepackage{babel}
\providecommand{\corollaryname}{Corollary}
\providecommand{\lemmaname}{Lemma}
\providecommand{\propositionname}{Proposition}
\providecommand{\remarkname}{Remark}
\providecommand{\theoremname}{Theorem}
\providecommand{\conjecturename}{Conjecture}
\providecommand{\procname}{Sampling procedure}

\begin{document}
\global\long\def\SLE{\mathrm{SLE}}
\global\long\def\SLEk{\mathrm{SLE}_{\kappa}}
\global\long\def\SLEkappa#1{\mathrm{SLE}_{#1}}
\global\long\def\SLEkapparho#1#2{\mathrm{SLE}_{#1}(#2)}
\global\long\def\SLEmeasure{\mathsf{P}}
\global\long\def\chordal{\oslash}

\global\long\def\PR{\mathsf{P}}
\global\long\def\EX{\mathsf{E}}

\global\long\def\sF{\mathcal{F}}
\global\long\def\sZ{\mathcal{Z}}
\global\long\def\sD{\mathcal{D}}
\global\long\def\sC{\mathcal{C}}
\global\long\def\sL{\mathcal{L}}
\global\long\def\sA{\mathcal{A}}
\global\long\def\sR{\mathcal{R}}
\global\long\def\sS{\mathcal{S}}
\global\long\def\sP{\mathcal{P}}

\global\long\def\bR{\mathbb{R}}
\global\long\def\bRpos{\mathbb{R}_{> 0}}
\global\long\def\bRnn{\mathbb{R}_{\geq 0}}
\global\long\def\bZ{\mathbb{Z}}
\global\long\def\bN{\mathbb{N}}
\global\long\def\bZpos{\mathbb{Z}_{> 0}}
\global\long\def\bZnn{\mathbb{Z}_{\geq 0}}
\global\long\def\bQ{\mathbb{Q}}
\global\long\def\bC{\mathbb{C}}

\global\long\def\Rsphere{\overline{\bC}}
\global\long\def\bD{\mathbb{D}}
\global\long\def\bH{\mathbb{H}}
\global\long\def\re{\Re\mathfrak{e}}
\global\long\def\im{\Im\mathfrak{m}}
\global\long\def\arg{\mathrm{arg}}
\global\long\def\ii{\mathfrak{i}}
\global\long\def\domain{\Lambda}
\global\long\def\bdrypt{\xi}
\global\long\def\bdryptb{\eta}
\global\long\def\bdry{\partial}
\global\long\def\cl#1{\overline{#1}}
\global\long\def\Mob{\mu}
\global\long\def\confmap{\phi}

\global\long\def\OO{\mathcal{O}}
\global\long\def\oo{\mathit{o}}

\global\long\def\ud{\mathrm{d}}
\global\long\def\der#1{\frac{\ud}{\ud#1}}
\global\long\def\pder#1{\frac{\partial}{\partial#1}}
\global\long\def\pdder#1{\frac{\partial^{2}}{\partial#1^{2}}}
\global\long\def\pddder#1{\frac{\partial^{3}}{\partial#1^{3}}}

\global\long\def\set#1{\left\{  #1\right\}  }
\global\long\def\setcond#1#2{\left\{  #1\;\big|\;#2\right\}  }

\global\long\def\sl{\mathfrak{sl}}
\global\long\def\Uqsltwo{\mathcal{U}_{q}(\mathfrak{sl}_{2})}
\global\long\def\Hcp{\Delta}
\global\long\def\qnum#1{\left[#1\right] }
\global\long\def\qfact#1{\left[#1\right]! }
\global\long\def\qbin#1#2{\left[\begin{array}{c}
	#1\\
	#2 
	\end{array}\right]}
\global\long\def\Wd{\mathsf{M}}
\global\long\def\HWsp{\mathsf{H}}
\global\long\def\Wbas{e}
\global\long\def\Tbas{\tau}
\global\long\def\MTbas{\theta}
\global\long\def\Sbas{s}
\global\long\def\TRbas{t}
\global\long\def\Bdryvec{\mathfrak{v}}

\global\long\def\Projection{\widehat{\mathfrak{p}}_{L,R}}
\global\long\def\Embedding{\mathfrak{I}_{L,R}}

\global\long\def\Arch{\mathrm{LP}}
\global\long\def\LP{\Arch}
\global\long\def\Conn{\mathrm{Conn}}
\global\long\def\arch#1#2{[\wideparen{#1,#2}]}
\global\long\def\link#1#2{\arch{#1}{#2}}
\global\long\def\nested{\boldsymbol{\underline{\Cap}}}
\global\long\def\unnested{\boldsymbol{\underline{\cap\cap}}}
\global\long\def\walks{\mathcal{W}}
\global\long\def\Catalan{\mathrm{C}}

\global\long\def\sciOp{\text{\Rightscissors}}
\global\long\def\tieOp{\wp}
\global\long\def\removeArch{/}
\global\long\def\removeLink{\removeArch}

\global\long\def\dmn{\mathrm{dim}}
\global\long\def\spn{\mathrm{span}}
\global\long\def\tens{\otimes}
\global\long\def\unitmat{\mathbb{I}}
\global\long\def\id{\mathrm{id}}
\global\long\def\isom{\cong}
\global\long\def\Kern{\mathrm{Ker}}

\global\long\def\chamber{\mathfrak{X}}
\global\long\def\FWint{\varphi}
\global\long\def\SurfFW{\mathfrak{L}^{\Supset}}
\global\long\def\PartF{\sZ}
\global\long\def\Sol{\sR}
\global\long\def\ConvSet{\sC}
\global\long\def\FKdual{\mathscr{L}}
\global\long\def\Quantumdual{\psi}

\global\long\def\Ampl{\zeta}
\global\long\def\Corr{\chi}
\global\long\def\Orders{\mathrm{VO}}
\global\long\def\SymmGrp{\mathfrak{S}}
\global\long\def\driving{D}

\global\long\def\Pf{\mathrm{Pf}}
\global\long\def\sgn{\mathrm{sgn}}
\global\long\def\dist{\mathrm{dist}}
\global\long\def\const{\mathrm{const.}}
\global\long\def\eps{\varepsilon}
\global\long\def\half{\frac{1}{2}}

\global\long\def\GFF{\mathrm{GFF}}
\global\long\def\Ising{\mathrm{Ising}}
\global\long\def\perco{\mathrm{perco}}


\author{K.~Kytölä and E.~Peltola}

\

\vspace{2.5cm}

\begin{center}
\LARGE \bf \scshape {Pure partition functions of multiple $\SLE$s}
\end{center}

\vspace{0.75cm}

\begin{center}
{\large \scshape Kalle Kyt\"ol\"a}\\
{\footnotesize{\tt kalle.kytola@aalto.fi}}\\
{\small{Department of Mathematics and Systems Analysis}}\\
{\small{P.O. Box 11100, FI-00076 Aalto University, Finland}}\bigskip{}
\\
{\large \scshape Eveliina Peltola}\\
{\footnotesize{\tt hanna.peltola@helsinki.fi}}\\
{\small{Department of Mathematics and Statistics}}\\
{\small{P.O. Box 68, FI-00014 University of Helsinki, Finland}}
\end{center}

\vspace{0.75cm}

\begin{center}
\begin{minipage}{0.85\textwidth} \footnotesize
{\scshape Abstract.}
Multiple Schramm-Loewner Evolutions (SLE) are conformally invariant
random processes of several curves, whose construction by growth processes
relies on partition functions --- M\"obius covariant solutions to
a system of second order partial differential equations. In this article,
we use a quantum group technique to construct a distinguished basis of solutions,
which conjecturally correspond to the extremal points of the
convex set of probability measures of multiple $\SLE$s.
\end{minipage}
\end{center}

\vspace{0.75cm}
\setcounter{tocdepth}{3}

\bigskip{}

\section{\label{sec: intro}Introduction}


In the 1980s, it was recognized that conformal symmetry should emerge in the
scaling limits of two-dimensional models of statistical
physics at critical points of continuous phase transitions
\cite{BPZ-infinite_conformal_symmetry_of_critical_fluctiations, Cardy-conformal_invariance_and_statistical_mechanics}.
This belief has provided an important motivation for remarkable developments in conformal
field theories (CFT) during the past three decades.
Especially in the past fifteen years, also progress in rigorously controlling scaling
limits of lattice models and showing their conformal invariance has been made.
Some such results concern correlations of local fields as in CFT
\cite{Kenyon-conformal_invariance_of_the_domino_tiling,
Hongler-conformal_invariance_of_Ising_model_correlations,
HS-energy_density, 
CHI-conformal_invariance_of_spin_correlations_in_the_planar_Ising_model}
or more general observables
\cite{Smirnov-critical_percolation,
CS-universality_in_2d_Ising, CI-holomorphic_spinor_observables_in_the_critical_Ising_model},
but the majority has been formulated in terms of random interfaces in the
models \cite{LSW-LERW_and_UST, CN-critical_percolation_exploration_path,
Zhan-scaling_limits_of_planar_LERW,
HK-Ising_interfaces_and_free_boundary_conditions,
CDHKS-convergence_of_Ising_interfaces_to_SLE, Izyurov-Smirnovs_observable_for_free_boundary_conditions,
Izyurov-critical_Ising_interfaces_in_multiply_connected_domains}.

The focus on interfaces and random geometry in general was largely inspired by
the seminal article \cite{Schramm-LERW_and_UST} of Schramm.
Schramm observed that if scaling limits of random interfaces between two marked boundary
points of simply connected domains satisfy two natural assumptions,
conformal invariance 
and domain Markov property, 
the limit must fall into a one-parameter family of random curves.
The parameter $\kappa>0$ is the most important characteristic of such interfaces,
and the corresponding curves are known as chordal $\SLEk$,
an abbreviation for Schramm-Loewner Evolution. This classification result is
sometimes referred to as Schramm's principle.

The simplest setup for Schramm's principle is the chordal case described above:
a simply connected domain $\domain$ with a curve connecting two marked
boundary points $\bdrypt, \bdryptb \in \bdry \domain$.
This setup arises in models of statistical mechanics 
when opposite boundary conditions are imposed on two complementary
arcs $\overline{\bdrypt \bdryptb}$ and $\overline{\bdryptb \bdrypt}$ of the domain boundary $\bdry \domain$,
forcing the existence of an interface between $\bdrypt$ and $\bdryptb$.
Some variations of the chordal $\SLE$ setup are obtained if one considers one marked point
on the boundary and another in the bulk, which leads to radial $\SLE$ \cite{Schramm-LERW_and_UST,RS-basic_properties_of_SLE},
or three marked boundary points, which leads to dipolar $\SLE$ \cite{Zhan-PhD_thesis,BBH-dipolar_SLEs}.
For the purposes of statistical mechanics, one of the most natural
generalizations involves dividing the boundary to an even number $2N$ of arcs,
and imposing alternating boundary conditions of opposite type on the arcs.
This forces interfaces to start at $2N$ marked boundary points, and
such interfaces will connect the points pairwise without crossing. 
In this situation, one has $N$ random interfaces,
see Figure~\ref{fig: Ising connectivities N equals 3}.
Such generalizations of $\SLE$-type processes have been considered in
\cite{BBK-multiple_SLEs, Graham-multiple_SLEs, Dubedat-commutation},
and they are generally termed multiple $\SLE$s or $N$-$\SLE$s.
Results about convergence of lattice model interfaces to multiple $\SLE$s
have been obtained in
\cite{CS-universality_in_2d_Ising, Izyurov-PhD_thesis}.

\subsection*{Classification of multiple $\SLE$s by partition functions}

Multiple $\SLE$s should still satisfy conformal invariance
and domain Markov property. An important difference arises, however,
when a classification along the lines of Schramm's principle is attempted.
While any configuration of a simply connected domain with two marked
boundary points $(\domain; \bdrypt, \bdryptb)$ can be conformally mapped to any other
according to Riemann mapping theorem, the same no longer holds with configurations
$(\domain; \bdrypt_1 , \ldots, \bdrypt_{2N})$ of a domain with $2N$ marked boundary points for $N \geq 2$
--- such configurations have nontrivial conformal moduli.
The requirement of conformal invariance is therefore less restrictive, and 
one should expect to find a larger family of random processes of $N$ curves.
These processes form a convex set, as one can randomly select between 
given possibilities. A natural suggestion was put forward in \cite{BBK-multiple_SLEs}:
the extremal points of this convex set should be processes  
which have a deterministic pairwise connectivity of the $2N$ boundary points by the $N$ non-crossing 
curves in the domain. Such extremal processes were termed pure geometries of multiple $\SLE$s.
Pure geometries of multiple $\SLE$s are thus labeled by non-crossing connectivities,
or equivalently, planar pair partitions $\alpha$. For fixed $N$, the number of them is the
Catalan number $\Catalan_N = \frac{1}{N+1} \binom{2N}{N}$.

\begin{figure}
\includegraphics[width=2.5cm]{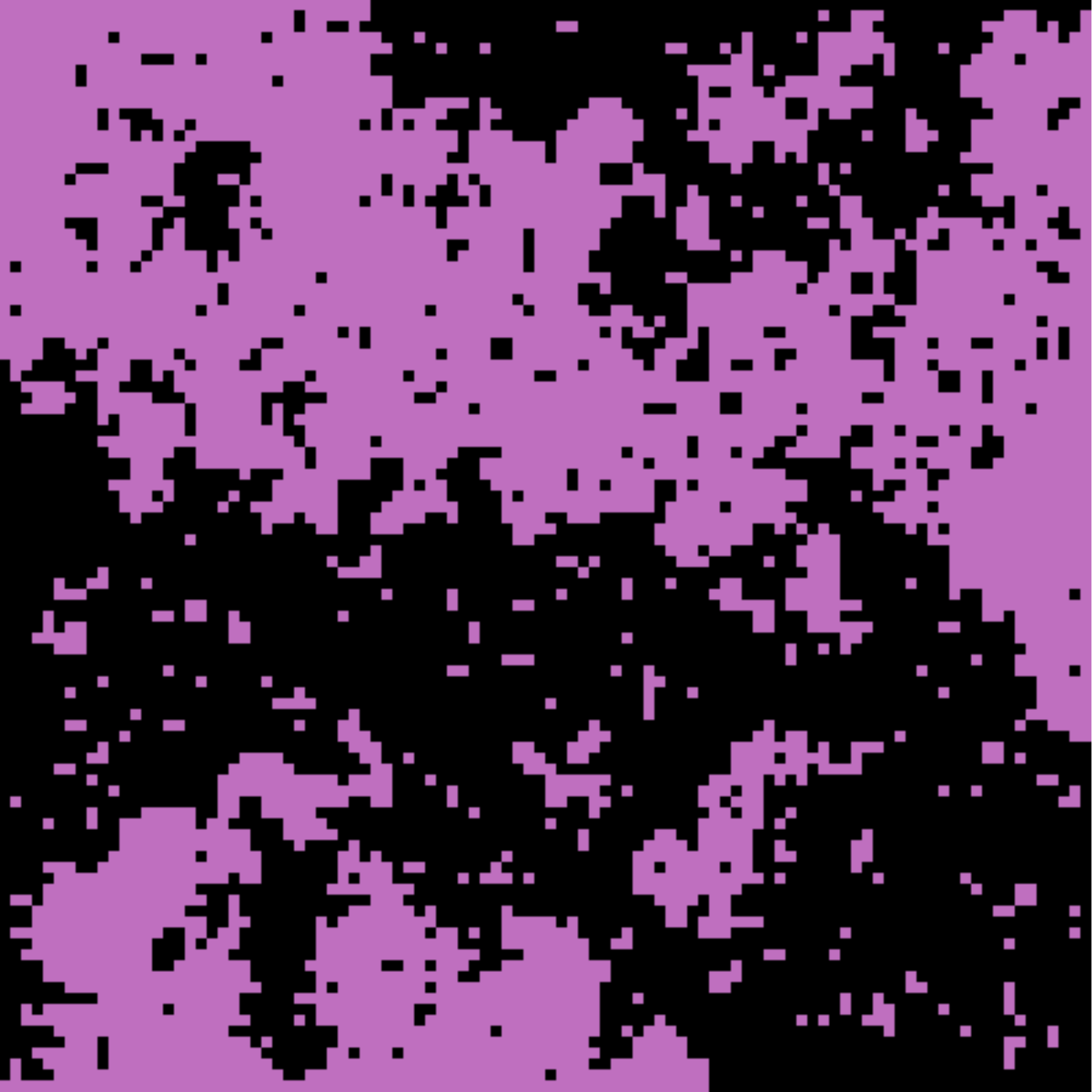} $\quad$
\includegraphics[width=2.5cm]{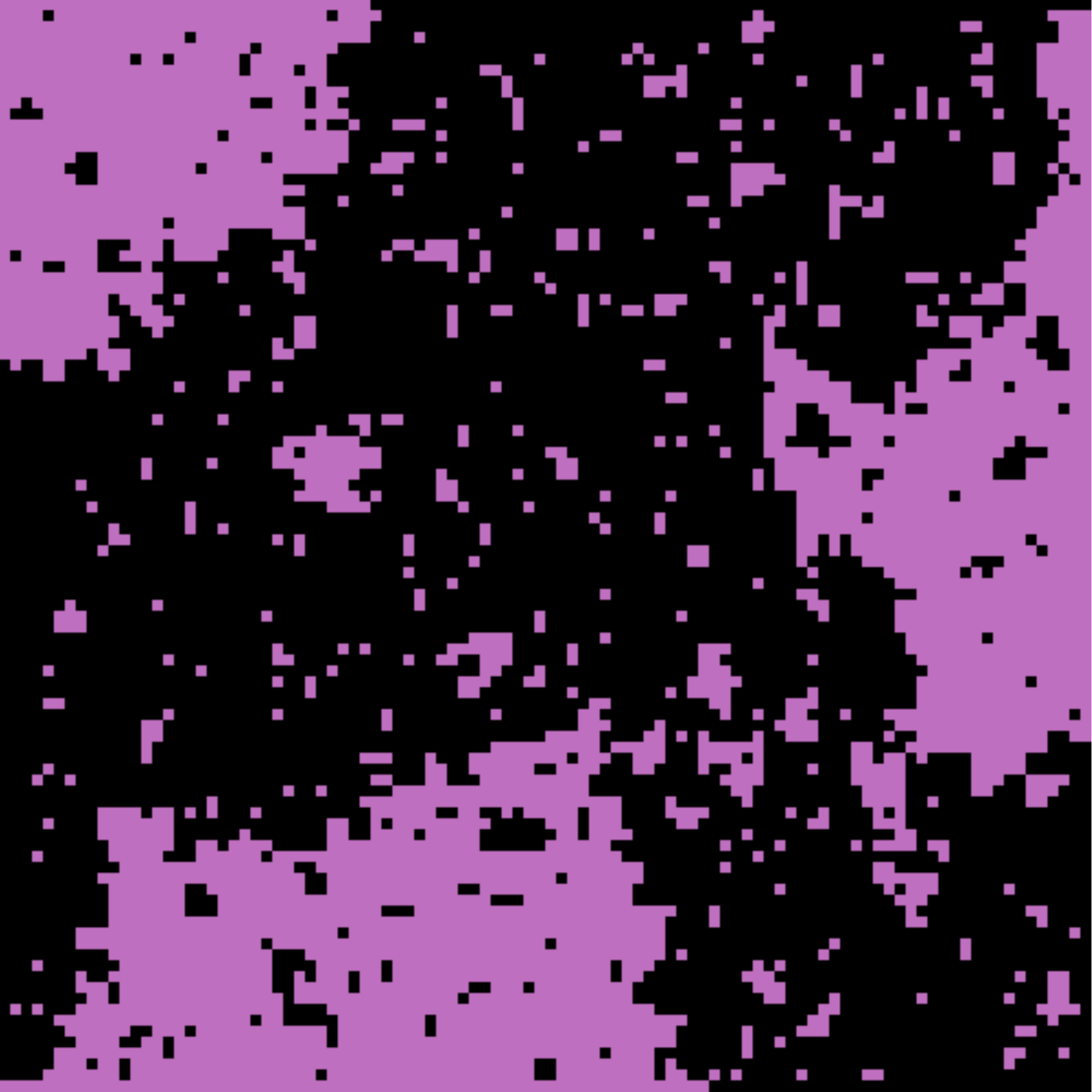} $\quad$
\includegraphics[width=2.5cm]{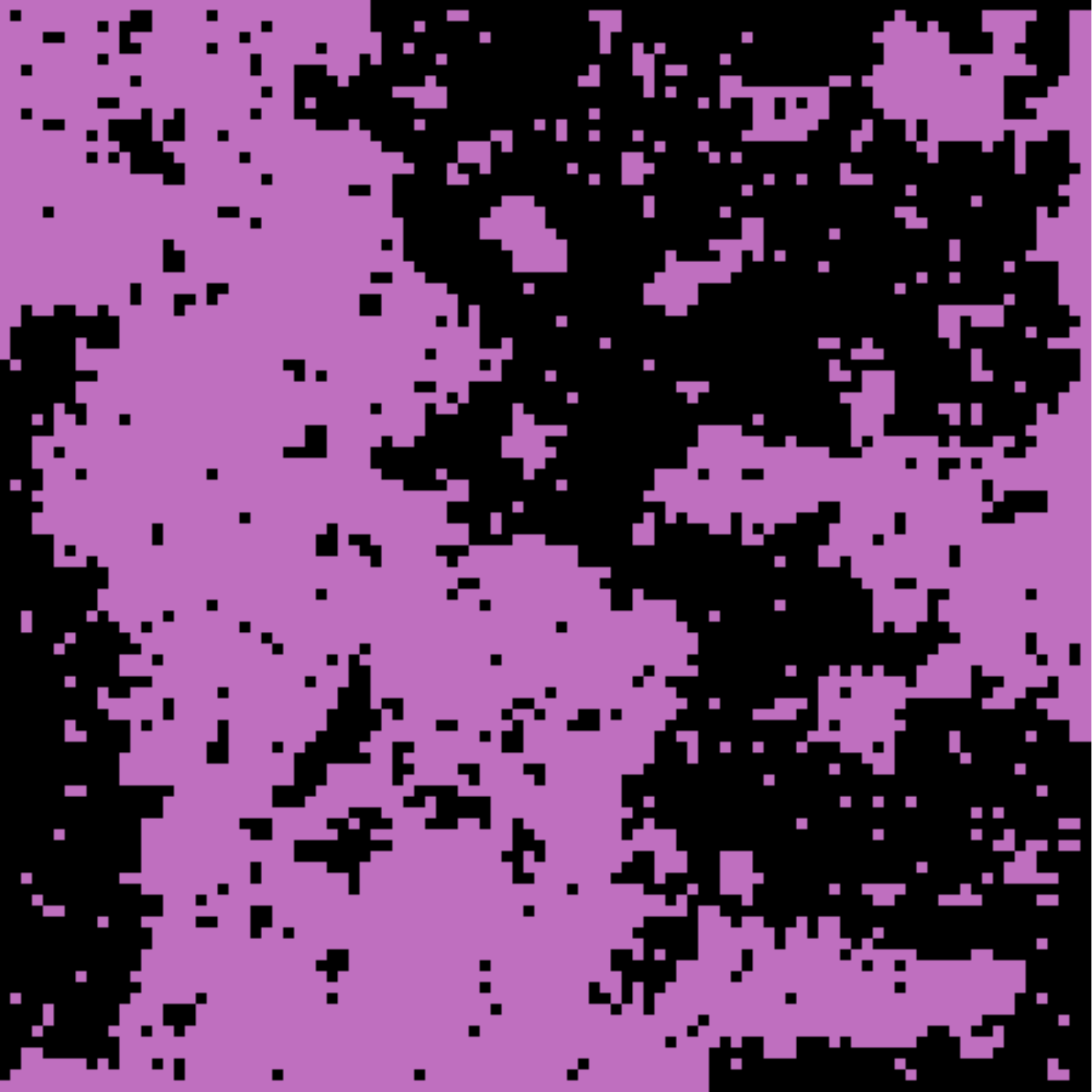} $\quad$
\includegraphics[width=2.5cm]{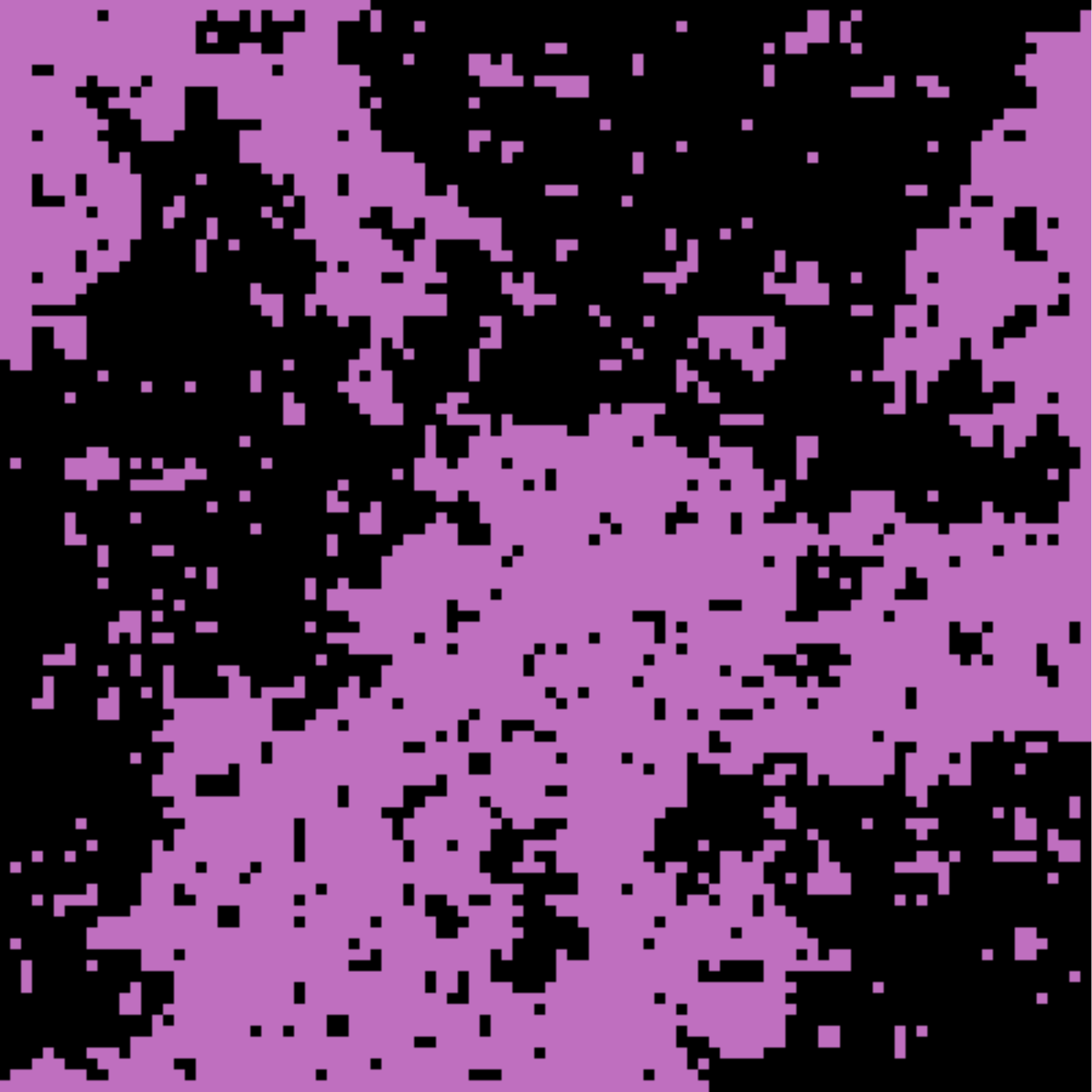} $\quad$
\includegraphics[width=2.5cm]{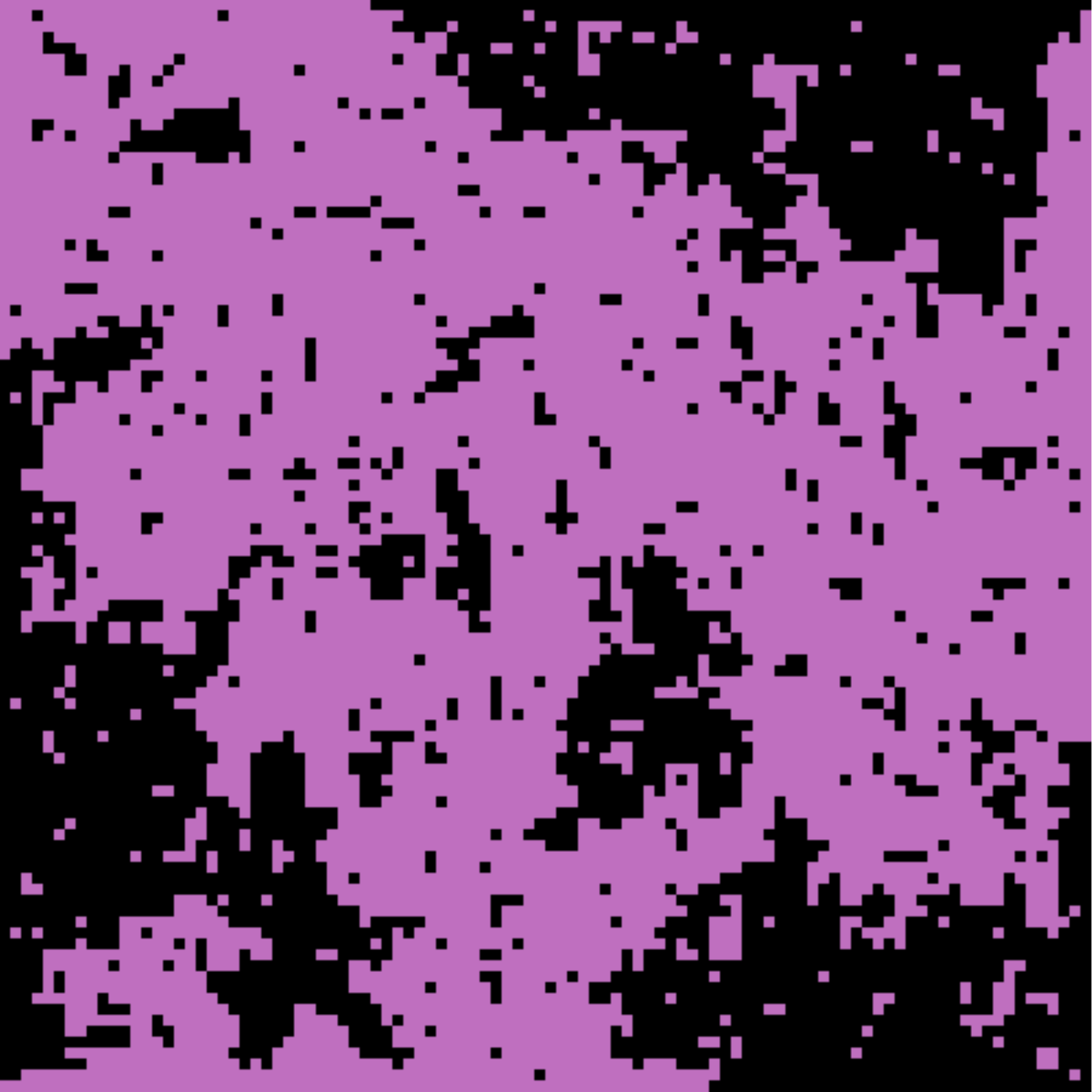} 
\\ \vspace{.5cm}
\includegraphics[width=2.5cm]{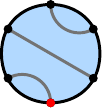} $\quad$
\includegraphics[width=2.5cm]{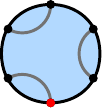} $\quad$
\includegraphics[width=2.5cm]{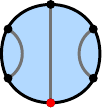} $\quad$
\includegraphics[width=2.5cm]{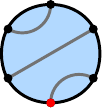} $\quad$
\includegraphics[width=2.5cm]{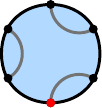} 
\caption{\label{fig: Ising connectivities N equals 3}
In statistical mechanics
models with alternating boundary conditions on $2N$ boundary segments,
interfaces form one of the $\Catalan_N$ possible planar connectivities.
The figure exemplifies the $N=3$ case with critical Ising model simulations
in a $100\times100$ square. Such critical Ising interfaces tend to multiple
$\SLEk$ with $\kappa=3$ in the scaling limit
\cite{Izyurov-PhD_thesis,Izyurov-critical_Ising_interfaces_in_multiply_connected_domains}.}
\end{figure}

This article pertains to an explicit
description and construction of the pure geometries of multiple $\SLE$s.
We follow the approach of
\cite{Dubedat-commutation, BBK-multiple_SLEs, Graham-multiple_SLEs, Dubedat-Euler_integrals},
in which local multiple $\SLE$s are constructed by growth processes of the curves starting from
the marked points $\bdrypt_1, \ldots, \bdrypt_{2N} \in \bdry \domain$.
The probabilistic details are postponed to Appendix~\ref{app: Local multiple SLEs},
where we give the precise definition of local multiple $\SLE$s, the ingredients of
their construction and classification, as well as relevant properties.
Crucially,
the construction of a local $N$-$\SLEk$ uses a
partition function, a function $\PartF$ defined on the chamber
\begin{align}
\chamber_{2N}=\; & \set{(x_{1},\ldots,x_{2N})\in\bR^{2N}\;\Big|\; x_{1}<\cdots<x_{2N}}.
\label{eq: chamber}
\end{align}
The partition function $\PartF$ is subject to a number of requirements.
First of all, it should be positive, $\PartF (x_1, \ldots, x_{2N}) > 0$,
since it is used in expressing the Radon-Nikodym derivatives of initial segments
of the local $N$-$\SLEk$ curves with respect to ordinary chordal $\SLE$s.
It has to satisfy a system of $2N$ linear partial differential equations (PDE)
of second order,
\begin{align}
 & \left[\frac{\kappa}{2}\pdder{x_{i}}+\sum_{j\neq i}\left(\frac{2}{x_{j}-x_{i}}\pder{x_{j}}-\frac{2h}{(x_{j}-x_{i})^{2}}\right)\right]\PartF(x_{1},\ldots,x_{2N})=0\qquad\text{for all }i=1,\ldots,2N,
\label{eq: multiple SLE PDEs}
\end{align}
where $h = \frac{6-\kappa}{2\kappa}$, to
guarantee a stochastic reparametrization invariance
\cite{Dubedat-commutation}.\footnote{These PDEs also arise in conformal field theory --- see, e.g.,
\cite{BBK-multiple_SLEs, FK-solution_space_for_a_system_of_null_state_PDEs_1}. The conformal weight
$h = h_{1,2}$ appears in the Kac table, and the PDEs are the null-field equations associated with the
degeneracy at level two of the boundary changing operators at the $2N$ marked boundary points.}
It has to transform covariantly (COV) under M\"obius transformations $\Mob(z) = \frac{az+b}{cz+d}$,
\begin{align}
\PartF(x_{1},\ldots,x_{2N})=\; & \prod_{i=1}^{2N}\Mob'(x_{i})^{h}\times\PartF(\Mob(x_{1}),\ldots,\Mob(x_{2N})) ,
\label{eq: multiple SLE Mobius covariance}
\end{align}
in order for the constructed process to be invariant under conformal self-maps
of the simply connected domain $\domain$.
The results of Appendix~\ref{app: Local multiple SLEs} state that local multiple $\SLE$s
are classified by such partition functions $\PartF$, and the convex structure of the
multiple $\SLE$s corresponds to the convex structure of the set of such functions.
\begin{thm*}[Theorem~\ref{thm: local multiple SLEs}] \
\begin{itemize}
\item Any partition function $\PartF$
can be used to construct a local multiple $\SLEk$,
and two  
functions $\PartF,\tilde{\PartF}$ give rise to the same local 
multiple $\SLEk$ if and only if they 
are constant multiples of each other.

\medskip

\item Any local multiple $\SLEk$ can be constructed from some partition function $\PartF$,
which is unique up to a multiplicative constant.

\medskip

\item For any $\,0\leq r\leq 1\,$, if
$\PartF=r\,\PartF_1+(1-r)\,\PartF_2$
is a convex combination of two
partition functions $\PartF_1$ and $\PartF_2$, 
then the local multiple $\SLEk$ probability measures 
associated to $\PartF$ are convex combinations of the probability measures
associated to $\PartF_1$ and $\PartF_2$,
with coefficients proportional to $r$ and $1-r$, and proportionality constants
depending on the conformal moduli of the domain with the marked points.
\end{itemize}
\end{thm*}

\subsection*{Multiple $\SLE$ pure partition functions}

The above results form a general classification of local multiple $\SLE$s by
the solution space of
the system \eqref{eq: multiple SLE PDEs}~--~\eqref{eq: multiple SLE Mobius covariance}.
This solution space can be shown to be finite dimensional \cite{Dubedat-Euler_integrals},
and indeed, the correct dimension $\Catalan_N$ 
has been established in the articles
\cite{FK-solution_space_for_a_system_of_null_state_PDEs_1,
FK-solution_space_for_a_system_of_null_state_PDEs_2,
FK-solution_space_for_a_system_of_null_state_PDEs_3}
(for solutions with at most polynomial
growth rate at diagonals and infinity).

The task is to construct multiple $\SLE$ pure geometries, for each
possible connectivity of the $N$ curves.
The sequence of marked boundary points $\bdrypt_1, \ldots, \bdrypt_{2N}$
is from now on assumed to appear in a positive (counterclockwise) order
along the boundary $\bdry \domain$.
We encode the connectivities of the non-crossing curves by
planar pair partitions $\alpha$ of the indices $1, \ldots, 2N$ of the
marked points $\bdrypt_1, \ldots, \bdrypt_{2N}$. As has become standard in the literature,
we refer to the pairs as links and
the pair partitions as link patterns.
A link formed by the pair $\set{a,b}$ of indices will be denoted by
$\link{a}{b}$.
A link pattern will be denoted by
\begin{align*}
\alpha=\set{\link{a_1}{b_1},\ldots,\link{a_N}{b_N}}.
\end{align*}
The non-crossing condition, i.e., the planar property of the pair partition,
can be expressed as the requirement $(a_j-a_k)(b_j-b_k)(b_j-a_k)(a_j-b_k) > 0$
whenever $j \neq k$. The set of link patterns of $N$ links is
denoted by $\LP_N$, and we recall that the number of these is a Catalan number,
$\#\LP_N = \Catalan_N = \frac{1}{N+1} \binom{2N}{N}$.
By convention, we include
the empty link pattern $\emptyset \in \LP_0$ in the case $N=0$.
The set of link patterns of any possible size
is denoted by $\LP = \bigsqcup_{N \geq 0} \LP_N$, and
for $\alpha\in\LP_N$, we denote $|\alpha|=N.$

We seek pure geometries of multiple $\SLE$s corresponding
to each link pattern, so in view of the
theorem above, the task is to construct their
corresponding partition functions $(\PartF_\alpha)_{\alpha \in \LP}$.
Each $\PartF_\alpha$ must solve the system
\eqref{eq: multiple SLE PDEs}~--~\eqref{eq: multiple SLE Mobius covariance}.
This system, which is the same for all link patterns $\alpha$ of the same
number of links, 
is supplemented by boundary conditions which depend on $\alpha$.

The boundary conditions are stated in terms of removing links from the
link pattern. After removal of one link from 
$\alpha \in \LP_N$, the indices 
must be relabeled by $1,\ldots,2N-2$.
When $\link{j}{j+1} \in \alpha$,
we denote by $\alpha \removeLink \link{j}{j+1}$ the link pattern of
$N-1$ links where the link $\link{j}{j+1}$ is removed and indices greater
than $j+1$ are reduced by two, see Figure~\ref{fig: removing link}.

\begin{figure}
\includegraphics[scale=1.5]{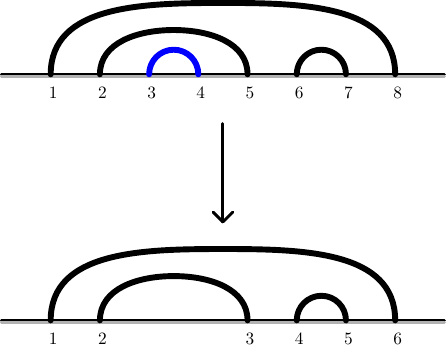}
\caption{\label{fig: removing link}
When removing the link $\link{3}{4}$ 
from the link pattern 
$\alpha=\set{\link{1}{8},\link{2}{5},\link{3}{4},\link{6}{7}}\in\LP_4$,
we obtain 
$\alpha\removeLink\link{3}{4}=\set{\link{1}{6},
\link{2}{3},\link{4}{5}}\in\LP_3$ after relabeling the remaining endpoints.}
\end{figure}

The boundary conditions are recursive in the number of links:
the conditions for $\PartF_\alpha$, $\alpha \in \LP_N$,
depend on the solutions $\PartF_{\hat{\alpha}}$, $\hat{\alpha} \in \LP_{N-1}$.
Specifically, we require of $\PartF_\alpha$ that
for all $j = 1,\ldots,2N-1$ and any $\xi \in (x_{j-1},x_{j+2})$, we have
\begin{align}
\lim_{x_{j},x_{j+1}\to\xi}\frac{\PartF_{\alpha}(x_{1},\ldots,x_{2N})}{(x_{j+1}-x_{j})^{-2h}}=\; & \begin{cases}
0\quad &
    \text{if } \link{j}{j+1} \notin \alpha\\
\PartF_{\alpha\removeLink\link{j}{j+1}}(x_{1},\ldots,x_{j-1},x_{j+2},\ldots,x_{2N}) &
    \text{if } \link{j}{j+1} \in \alpha .
\end{cases}  \label{eq: multiple SLE asymptotics}
\end{align}
These conditions were proposed in \cite{BBK-multiple_SLEs},
where pure geometries were introduced, and they were conjectured to be
sufficient to determine a solution to the system
\eqref{eq: multiple SLE PDEs}~--~\eqref{eq: multiple SLE Mobius covariance}.
The probabilistic meaning of these conditions is clarified in
Appendix~\ref{app: Local multiple SLEs}, Propositions~\ref{prop: local abs continuity wrt Bessel processes}
and \ref{prop: cascade relation}.
Solutions to the system
\eqref{eq: multiple SLE PDEs}~--~\eqref{eq: multiple SLE asymptotics}
will be called multiple $\SLE$ pure partition functions.
The collection $(\PartF_\alpha)_{\alpha \in \Arch}$ solving these conditions turns out to
be unique up to one multiplicative factor, and although positivity is neither required
nor proved here, we conjecture that the factor can be chosen so that all
$\PartF_\alpha$ are simultaneously positive.

The main result of this article is the explicit construction of
these M\"obius covariant solutions to the system of PDEs with boundary conditions.
The statement concerns 
the physically relevant parameter range $\kappa \leq 8$, and generic values $\kappa \notin \bQ$,
which avoids certain algebraic and analytic degeneracies.
\begin{thm*}[Theorem~\ref{thm: pure partition functions}]
For $\kappa \in (0,8) \setminus \bQ$, there exists a collection $\left(\PartF_{\alpha}\right)_{\alpha\in\LP}$ 
of functions 
\begin{align*}
\PartF_{\alpha}\;\colon\;\chamber_{2|\alpha|}\rightarrow\bC
\end{align*}
such that the system of 
equations~\eqref{eq: multiple SLE PDEs}~--~\eqref{eq: multiple SLE asymptotics} 
holds for all $\alpha\in\LP$.
%
For any $N\in\bZnn$, the collection 
$\left(\PartF_{\alpha}\right)_{\alpha\in\LP_N}$ is linearly independent 
and it spans a $\Catalan_N$-dimensional space of solutions to
\eqref{eq: multiple SLE PDEs}~--~\eqref{eq: multiple SLE Mobius covariance}.
\end{thm*}

Our solution is based on a systematic quantum group technique developed in
\cite{KP-conformally_covariant_boundary_correlation_functions_with_a_quantum_group}.
By the ``spin chain~--~Coulomb gas correspondence'' of that article,
we translate the problem of solving the system
\eqref{eq: multiple SLE PDEs}~--~\eqref{eq: multiple SLE asymptotics} 
to a linear problem in a $2^N$-dimensional representation of a quantum group,
the $q$-deformation of $\mathfrak{sl}_2$. With this translation,
we first exhibit an explicit formula for all maximally nested (rainbow) link
patterns, and then obtain the solutions for other link patterns
using a recursion on the partially ordered set of link patterns
with a given number $N$ of links.

\subsection*{Solutions to $\SLE$ boundary visit amplitudes}

Using our solution to the problem of multiple $\SLE$ pure partition functions,
we also derive the existence of solutions to another question to which 
the ``spin chain~--~Coulomb gas correspondence'' is applied --- the chordal $\SLE$
boundary visit amplitudes, treated in \cite{JJK-SLE_boundary_visits}.
These amplitudes are functions indexed by possible orders of visits,
$\Orders = \bigsqcup_{N' \in \bN} \Orders_{N'}$,
where $\Orders_{N'} = \set{-,+}^{N'}$ is the set of possible orders of
visits to $N'$ boundary points. Like multiple $\SLE$ pure partition
functions, the boundary visit amplitudes
$\left(\Ampl_{\omega}\right)_{\omega\in\Orders}$
are required to satisfy partial differential equations, covariance
(translation invariance and homogeneity), and boundary conditions
(prescribed asymptotic behaviors).
We postpone the precise statement of the problem to
Section~\ref{sec: SLE boundary visits}.
Uniqueness of solutions in the class of functions obtained by the
``spin chain~--~Coulomb gas correspondence'' is relatively easy,
and existence was shown for $N' \leq 4$ in
\cite{JJK-SLE_boundary_visits} by
direct calculations. We prove the existence for all $N'$.
\begin{thm*}[Theorem~\ref{thm: bdry visit amplitudes}]
For $\kappa \in (0,8) \setminus \bQ$, there exists a collection
$\left(\Ampl_{\omega}\right)_{\omega\in\Orders}$ of functions 
such that the system of 
partial differential equations, covariance, 
and boundary conditions required in \cite{JJK-SLE_boundary_visits}
holds for all $\omega\in\Orders$.
\end{thm*}

\subsection*{Relation to other work}

We follow the approach of constructing multiple $\SLE$s by
a growth process, where a partition function $\PartF$ is
needed as an input.
An alternative approach is the so called configurational measure of
multiple $\SLE$s, where the Radon-Nikodym derivative of the law of
the full global configuration of curves (w.r.t. $N$ chordal $\SLE$s)
is written in terms of Brownian loop measure.
For $\kappa \leq 4$, the configurational measures for the maximally nested (rainbow)
link patterns were constructed in \cite{LK-configurational_measure}.
However, constructions of configurational measures with $4 < \kappa \leq 8$ have not been given. 
The most obvious advantage of the configurational measure is the direct treatment of the
global configuration. In contrast, the growth process construction straightforwardly only
allows to define localized versions of multiple $\SLE$s 
(see Appendix~\ref{app: Local multiple SLEs}),
and the extension to globally defined random curves poses challenges
similar to the reversibility of chordal $\SLE$ \cite{Zhan-reversibility, MS-imaginary_geometry_3}.  
The partition function approach, however, is somewhat more explicit and better suited 
for example for sampling multiple $\SLE$s.
The two approaches should of course produce the same results, and in a
future work we plan to show this.
The total masses of the unnormalized configurational measures
will, in particular, be explicitly expressible in terms of our partition functions.

The pure partition functions are intimately related to crossing probabilities in
models of statistical physics \cite{BBK-multiple_SLEs}, and they have also been called 
``connectivity weights'' in the literature \cite{FK-solution_space_for_a_system_of_null_state_PDEs_4,
FSK-multiple_SLE_connectivity_weights_for_4_6_8}.
Note that for $N=2$, M\"obius covariance \eqref{eq: multiple SLE Mobius covariance} allows to reduce
the partial differential equations \eqref{eq: multiple SLE PDEs} to ordinary differential
equations, and pure partition functions for $2$-SLE are then straightforwardly solvable
in terms of hypergeometric functions \cite{BBK-multiple_SLEs}. They generalize
Cardy's formula \cite{Cardy-critical_percolation_in_finite_geometries}
for crossing probabilities for critical percolation, to which they reduce at $\kappa=6$.
For percolation and $N=3$, a recent numerical study \cite{FZS-percolation_crossing_probabilities_in_hexagons}
confirms the validity of crossing probability formulas found 
in \cite{Simmons-logarithmic_operator_intervals_in_the_boundary_theory_of_critical_percolation}.
Very recently, explicit formulas for connectivity weights for cases $N=3$ and $N=4$
and a formula for rainbow connectivity weights for general $N$ were worked out in
\cite{FSK-multiple_SLE_connectivity_weights_for_4_6_8}.
Virtually all work for $N>2$ relies on some form of Coulomb gas integral solutions
\cite{DF-multipoint_correlation_functions}, 
which also underlie the correspondence
\cite{KP-conformally_covariant_boundary_correlation_functions_with_a_quantum_group}
that is crucially employed in the present work.

The solutions to the system
\eqref{eq: multiple SLE PDEs}~--~\eqref{eq: multiple SLE Mobius covariance} for all $N$
have also been studied in the series of articles
\cite{FK-solution_space_for_a_system_of_null_state_PDEs_1,
FK-solution_space_for_a_system_of_null_state_PDEs_2,
FK-solution_space_for_a_system_of_null_state_PDEs_3,
FK-solution_space_for_a_system_of_null_state_PDEs_4}.
The main results are closely parallel to ours, although the works have been
independent. 
Our integrals treat all marked points symmetrically
at the expense of having dimension one higher than those of Flores and Kleban.
Importantly, by different analytic methods, Flores and Kleban prove also the upper
bound for the dimension of the solution space, and by a limiting procedure, they
obtain results about rational $\kappa$.

We include a probabilistic justification for the system
\eqref{eq: multiple SLE PDEs}~--~\eqref{eq: multiple SLE Mobius covariance}
and boundary conditions \eqref{eq: multiple SLE asymptotics},
in the form of a classification and properties of local multiple $\SLE$s, in
Appendix~\ref{app: Local multiple SLEs}.
Our method of solving this system has the advantage
of being completely systematic, in translating the problem
to algebraic calculations in finite dimensional representations of a quantum group.
The translation, based on \cite{KP-conformally_covariant_boundary_correlation_functions_with_a_quantum_group},
allows to deduce properties of the solutions conceptually, using representation theory.
This systematic formalism will also be valuable
in further applications, which require establishing for example bounds
or monodromy properties of the solutions.


In \cite{JJK-SLE_boundary_visits}, the quantum group method
of \cite{KP-conformally_covariant_boundary_correlation_functions_with_a_quantum_group}
was applied to formulas for the probability that a chordal $\SLE$ curve passes
through small neighborhoods of given points on the boundary.
Solutions 
were given for cases having up to four visited points.
In Section~\ref{sec: SLE boundary visits}, we present a solution for arbitrary number
of visited points,
making use of natural representation theoretic mappings to reduce
the problem to the main results of this article.

Related questions of random connectivities in planar non-crossing
ways appear also in various other interesting contexts. 
Notably, the famous Razumov-Stroganov conjecture has a reformulation in terms of probabilities of the
different planar connectivities alternatively in a percolation model
on a semi-infinite lattice cylinder, or in a fully packed loop model in a lattice square
\cite{RS-combinatorial_nature_of_ground_state_vector_of_O1_loop_model,CS-proof_of_RS_conjecture}.
As another example, the boundary conditions
that enforce the existence of $N$ interfaces connecting $2N$ boundary points
can be studied in models of statistical mechanics on random lattices, i.e.,
in discretized quantum gravity. 
Partition functions of the various
connectivities for the Ising model on random lattices have been found by matrix model
techniques in \cite{EO-mixed_correlation_functions_in_2_matrix_model}.
In both of the above mentioned problems, some relations to the present work can be
anticipated, since conformal field theory is expected to underlie each of the models.
Unveiling precise connections to these, however, is left for future research.

\subsection*{Organization of the article}

In Section~\ref{sec: quantum group method}, we give the definition of 
the quantum group $\Uqsltwo$, summarize needed facts about its representation theory, and
present a special case of the ``spin chain~--~Coulomb gas correspondence'', 
which is our main tool in the construction of multiple $\SLE$ partition functions.
Section~\ref{sec: Multiple SLE pure geometries} contains the solution of
a quantum group reformulation of the problem of multiple $\SLE$ pure partition functions. 
In Section~\ref{sec: Multiple SLE partition functions},
with the help of the correspondence of Section~\ref{sec: quantum group method}, 
this solution is translated to the construction of
the pure partition functions. 
We also discuss symmetric partition functions relevant for
models invariant under cyclic permutations of the marked points, and give
examples of such symmetric partition functions for the Ising model, 
Gaussian free field, and percolation.

Section~\ref{sec: SLE boundary visits} is devoted to the proof
of the existence of solutions to the 
$\SLE$ boundary visit problem. 
Finally, in Appendix~\ref{app: Local multiple SLEs}, we include
an account of the probabilistic aspects: the
classification and construction of local multiple $\SLE$s by
partition functions, and the role of the
requirements~\eqref{eq: multiple SLE PDEs}~--~\eqref{eq: multiple SLE asymptotics}.

\subsection*{Acknowledgments}
K.~K. is supported by the Academy of Finland, and 
E.~P. by Vilho, Yrj\"o and Kalle V\"ais\"al\"a Foundation.
During this work, we have benefited from interesting discussions with
Steven Flores, Cl\'ement Hongler, Yacine Ikhlef, Kostya Izyurov, Peter Kleban,
Michael Kozdron, Greg Lawler, and Jake Simmons.

\bigskip

\section{\label{sec: quantum group method}The quantum group method}

In this section, we present the quantum group method in the form
it will be used for the solution of the problem
\eqref{eq: multiple SLE PDEs}~--~\eqref{eq: multiple SLE asymptotics}.
The method was developed more generally in
\cite{KP-conformally_covariant_boundary_correlation_functions_with_a_quantum_group}.

The relevant quantum group is a $q$-deformation $\Uqsltwo$ of the
Lie algebra $\sl_{2}(\bC)$, 
and the deformation parameter $q$ is
related to $\kappa$ by $q = e^{\ii 4 \pi / \kappa}$.
We assume that $\kappa \in (0,8) \setminus \bQ$, so that $q$ is not
a root of unity. The method associates functions of $n$ variables
to vectors in a tensor product of $n$ irreducible representations of
this quantum group.

\subsection{\label{subsec: quantum group}The quantum group and its representations}

Define, for $m\in\bZ$, $n\in\bZpos$, the
$q$-integers as
\begin{align}
\qnum{m} = \; & \frac{q^{m}-q^{-m}}{q-q^{-1}} \label{eq: q num basic def}
    =  q^{m-1}+q^{m-3}+\cdots+q^{3-m}+q^{1-m} 
\end{align}
and the $q$-factorials as $\qfact{n} = \qnum{n} \qnum{n-1} \cdots \qnum{2} \qnum{1}$. 

The quantum group $\Uqsltwo$ is the associative unital algebra
over $\bC$ generated by $E,F,K,K^{-1}$ subject to the relations
\begin{align*}
 KK^{-1}=&\; 1=K^{-1}K,\qquad KE=q^{2}EK,\qquad KF=q^{-2}FK,\\
 EF-FE=&\; \frac{1}{q-q^{-1}}\left(K-K^{-1}\right).\nonumber 
\end{align*}
There is a unique Hopf algebra structure on $\Uqsltwo$ with the coproduct,
an algebra homomorphism
\begin{align*}
\Delta\colon\; & \Uqsltwo\rightarrow\Uqsltwo\tens\Uqsltwo,
\end{align*}
given on the generators by the expressions
\begin{align}\label{eq: coproduct}
\Hcp(E)=\;E\tens K+1\tens E,\qquad\Hcp(K)=\;K\tens K,\qquad\Hcp(F)=\;F\tens1+K^{-1}\tens F.
\end{align}
The coproduct is used to define a representation structure on
the tensor product $\mathsf{M}'\tens \mathsf{M}''$
of two representations $\mathsf{M}'$ and $\mathsf{M}''$ as follows.
When we have
\begin{align*}
\Hcp(X) = \; & \sum_{i}X_{i}'\tens X_{i}''\,\in\,\Uqsltwo\tens\Uqsltwo
\end{align*}
and $v'\in \mathsf{M}'$, $v''\in \mathsf{M}''$, we set
\begin{align*}
X.(v'\tens v'') = \; & \sum_{i}(X_{i}'.v')\tens(X_{i}''.v'')\,\in\, \mathsf{M}'\tens \mathsf{M}''.
\end{align*}
For calculations with tensor products of $n$ representations, one similarly uses
the $(n-1)$-fold coproduct 
\begin{align*}
\Hcp^{(n)}=\; & (\Hcp\tens\id^{\tens(n-2)})\circ(\Hcp\tens\id^{\tens(n-3)})\circ\cdots
\circ(\Hcp\tens\id)\circ\Hcp ,
\qquad \Hcp^{(n)} \colon\Uqsltwo\rightarrow\Big(\Uqsltwo\Big)^{\tens n} .
\end{align*}
The coassociativity $(\id\tens\Hcp)\circ\Hcp=(\Hcp\tens\id)\circ\Hcp$ of the coproduct ensures that
we may speak of multiple tensor products without specifying the positions of
parentheses.

We will use representations $\Wd_d$, $d \in \bZpos$, which can be thought of as $q$-deformations
of the $d$-dimensional irreducible representations of the semisimple Lie
algebra $\sl_{2}(\bC)$.
In our chosen basis $\Wbas_{0}^{(d)},\Wbas_{1}^{(d)},\ldots,\Wbas_{d-1}^{(d)}$ of $\Wd_d$,
the actions of the generators are
\begin{align*}
K.\Wbas_{l}^{(d)}=\; & q^{d-1-2l}\,\Wbas_{l}^{(d)}\\
F.\Wbas_{l}^{(d)}=\; & \begin{cases}
\Wbas_{l+1}^{(d)} & \text{if }l\neq d-1\\
0 & \text{if }l=d-1
\end{cases}\\
E.\Wbas_{l}^{(d)}=\; & \begin{cases}
\qnum l\qnum{d-l}\,\Wbas_{l-1}^{(d)} & \text{if }l\neq0\\
0 & \text{if }l=0
\end{cases}.
\end{align*}
The representation $\Wd_d$ thus defined is irreducible, see e.g.
\cite[Lemma~2.3]{KP-conformally_covariant_boundary_correlation_functions_with_a_quantum_group}.
For simplicity of notation, we often omit the superscript reference to the dimension $d$,
and denote the basis vectors simply by $\Wbas_0, \ldots, \Wbas_{d-1}$.

Tensor products of these representations decompose to direct sums of irreducibles according to
$q$-deformed Clebsch-Gordan formulas, stated for later use in the lemma below.
\begin{lem}[{see e.g. \cite[Lemma~2.4]{KP-conformally_covariant_boundary_correlation_functions_with_a_quantum_group}}]
\label{lem: tensor product representations of quantum sl2}
Let $d_1, d_2 \in \bZpos$, and 
consider the representation $\Wd_{d_{2}}\tens\Wd_{d_{1}}$.
Let
$m\in \set{0,1,\ldots,\min(d_{1},d_{2})-1}$,
and denote $d=d_{1}+d_{2}-1-2m$.
The vector
\begin{align} 
\Tbas_{0}^{(d;d_{1},d_{2})} = \; & \sum_{l_{1},l_{2}}T_{0;m}^{l_{1},l_{2}}(d_{1},d_{2})\times(\Wbas_{l_{2}}\tens\Wbas_{l_{1}})\label{eq: tensor product hwv}, \\
\nonumber
\text{where } \qquad
T_{0;m}^{l_{1},l_{2}}(d_{1},d_{2})=\; &
    \delta_{l_{1}+l_{2},m}\times(-1)^{l_{1}}\frac{\qfact{d_{1}-1-l_{1}}\,\qfact{d_{2}-1-l_{2}}}{\qfact{l_{1}}\qfact{d_{1}-1}\qfact{l_{2}}\qfact{d_{2}-1}}\,\frac{q^{l_{1}(d_{1}-l_{1})}}{(q-q^{-1})^{m}} , 
\end{align}
satisfies $E.\Tbas_{0}^{(d;d_{1},d_{2})}=0$ and $K.\Tbas_{0}^{(d;d_{1},d_{2})} = q^{d-1} \, \Tbas_{0}^{(d;d_{1},d_{2})}$, i.e.,
$\Tbas_{0}^{(d;d_{1},d_{2})}$ is a highest weight vector of a subrepresentation of $\Wd_{d_{2}}\tens\Wd_{d_{1}}$
isomorphic to $\Wd_{d}$. 
The subrepresentations corresponding to different $d$ span the tensor product $\Wd_{d_2} \tens \Wd_{d_1}$,
which thus has a decomposition
\begin{align}\label{eq: decomposition of tensor product}
\Wd_{d_{2}}\tens\Wd_{d_{1}}\isom\; & \Wd_{d_{1}+d_{2}-1}\oplus\Wd_{d_{1}+d_{2}-3}\oplus\cdots\oplus\Wd_{|d_{1}-d_{2}|+3}\oplus\Wd_{|d_{1}-d_{2}|+1} .
\end{align}
\end{lem}
For the subrepresentation $\Wd_{d} \subset \Wd_{d_{2}}\tens\Wd_{d_{1}}$,
we often use the basis vectors $\Tbas_{l}^{(d;d_{1},d_{2})} = F^l . \Tbas_{0}^{(d;d_{1},d_{2})}$.

\subsection{Tensor products of two-dimensional irreducibles}

For the solution of multiple $\SLE$ pure partition functions, 
we use in particular the two-dimensional representation $\Wd_2$ 
and its tensor powers $\Wd_2^{\tens 2 N}$.
A special case of Lemma~\ref{lem: tensor product representations of quantum sl2} 
states that
\begin{align*}
\Wd_{2} \tens \Wd_{2} \isom \Wd_{3}\oplus\Wd_{1} .
\end{align*}
For this tensor product, we select the following basis 
that respects the decomposition:
the singlet subspace $\Wd_1 \subset \Wd_{2}\tens\Wd_{2}$ is spanned by
\begin{align}\label{eq: singlet basis vector}
\Sbas := \Tbas_{0}^{(1;2,2)}=\frac{1}{q-q^{-1}}\left(\Wbas_{1}\tens\Wbas_{0}-q\,\Wbas_{0}\tens\Wbas_{1}\right)
\end{align}
and the triplet subspace $\Wd_3 \subset \Wd_{2}\tens\Wd_{2}$ by
\begin{align}\label{eq: triplet basis vectors}
\Tbas_{0}^{(3;2,2)}=\Wbas_{0}\tens\Wbas_{0},\qquad
\Tbas_{1}^{(3;2,2)}=q^{-1}\,\Wbas_{0}\tens\Wbas_{1}+\Wbas_{1}\tens\Wbas_{0},\qquad
\Tbas_{2}^{(3;2,2)}=[2]\,\Wbas_{1}\tens\Wbas_{1} .
\end{align}

More generally, the $n$-fold tensor product $\Wd_2^{\tens n}$ decomposes 
to a direct sum of irreducibles as follows.
\begin{lem}\label{lem: multiplicity of singlet}
We have, for $n\in\bZnn$, a decomposition
\begin{align*}
\Wd_2^{\otimes n}\isom\bigoplus_{d}m_d^{(n)}\Wd_{d},\qquad\text{where}\qquad
m_d^{(n)}=\; & \begin{cases}\frac{2d}{n+d+1}\binom{n}{\frac{n+d-1}{2}}\quad & \text{if }n+d-1\in2\,\bZnn\text{ and }1\leq d\leq n+1\\
0 & \text{otherwise.}
\end{cases}
\end{align*}
\end{lem}
\begin{proof}
A standard proof proceeds by induction on $n$. Clearly the assertion is true for $n=0$, 
as $\Wd_2^{\otimes 0}\isom\bC\isom\Wd_1$. Assuming the decomposition formula
for $\Wd_2^{\otimes n}$ and using 
Equation~\eqref{eq: decomposition of tensor product} we obtain
\begin{align*}
\Wd_2^{\otimes (n+1)}\isom\Wd_2\otimes\left(\bigoplus_{d'}m_{d'}^{(n)}\Wd_{d'}\right)
=\bigoplus_{d'}m_{d'}^{(n)}\left(\Wd_{d'+1}\oplus\Wd_{d'-1}\right).
\end{align*}
There are $m_{d-1}^{(n)}+m_{d+1}^{(n)}$ subrepresentations contributing to $\Wd_d$,
giving the recursion $m_{d-1}^{(n)}+m_{d+1}^{(n)}=m_d^{(n+1)}$.
The solution to this recursion with the initial condition $m_d^{(0)} = \delta_{d,1}$
is the asserted formula.
\end{proof}
We will in particular use the sum of all one-dimensional
subrepresentations,
\begin{align}\label{eq: highest vector space}
\HWsp_1=\HWsp_1\left(\Wd_{2}^{\tens 2N}\right)=\set{v\in\Wd_{2}^{\tens 2N}\;\Big|\;E.v=0 ,\;K.v=v} .
\end{align}
For brevity, the dependence on $N$ is suppressed in the notation $\HWsp_1$.
From the lemma above we get that the dimension of $\HWsp_1$ is a Catalan number
\[ \dmn(\HWsp_1) = m_1^{(2N)}=\frac{1}{N+1}\binom{2N}{N}=\Catalan_N . \]

In the decomposition $\Wd_{2} \tens \Wd_{2} \isom \Wd_{3}\oplus\Wd_{1}$,
we denote the projection to the singlet subspace by
\begin{align*}
\pi\colon\; &\Wd_{2}\tens\Wd_{2}\to\Wd_{2}\tens\Wd_{2} , \qquad 
\pi(\Sbas)=\Sbas \qquad \text{ and } \qquad
\pi(\Tbas^{(3;2,2)}_l) = 0  \text{ for }  l = 0,1,2 . 
\end{align*}
More generally, in the tensor product $\Wd_2^{\tens n}$, with $n \geq 2$,
we denote by $\pi_j$
this projection acting in the components $j$ and $j+1$ counting from the right, i.e.,
\begin{align*}
\pi_j = \id^{\tens (n-1-j)} \tens \pi \tens \id^{\tens (j-1)} \colon\; &\Wd_{2}^{\tens n}\to\Wd_{2}^{\tens n}.
\end{align*}
The one-dimensional irreducible $\Wd_1$ is called the trivial representation ---
by Lemma~\ref{lem: tensor product representations of quantum sl2}, it is the neutral
element of the tensor product operation.
With the identification 
$\Wd_{1}\isom\bC$ via $\Sbas\mapsto1$,
we denote the projection to the trivial subrepresentation of $\Wd_2 \tens \Wd_2$ by
\begin{align*}
\hat{\pi}\colon\; &\Wd_{2}\tens\Wd_{2}\to \bC , \qquad
\hat{\pi}(\Sbas) = 1 \qquad \text{ and } \qquad
\hat{\pi}(\Tbas^{(3;2,2)}_l) = 0  \text{ for }  l = 0,1,2 , 
\end{align*}
and similarly
\begin{align*}
\hat{\pi}_j = \id^{\tens (n-1-j)} \tens \hat{\pi} \tens \id^{\tens (j-1)} \colon\; &\Wd_{2}^{\tens n}\to\Wd_{2}^{\tens (n-2)}.
\end{align*}

Lemmas~\ref{lem: projection formulas} and \ref{lem: all projections vanish},
and Corollary~\ref{cor: all projections vanish gives zero} below contain auxiliary results that will be used later on.

The formulas given in the next lemma are essentially a reformulation
of the fact that the projections to subrepresentations in tensor powers
of $\Wd_2$ form a Temperley-Lieb algebra.
\begin{lem}\label{lem: projection formulas}
The maps $\pi$ and $\hat{\pi}$
satisfy the following relations.
\begin{description}
\item[(a)] We have $\,\pi(v)=\hat{\pi}(v)\Sbas\,$ 
for any $v\in\Wd_{2}\tens\Wd_{2}$.
The values of $\hat{\pi}$ on the tensor product basis are
\begin{align*}
&\hat{\pi}(\Wbas_{0}\tens\Wbas_{0})=0,\qquad\qquad\quad\qquad\qquad\hat{\pi}(\Wbas_{1}\tens\Wbas_{1})=0,\\
&\hat{\pi}(\Wbas_{0}\tens\Wbas_{1})=\frac{q^{-1}-q}{\qnum 2},\qquad\qquad\qquad\hat{\pi}(\Wbas_{1}\tens\Wbas_{0})=\frac{1-q^{-2}}{\qnum 2}.
\end{align*}
\item[(b)] On the triple tensor product 
$\Wd_{2}\tens\Wd_{2}\tens\Wd_{2},$ we have
\begin{align*}
(\id_{\Wd_{2}}\tens\pi)\,(\Sbas\tens\Wbas_{l})=-\frac{1}{\qnum 2}\,\Wbas_{l}\tens\Sbas,\qquad\qquad(\pi\tens\id_{\Wd_{2}})\,(\Wbas_{l}\tens\Sbas)=-\frac{1}{\qnum 2}\,\Sbas\tens\Wbas_{l}
\end{align*}
and consequently,
\begin{align*}
(\id_{\Wd_{2}}\tens\hat{\pi})\,(\Sbas\tens\Wbas_{l})=-\frac{1}{ \qnum 2}\,\Wbas_{l},\qquad\qquad(\hat{\pi}\tens\id_{\Wd_{2}})\,(\Wbas_{l}\tens\Sbas)=-\frac{1}{\qnum 2}\,\Wbas_{l}.
\end{align*}
\end{description}
\end{lem}
\begin{proof}
Part (a) follows by straightforward calculations, using 
Equations~\eqref{eq: singlet basis vector} and \eqref{eq: triplet basis vectors}. For (b), one applies (a) and \eqref{eq: singlet basis vector}.
\end{proof}

The next lemma characterizes the highest dimensional subrepresentation
of a tensor power of $\Wd_2$.
\begin{lem}\label{lem: all projections vanish}
The following conditions are equivalent for any $v\in\Wd_{2}^{\tens n}$: 
\begin{description}
\item[(a)] $\hat{\pi}_j(v)=0\;$ for all $\;1\leq j<n.$
\item[(b)] $v\in\Wd_{n+1}\subset\Wd_{2}^{\tens n},\;$ where $\Wd_{n+1}$ is the irreducible subrepresentation generated by $\Wbas_{0}\tens\cdots\tens\Wbas_{0}.$
\end{description}
\end{lem}
\begin{proof}
It is clear by Lemma \ref{lem: projection formulas}(a) that the highest 
weight vector $\Wbas_{0}\tens\cdots\tens\Wbas_{0}$ satisfies (a) above. On the 
other hand, since the projections $\pi_j$ commute with the action of 
$\Uqsltwo$, (a) also holds for any other vector in $\Wd_{n+1}.$
Hence, (b) implies (a).

Suppose then that (a) is true. We may assume that $K.v = q^{n-2\ell}v$
for some $\ell$ and write
\begin{align*}
v=\;\sum_{\substack{l_1,\ldots,l_n\in\{0,1\}\\l_1+\cdots+l_n=\,\ell}}c_{l_1,\ldots,l_n}\times(\Wbas_{l_n}\tens\cdots\tens\Wbas_{l_1}).
\end{align*}
Using Lemma \ref{lem: projection formulas}(a), we calculate
\begin{align*}
0=\hat{\pi}_j(v)=\;\sum_{\substack{l_1,\ldots,l_{j-1},l_{j+2},\ldots,l_n\in\{0,1\}\\l_1+\cdots+l_{j-1}+l_{j+2}+\cdots+l_n=\,\ell-1}}\frac{1-q^{-2}}{\qnum 2}&\left(-q\,c_{l_1,\ldots,l_{j+2},1,0,l_{j-1},\ldots,l_n}+c_{l_1,\ldots,l_{j+2},0,1,l_{j-1},\ldots,l_n}\right)\\
&\times(\Wbas_{l_n}\tens\cdots\tens\Wbas_{l_{j+2}}\tens\Wbas_{l_{j-1}}\tens\cdots\tens\Wbas_{l_1})
\end{align*}
which gives
\begin{align*}
c_{l_1,\ldots,l_{j+2},0,1,l_{j-1},\ldots,l_n}=q\,c_{l_1,\ldots,l_{j+2},1,0,l_{j-1},\ldots,l_n}
\end{align*}
so fixing the value of $c_{1,\ldots,1,0,\ldots,0}\in\bC$ determines the 
other coefficients recursively. Hence, the space 
$\Big( \bigcap_{j=1}^{n-1} \Kern (\hat{\pi}_j) \Big) \, \cap \, \Kern(K-q^{n-2\ell})$
is at most one dimensional. On the other hand, we already noticed
that $\Wd_{n+1} \subset \bigcap_{j=1}^{n-1} \Kern (\hat{\pi}_j)$,
and since the subrepresentation $\Wd_{n+1}$ intersects all nontrivial $K$-eigenspaces of $\Wd_2^{\tens n}$,
we get that $\Wd_{n+1} = \bigcap_{j=1}^{n-1} \Kern (\hat{\pi}_j)$.
Hence, (a) implies (b).
\end{proof}

The following consequence will be used in proving uniqueness results.
\begin{cor}\label{cor: all projections vanish gives zero}
If $n\in\bZpos$ and the vector $v\in\Wd_{2}^{\tens n}$ satisfies $E.v=0$, $K.v=v$ and $\hat{\pi}_j(v)=0$ for all $1\leq j<n$, then $v=0.$  
\end{cor}
\begin{proof}
The conditions $E.v=0$ and $K.v=v$ state that $v$ lies in a trivial
subrepresentation $\Wd_{1}\subset\Wd_{2}^{\tens n}$.
On the other hand, 
by the previous lemma,
$v \in \bigcap_{j=1}^{n-1} \Kern(\hat{\pi}_j)$ implies $v\in\Wd_{n+1}\subset\Wd_{2}^{\tens n}$.
Combining, we get
$v\in\Wd_{1}\cap\Wd_{n+1} = \set{0} $.
\end{proof}

\subsection{The spin chain - Coulomb gas correspondence}

Our solution to the system 
\eqref{eq: multiple SLE PDEs}~--~\eqref{eq: multiple SLE asymptotics}
is based on the following correspondence, which is a particular case
of the main results of \cite{KP-conformally_covariant_boundary_correlation_functions_with_a_quantum_group}.
It associates solutions of \eqref{eq: multiple SLE PDEs}~--~\eqref{eq: multiple SLE Mobius covariance}
to vectors in the trivial subrepresentation $\HWsp_1(\Wd_2^{\tens 2N})$, 
which was defined in \eqref{eq: highest vector space}.
We denote below by $\bH = \set{ z \in \bC \; \big| \; \im(z) > 0}$
the upper half-plane, and recall that its conformal self-maps are M\"obius transformations 
$\Mob(z) = \frac{a z + b}{c z + d}$, where $a,b,c,d \in \bR$ and $a d - b c > 0$.

\begin{thm}[{\cite{KP-conformally_covariant_boundary_correlation_functions_with_a_quantum_group}}]\label{thm: SCCG correspondence special case}
Let $\kappa \in (0,8) \setminus \bQ$ and $q=e^{\ii 4\pi/\kappa}$.
There exist linear maps $\sF \colon \HWsp_1(\Wd_2^{\tens 2N}) \to \sC^\infty(\chamber_{2N})$,
for all $N \in \bZnn$, such that the following hold for any $v\in \HWsp_1(\Wd_2^{\tens 2N})$.
\begin{description}
\item [{(PDE)}] The function $\PartF = \sF[v] \colon \chamber_{2N} \to \bC$ satisfies the partial differential equations \eqref{eq: multiple SLE PDEs}.
\item [{(COV)}] For any M\"obius transformation $\Mob\colon\bH\to\bH$ 
such that $\Mob(x_{1})<\Mob(x_{2})<\cdots<\Mob(x_{2N})$,
the covariance \eqref{eq: multiple SLE Mobius covariance} holds for the function $\PartF = \sF[v]$.
\item [{(ASY)}]
For each $j=1,\ldots,2N-1$, we have
$\hat{v} = \hat{\pi}_{j}(v) \in \HWsp_1(\Wd_2^{\tens 2(N-1)})$.
Denote $B = \frac{\Gamma(1-4/\kappa)^2}{\Gamma(2-8/\kappa)}$ and $h = \frac{6-\kappa}{2 \kappa}$, and let 
$\xi \in (x_{j-1},x_{j+2})$.
The function $\sF[v] \colon \chamber_{2N} \to \bC$ has the 
asymptotics 
\begin{align*}
\lim_{x_{j},x_{j+1}\to\xi}\frac{\sF[v](x_{1},\ldots,x_{2N})}{(x_{j+1}-x_{j})^{-2 h}}
=\; & B\times\sF[\hat{v}](x_{1},\ldots,x_{j-1},x_{j+2},\ldots,x_{2N}) .
\end{align*}
\end{description}
\end{thm}
\begin{proof}
This follows as a special case of Theorem~4.17 and Lemma~3.4
in \cite{KP-conformally_covariant_boundary_correlation_functions_with_a_quantum_group}.
\end{proof}


\bigskip{}

\section{\label{sec: Multiple SLE pure geometries}Quantum group solution of the pure partition functions}

In view of Theorem~\ref{thm: SCCG correspondence special case}, the problem 
of finding the multiple $\SLE$ pure partition functions for $N$ curves
is reduced to finding a certain basis of the trivial subrepresentation
$\HWsp_1(\Wd_{2}^{\tens 2N})$.
More precisely, 
Equations~\eqref{eq: multiple SLE PDEs}~--~\eqref{eq: multiple SLE asymptotics}
correspond to the following linear system of equations for
$v_\alpha\in\Wd_{2}^{\tens 2|\alpha|}$, $\alpha\in\LP$:
\begin{align}
&K.v_\alpha=v_\alpha\label{eq: multiple sle cartan eigenvalue}\\
&E.v_\alpha=0\label{eq: multiple sle highest weight vector}\\
&\hat{\pi}_j(v_\alpha)=\;\label{eq: multiple sle projection conditions}\begin{cases} 0 & \mbox{if } \link{j}{j+1}\notin\alpha\\ v_{\alpha\removeLink\link{j}{j+1}} & \mbox{if } \link{j}{j+1}\in\alpha \end{cases}
\qquad\text{for all }j=1,\ldots,2|\alpha|-1.
\end{align}
 
The first two equations
\eqref{eq: multiple sle cartan eigenvalue}~--~\eqref{eq: multiple sle highest weight vector} 
state that $v_\alpha$  belongs to 
the trivial subrepresentation $\HWsp_1(\Wd_{2}^{\tens 2|\alpha|})$.
By the (PDE) and (COV) parts of
Theorem~\ref{thm: SCCG correspondence special case},
they correspond to 
Equations~\eqref{eq: multiple SLE PDEs}~--~\eqref{eq: multiple SLE Mobius covariance}
for a partition function
\begin{align*}
\PartF_\alpha(x_{1},\ldots,x_{2|\alpha|})\propto\sF[v_\alpha](x_{1},\ldots,x_{2|\alpha|}).
\end{align*}
Equations~\eqref{eq: multiple sle projection conditions} 
correspond to the 
asymptotic conditions~\eqref{eq: multiple SLE asymptotics} 
specified by the link pattern $\alpha$,
by the (ASY) part of Theorem~\ref{thm: SCCG correspondence special case}.

The main result of this section is the construction of the solutions
$v_\alpha$. 
\begin{thm}\label{thm: existence of multiple SLE vectors}
There exists a unique collection $\left(v_{\alpha}\right)_{\alpha\in\LP}$ of vectors $v_{\alpha}\in\Wd_{2}^{\tens 2|\alpha|}$ such that the system of equations~\eqref{eq: multiple sle cartan eigenvalue}~--~\eqref{eq: multiple sle projection conditions} holds for all $\alpha\in\LP$, with the normalization $v_\emptyset=1$.
\end{thm}

The proof is based on a number of steps achieved in
Sections~\ref{subsec: uniqueness}~--~\ref{subsec: general pattern}, which are combined
in Section~\ref{subsec: proof of existence}.
In Section \ref{sec: Multiple SLE partition functions}, the solutions $v_\alpha$ will be
converted to the multiple $\SLE$ pure partition functions $\PartF_\alpha$
with the help of Theorem~\ref{thm: SCCG correspondence special case}.

\subsection{\label{subsec: uniqueness}Uniqueness of solutions}

We will first show that solutions to the
system~\eqref{eq: multiple sle cartan eigenvalue}~--~\eqref{eq: multiple sle projection conditions}
are necessarily unique, up to normalization. 
The corresponding homogeneous system requires
$v \in \HWsp_1(\Wd_2^{\tens 2N})$ to have vanishing projections
$\hat{\pi}_j(v)$ for each $j$.
Corollary~\ref{cor: all projections vanish gives zero}
shows that the homogeneous problem only admits the trivial solution.

\begin{prop}\label{prop: uniqueness}
Let $\left(v_\alpha\right)_{\alpha\in\LP}$ and 
$\left(v'_\alpha\right)_{\alpha\in\LP}$ be two collections 
of solutions to~\eqref{eq: multiple sle cartan eigenvalue}~--~\eqref{eq: multiple sle projection conditions} 
such that $v_\emptyset,v'_\emptyset\neq0$. 
Then there is a constant $c\in\bC\setminus\{0\}$ so that
\begin{align*}
v'_\alpha=c\,v_\alpha\qquad\text{for all }\alpha\in\LP.
\end{align*}
\end{prop}
\begin{proof}
Clearly $v'_\emptyset=c\,v_\emptyset$ for some
$c\in\bC\setminus\{0\}$. Let $N\geq1$ and suppose the condition
$v'_\beta=c\,v_\beta$ holds for all $\beta\in\LP_{N-1}$.
Then the equations~\eqref{eq: multiple sle cartan eigenvalue}~--~\eqref{eq: multiple sle projection conditions} 
for $v'_\alpha$ and $c\,v_\alpha$ coincide for each
$\alpha\in\LP_N,$ and it follows from 
Corollary~\ref{cor: all projections vanish gives zero} that
$v'_\alpha-c\,v_\alpha=0.$ The assertion follows by induction on $N$.
\end{proof}

We next proceed with the construction of the solutions to 
\eqref{eq: multiple sle cartan eigenvalue}~--~\eqref{eq: multiple sle projection conditions}. 
From now on, we shall fix the normalization by $v_\emptyset=1$.

\subsection{\label{subsec: construction}\label{subsub: rainbow pattern}Solution for rainbow patterns}

\begin{figure}
\includegraphics[scale=1.5]{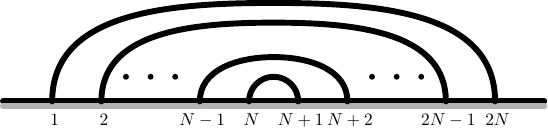}
\caption{\label{fig: rainbow pattern}
The rainbow pattern $\nested_N$.}
\end{figure}

We begin with the solution corresponding to a special case, the rainbow link pattern defined for $N\in\bN$ by
\begin{align*}
\nested_N=\left\lbrace\link{1}{2N},\link{2}{2N-1},\ldots,\link{N}{N+1}\right\rbrace\in\LP_N,
\end{align*}
illustrated in Figure~\ref{fig: rainbow pattern}. 
For the rainbow pattern $\alpha=\nested_N$,
Equations~\eqref{eq: multiple sle cartan eigenvalue}~--~\eqref{eq: multiple sle projection conditions} 
allow us to find $v_{\nested_N}$ recursively, as they involve only $\nested_{N-1}$ in their inhomogeneous terms:
\begin{align}
(K-1).v_{\nested_N} = \; & 0\label{eq: cartan eigenvalue for nested}\\
E.v_{\nested_N} = \; & 0\label{eq: highest weight vector for nested}\\
\hat{\pi}_N(v_{\nested_N}) = \; & v_{\nested_{N-1}}\qquad\text{and}\qquad
\hat{\pi}_j(v_{\nested_N}) =  0 \quad\text{for}\;j\neq N .
\label{eq: projections for nested}
\end{align}

To give an explicit expression for the solutions $v_{\nested_N}$
we introduce the following notation.
Recall from Section~\ref{subsec: quantum group} that 
the tensor product $\Wd_{2}^{\tens N}$ contains the 
irreducible subrepresentation $\Wd_{N+1}$ generated by 
$\Wbas_{0}\tens\cdots\tens\Wbas_{0}.$ 
Denote
\begin{align*}
\MTbas_{0}^{(N)}=\Wbas_{0}\tens\cdots\tens\Wbas_{0}\,\in\Wd_{N+1}\subset\Wd_{2}^{\tens N}\qquad\text{and }\qquad
\MTbas_{l}^{(N)}=F^l.\MTbas_{0}^{(N)}\quad\text{for }0\leq l\leq N.
\end{align*}
Then 
$(\MTbas_{l}^{(N)})_{0\leq l\leq N}$ is a basis of the 
subrepresentation $\Wd_{N+1}\subset\Wd_{2}^{\tens N}.$

We will prove that the solutions for the rainbow
patterns are given by the following formulas.
\begin{prop}\label{prop: solution for nested}
The vectors
\begin{align*}
v_{\nested_N}:=\;\frac{1}{(q^{-2}-1)^{N}}\frac{\qnum{2}^{N}}{\qfact{N+1}}\,\sum_{l=0}^N(-1)^lq^{l(N-l-1)}\times\left(\MTbas_l^{(N)}\tens\MTbas_{N-l}^{(N)}\right)\,\in\Wd_{2}^{\tens 2N}
\end{align*}
for $N\in\bZnn$ determine the unique solution to \eqref{eq: cartan eigenvalue for nested}~--~\eqref{eq: projections for nested} with $v_\emptyset=1$.
\end{prop}

Below, we record a recursion for the vectors $\MTbas_l^{(N)}$, needed in 
the proof of Proposition~\ref{prop: solution for nested}.
We use the convention 
$\MTbas_{l}^{(N)}=0$ for $l<0$ and $l>N$. 

\begin{lem}\label{lem: recursion for tensor basis}
The following formulas hold for $\MTbas_{l}^{(N)}\in\Wd_{N+1}\subset\Wd_{2}^{\tens N}$. 
\begin{description}
\item[(a)] $\MTbas_{l}^{(N)}=\,\MTbas_{l}^{(N-1)}\tens\Wbas_{0}+q^{l-N}\qnum l\,\MTbas_{l-1}^{(N-1)}\tens\Wbas_{1}$
\item[(b)] $\MTbas_{l}^{(N)}=\,q^{-l}\,\Wbas_{0}\tens\MTbas_{l}^{(N-1)}+\qnum l\,\Wbas_{1}\tens\MTbas_{l-1}^{(N-1)}$
\end{description}
\end{lem}
\begin{proof}
The asserted formulas clearly hold for $l=0$. The general case follows
by induction on $l$, using the action of $F$ on a double tensor
product from \eqref{eq: coproduct} and the 
definition \eqref{eq: q num basic def} of the $q$-integers $\qnum{l}$.
\end{proof}


\begin{proof}[Proof of Proposition~\ref{prop: solution for nested}]
The normalization $v_\emptyset=1$ is clear from the asserted formula.
Equation~\eqref{eq: cartan eigenvalue for nested} follows directly:
each term $\MTbas^{(N)}_l \tens \MTbas^{(N)}_{N-l}$ is a $K$-eigenvector
of eigenvalue $1$, by the action \eqref{eq: coproduct} of $K$ on
$\Wd_{N+1}\tens\Wd_{N+1}$. Similarly, by \eqref{eq: coproduct}, we have
\begin{align*}
E.v_{\nested_N}=&\sum_{l=0}^Nc_l^{(N)}\times\left(E.\MTbas_l^{(N)}\tens K.\MTbas_{N-l}^{(N)}+\MTbas_l^{(N)}\tens E.\MTbas_{N-l}^{(N)}\right)\\
= \; & \sum_{l=0}^{N-1}\qnum{l+1}\qnum{N-l}\left(c_{l+1}^{(N)}q^{-N+2l+2}+c_l^{(N)}\right)\times\left(\MTbas_l^{(N)}\tens\MTbas_{N-l-1}^{(N)}\right),\\
\text{where}\quad c_l^{(N)} = \; & \frac{1}{(q^{-2}-1)^{N}}\frac{\qnum{2}^{N}}{\qfact{N+1}}(-1)^lq^{l(N-l-1)},
\end{align*}
and Equation~\eqref{eq: highest weight vector for nested} 
follows from the recursion
\begin{align}\label{eq: recursion for coefficients}
c_{l+1}^{(N)}=-q^{N-2(l+1)}c_l^{(N)}.
\end{align}

When $j\neq N$, we have $\hat{\pi}_j(v_{\nested_N})=0$
because all the vectors
$\MTbas_{l}^{(N)}\in\Wd_{N+1}\subset\Wd_{2}^{\tens N}$ have that
property, by Lemma~\ref{lem: all projections vanish}. 
It remains to calculate the projection
$\hat{\pi}_N(v_{\nested_N})$. 
Using Lemma~\ref{lem: recursion for tensor basis}, we write
\begin{align*}
v_{\nested_N}=\;\sum_{l=0}^Nc_l^{(N)}\times\left(\MTbas_{l}^{(N-1)}\tens\Wbas_{0}+q^{l-N}\qnum{l}\,\MTbas_{l-1}^{(N-1)}\tens\Wbas_{1}\right)\tens\left(q^{l-N}\,\Wbas_{0}\tens\MTbas_{N-l}^{(N-1)}+\qnum{N-l}\,\Wbas_{1}\tens\MTbas_{N-l-1}^{(N-1)}\right).
\end{align*}
With the help of Lemma~\ref{lem: projection formulas}(a), we calculate
\begin{align}\label{eq: projection for rainbow}
\hat{\pi}_N(v_{\nested_N})=&\;\sum_{l=0}^Nc_l^{(N)}\times\Big(\frac{q^{-1}-q}{\qnum 2}\qnum{N-l}\,\MTbas_{l}^{(N-1)}\tens\MTbas_{N-l-1}^{(N-1)}+\frac{1-q^{-2}}{\qnum 2}q^{2l-2N}\qnum{l}\,\MTbas_{l-1}^{(N-1)}\tens\MTbas_{N-l}^{(N-1)}\Big)\nonumber\\
=&\;\frac{q^{-1}-q}{\qnum 2}\,\sum_{l=0}^{N-1}\left(c_l^{(N)}\qnum{N-l}-c_{l+1}^{(N)}q^{2l-2N+1}\qnum{l+1}\right)\times\left(\MTbas_{l}^{(N-1)}\tens\MTbas_{N-l-1}^{(N-1)}\right).
\end{align}
Using the recursion~\eqref{eq: recursion for coefficients} for $c_l^{(N)}$,
the relation
\begin{align*}
\qnum{N-l} + q^{-N-1}\qnum{l+1}=q^{-l-1}\qnum{N+1}
\end{align*}
for the $q$-integers, and the formula $\qnum 2=q+q^{-1}$,
we observe that
\begin{align*}
\frac{q^{-1}-q}{\qnum 2}\left(c_l^{(N)}\qnum{N-l}-c_{l+1}^{(N)}q^{2l-2N+1}\qnum{l+1}\right)=c_l^{(N-1)}.
\end{align*}
Substituting this to \eqref{eq: projection for rainbow}, it follows that
\begin{align*}
&\hat{\pi}_N(v_{\nested_N})=v_{\nested_{N-1}}.
\end{align*}
This concludes the proof.
\end{proof}
\begin{figure}
\begin{displaymath}
\xymatrixcolsep{4pc}
\xymatrixrowsep{4pc}
\xymatrix{
        & \begin{minipage}{4cm}\includegraphics[scale=.88]{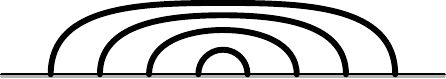}\end{minipage}\ar[d] \\
        & \hspace{-0mm}\begin{minipage}{4cm}\includegraphics[scale=.88]{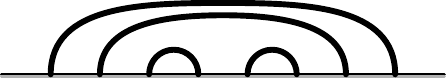}\end{minipage} \ar[dl]\ar[dr] & \\
        \hspace{-0mm}\begin{minipage}{4cm}\includegraphics[scale=.88]{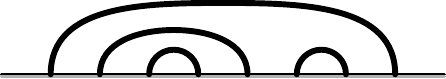}\end{minipage}\ar[d]\ar[dr] & & 
        \hspace{-0mm}\begin{minipage}{4cm}\includegraphics[scale=.88]{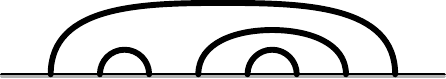}\end{minipage}\ar[dl]\ar[d]\\
        \hspace{-0mm}\begin{minipage}{4cm}\includegraphics[scale=.88]{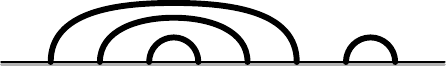}\end{minipage}\ar[d] & 
        \hspace{-0mm}\begin{minipage}{4cm}\includegraphics[scale=.88]{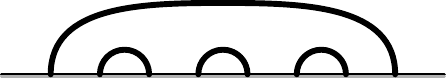}\end{minipage}\ar[dl]\ar[d]\ar[dr] & 
        \hspace{-0mm}\begin{minipage}{4cm}\includegraphics[scale=.88]{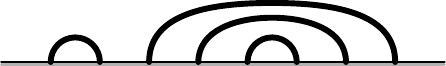}\end{minipage}\ar[d] \\
        \hspace{-0mm}\begin{minipage}{4cm}\includegraphics[scale=.88]{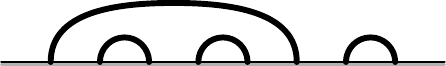}\end{minipage}\ar[d]\ar[dr] & 
        \hspace{-0mm}\begin{minipage}{4cm}\includegraphics[scale=.88]{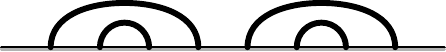}\end{minipage}\ar[dl]\ar[dr] & 
        \hspace{-0mm}\begin{minipage}{4cm}\includegraphics[scale=.88]{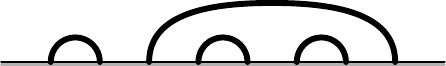}\end{minipage}\ar[dl]\ar[d] \\
        \hspace{-0mm}\begin{minipage}{4cm}\includegraphics[scale=.88]{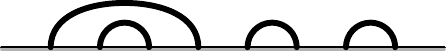}\end{minipage}\ar[dr] & 
        \hspace{-0mm}\begin{minipage}{4cm}\includegraphics[scale=.88]{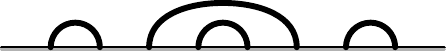}\end{minipage}\ar[d] & 
        \hspace{-0mm}\begin{minipage}{4cm}\includegraphics[scale=.88]{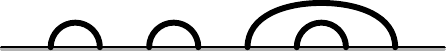}\end{minipage}\ar[dl] \\
        & \hspace{-0mm}\begin{minipage}{4cm}\includegraphics[scale=.88]{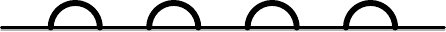}\end{minipage} &
    }
\end{displaymath}
\caption{\label{fig: partially ordered set}
The partially ordered set $\LP_4$.}
\end{figure}

\subsection{\label{subsec: general pattern}Solutions for general patterns}
Next we introduce a natural partial order and certain tying operations 
on the set $\LP$ of link patterns. 
We will prove auxiliary results of combinatorial nature, 
in order to obtain a recursive procedure for solving the 
system~\eqref{eq: multiple sle cartan eigenvalue}~--~\eqref{eq: multiple sle projection conditions}
with a general link pattern $\alpha\in\LP_N$.

The partial order on the set $\LP$ is inherited from the
set $\walks$ of walks which $\LP$ is in bijection with. 
More precisely, $\walks=\bigsqcup_{N\geq0}\walks_N$, where
\begin{align*}
\walks_N=\set{W\,\colon\,\{0,1,\ldots,2N\}\rightarrow\bZnn\;\Big|\;W_0=W_{2N}=0,\;|W_t-W_{t-1}|=1\;\text{ for all }\;t=1,\ldots,2N }
\end{align*}
is the set of non-negative walks 
of $2N$ steps 
starting and ending at zero. A walk $W^{(\alpha)}\in\walks_N$ associated to
a link pattern $\alpha\in\LP_N$ is defined recursively as
\begin{align*}
W^{(\alpha)}_0=\,0\qquad\text{and}\qquad W^{(\alpha)}_t=\;\begin{cases}
W^{(\alpha)}_{t-1}+1\quad & \text{if }\,t\,\text{ is a left endpoint of a link in}\;\alpha\\
W^{(\alpha)}_{t-1}-1\quad & \text{if }\,t\,\text{ is a right endpoint of a link in}\;\alpha,
\end{cases}
\end{align*}
and $\alpha \mapsto W^{(\alpha)}$ defines a bijection $\LP_N\to\walks_N$. On $\walks_N$,
there is a natural partial order defined by
\begin{align*}
W\preceq W'\qquad\iff\qquad W_t\leq W'_t\qquad\text{for all }t=1,\ldots,2N,
\end{align*}
which induces the partial order on $\LP_N$ by
$\;\alpha\preceq \alpha'  \iff  W^{(\alpha)}\preceq W^{(\alpha')}$.
We observe that the rainbow pattern
$\nested_N\in\LP_N$ is the unique maximal element in $\LP_N$
with respect to the above partial order. 
As an example, the partially ordered set $\LP_4$ is depicted in 
Figure~\ref{fig: partially ordered set}.

For $\alpha\in\LP_N$, we denote by
\begin{align*}
&\LP^{\succeq\alpha} = \left(\bigcup_{n<N}\LP_n\right)\cup\Big\{\beta\in\LP_N\;|\;\beta\succeq\alpha\Big\}
\qquad \text{and} \qquad \LP^{\succ\alpha} = \LP^{\succeq\alpha} \setminus \set{\alpha}.
\end{align*}

Let $j\in\{1,\ldots,2N-1\}$ be fixed. 
We define the tying operation
$\tieOp_j\colon\LP_N\to\LP_N$ by
\begin{align}\label{eq: tying operation}
\tieOp_j(\alpha)=\alpha^{\tieOp_j}:=
\begin{cases}
\alpha & \mbox{if }\,\link{j}{j+1}\in\alpha \\
\Big( \alpha \setminus \set{\link{j}{l_1}, \link{j+1}{l_2}} \Big)\cup \set{\link{j}{j+1},\link{l_1}{l_2}} \quad & \mbox{if }\,\link{j}{j+1}\notin\alpha,
\end{cases}
\end{align}
where $l_1$ and $l_2$ are the endpoints of the links in $\alpha$ where 
$j$ and $j+1$ are connected to in the latter case 
--- see also Figure~\ref{fig: tying operation}.
Observe that
$\link{j}{j+1}\in\alpha$ if and only if $\alpha=\alpha^{\tieOp_j}$.

\begin{figure}
\includegraphics[scale=1]{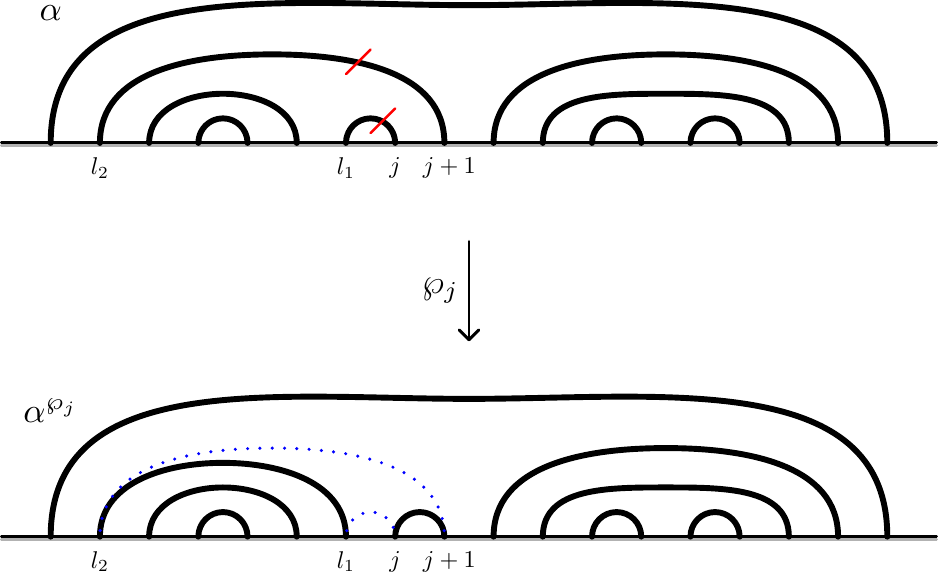}
\caption{\label{fig: tying operation}
In the tying operation $\tieOp_j$, 
the link pattern $\alpha$ is mapped to 
$\alpha^{\tieOp_j}$ by cutting 
the links $\link{l_1}{j}$ and $\link{l_2}{j+1}$ and 
connecting the endpoints so as to form the links
$\link{j}{j+1}$ and $\link{l_2}{l_1}$.}
\end{figure}

\begin{figure}
\includegraphics[width=1\textwidth]{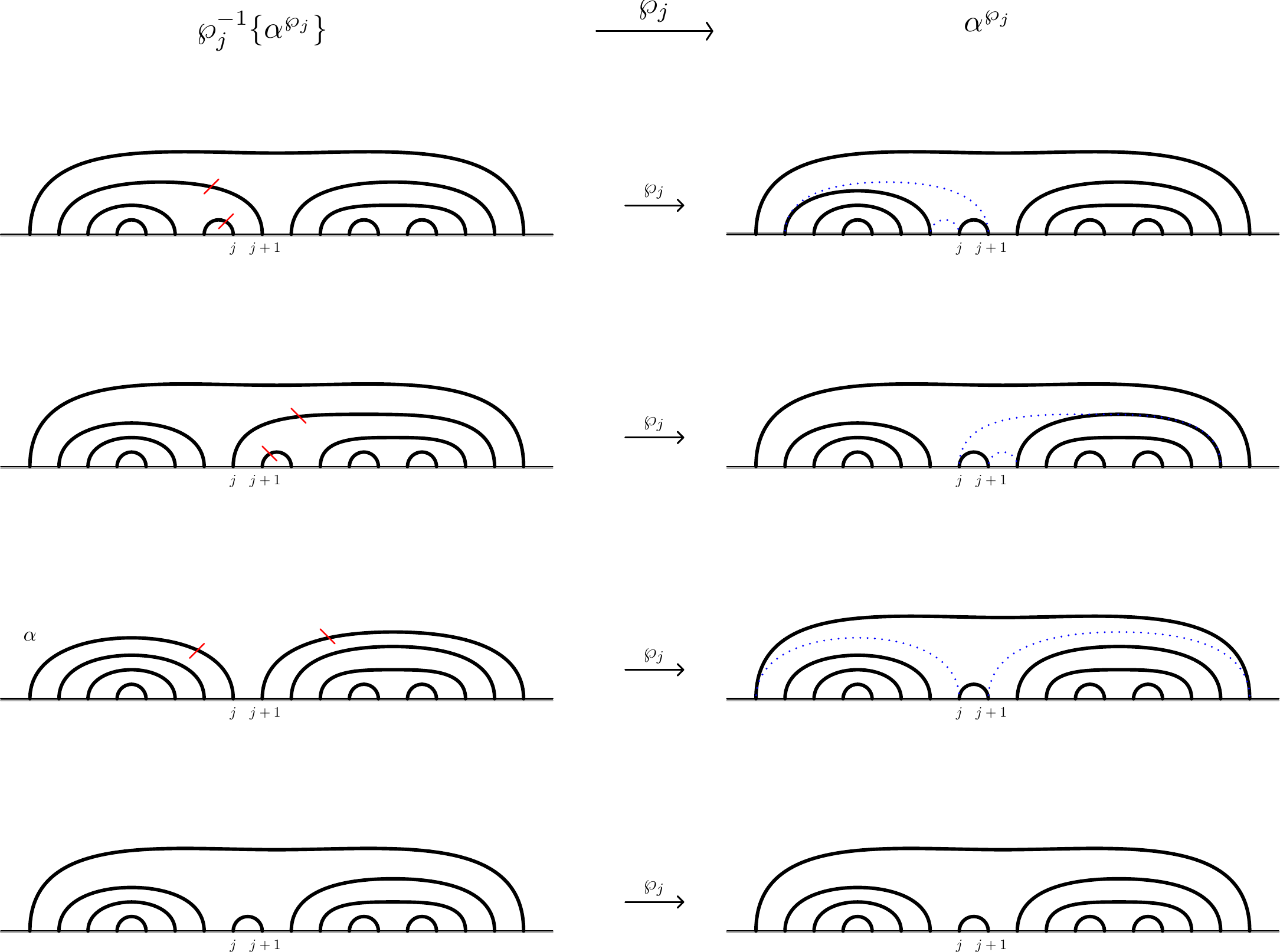}
\caption{\label{fig: inverse image of tying}
The inverse image of the tying operation $\tieOp_j$
for the link pattern $\alpha$ can be found by cutting 
the link $\link{j}{j+1}$ in $\alpha^{\tieOp_j}$ and connecting 
its endpoints in all possible ways to a link lying in 
the same connected component as $\link{j}{j+1}$.
In this example, the index $j$ is chosen as in Lemma~\ref{lem: downwards lemma}(i).}
\end{figure}

The solutions $(v_\alpha)_{\alpha \in \LP}$ will be shown to have a certain property, which is
also used for their construction recursively down the partially
ordered set of link patterns (see Figure~\ref{fig: partially ordered set}).
A closely related formula for connectivity weights was independently observed in
\cite[Equation~(69)]{FK-solution_space_for_a_system_of_null_state_PDEs_4}.
This key property is the equality
\begin{align}\label{eq: solution for general pattern}
v_\alpha = \qnum 2\,\Big(\id-\pi_j\Big)\,(v_{\alpha^{\tieOp_j}})\;-\;
    \sum_{\substack{\beta\in\LP_N\setminus\{\alpha,\alpha^{\tieOp_j}\}\\\beta^{\tieOp_j}\,=\,\alpha^{\tieOp_j}}}\;v_{\beta}
\end{align}
whenever $\alpha \in \LP_N$ and $j$ is such that
$\link{l_1}{j} , \link{j+1}{l_2} \in \alpha$ with $l_1<j<j+1<l_2$.
Before turning to the construction, we record some combinatorial observations about
this setup.
\begin{lem}\label{lem: downwards lemma}
Fix $\alpha\in\LP_N$. 
The following are equivalent for $j \in \set{1,\ldots,2N-1}$:
\begin{description}
\item[(i)] $\link{l_1}{j} , \link{j+1}{l_2} \in \alpha$ with $l_1<j<j+1<l_2$.
\item[(ii)] $\link{j}{j+1} \notin \alpha$, and 
$\alpha\preceq\beta$ for all 
$\beta \in \tieOp_j^{-1}\set{\alpha^{\tieOp_j}}$.
\end{description}
Moreover, an index $j$ satisfying (i) and (ii) 
exists if $\alpha \neq \nested_N$,
and the following properties then hold.
\begin{description}
\item[(a)] If $|k-j|=1$, 
then there exists a unique $\beta_0 \in \tieOp_j^{-1} \set{\alpha^{\tieOp_j}}$ such that $\link{k}{k+1} \in \beta_0$.
This $\beta_0$ satisfies $\beta_0 \removeLink \link{k}{k+1} = \alpha^{\tieOp_j} \removeLink \link{j}{j+1}$.
\item[(b)] If $|k-j| > 1$ and $\link{k}{k+1} \notin \alpha$, then $\link{k}{k+1} \notin \alpha^{\tieOp_j}$.
\item[(c)] If $|k-j| > 1$ and $\link{k}{k+1} \in \alpha$, then the mapping $\beta \mapsto \hat{\beta} = \beta \removeLink \link{k}{k+1}$
defines a bijection
\[ \Big\{ \beta \in \tieOp_j^{-1} \set{\alpha^{\tieOp_j}}  \; \Big| \; \link{k}{k+1} \in \beta \Big\}
  \; \longrightarrow \;  \tieOp_{j'}^{-1} \set{\hat{\alpha}^{\tieOp_{j'}}}   , \]
where $\hat{\alpha} = \alpha \removeLink \link{k}{k+1}$, 
and $j' = j$ if $k>j$ and $j' = j-2$ if $k<j$.
\end{description}
\end{lem}
\begin{proof}
The equivalence of (i) and (ii) is straightforward from the definition of the partial order.
See also Figure~\ref{fig: inverse image of tying}.

For $\alpha \neq \nested_N$, the existence of $j$ satisfying condition (i) is easy to see.
For property (a), the desired link pattern is $\beta_0 = (\alpha^{\tieOp_j})^{\tieOp_k}$ ---
this is illustrated in Figure~\ref{fig: beta0}.
Property (b) follows from (i).
For property (c), the assumption of the presence of the link $\link{k}{k+1}$ implies that neither $j$ nor $j+1$
is connected to $k$ or $k+1$, so the link $\link{k}{k+1}$ plays no role in the tying operations: the map is well
defined to the asserted range, and the inverse map is obvious.
\end{proof}

\begin{figure}
\includegraphics[scale=1]{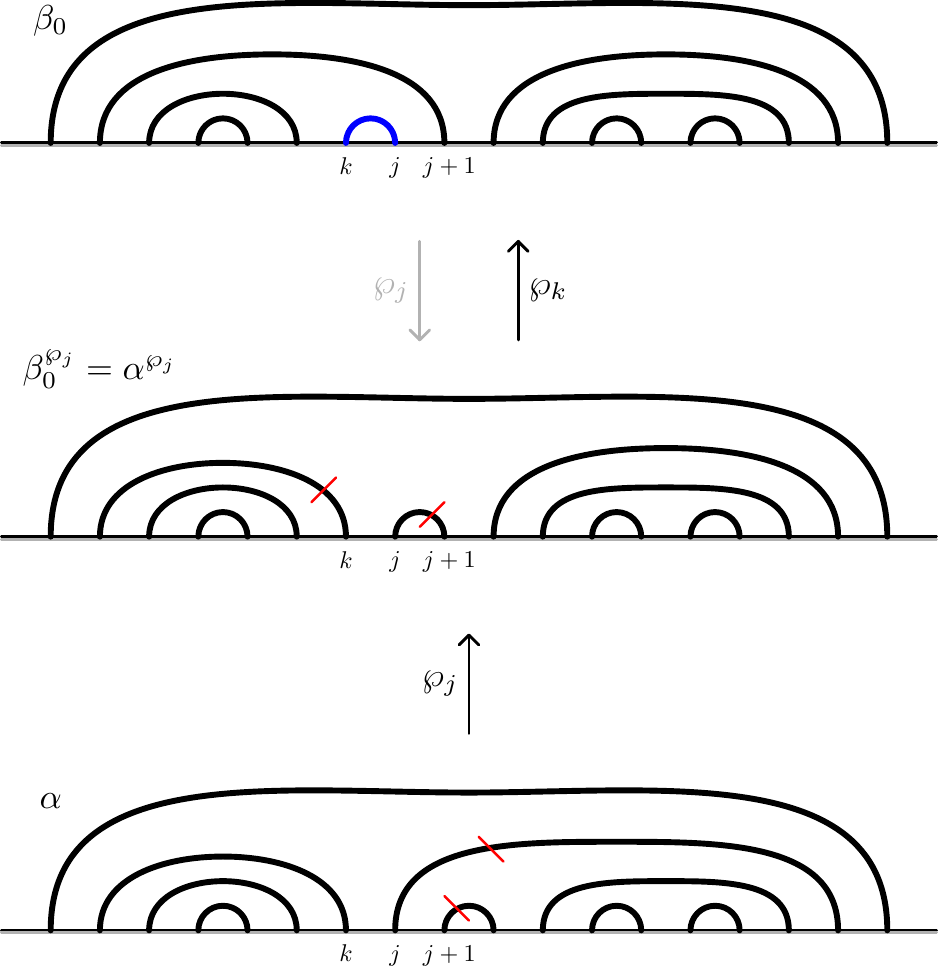}
\caption{\label{fig: beta0}
Suppose $k=j-1$ and $j$ is chosen as in 
Lemma~\ref{lem: downwards lemma}(i).
The unique $\beta_0 \in \tieOp_j^{-1} \set{\alpha^{\tieOp_j}}$
containing the link $\link{k}{k+1}=\link{k}{j}$
can be found by applying the map $\tieOp_k$ to 
the link pattern $\alpha^{\tieOp_j}$, 
as depicted in the figure.
Notice also that we clearly have
$\beta_0 \removeLink \link{k}{k+1} = \alpha^{\tieOp_j} \removeLink \link{j}{j+1}$.}
\end{figure}

We are now ready to construct the solutions to 
\eqref{eq: multiple sle cartan eigenvalue}~--~\eqref{eq: multiple sle projection conditions} 
for all $\alpha\in\LP$ from the 
solutions corresponding to the maximal patterns $\nested_N$,
given in Proposition~\ref{prop: solution for nested}.

\begin{prop}\label{prop: construction of solutions}
Let $\alpha\in\LP_N\setminus\{\nested_N\}$. 
Suppose that the collection 
$\left(v_{\beta}\right)_{\beta\in\LP^{\succ\alpha}}$ satisfies the 
system~\eqref{eq: multiple sle cartan eigenvalue}~--~\eqref{eq: multiple sle projection conditions} 
and the equations \eqref{eq: solution for general pattern} for all $\beta\in\LP^{\succ\alpha}$.
Then the vector $v_\alpha \in \Wd_{2}^{\tens 2N}$ can be defined in accordance with Equation~\eqref{eq: solution for general pattern} as
\begin{align*}
v_\alpha \; := \; \qnum 2\,\Big(\id-\pi_j\Big)\,(v_{\alpha^{\tieOp_j}})\;-\;
    \sum_{\substack{\beta\in\LP_N\setminus\{\alpha,\alpha^{\tieOp_j}\}\\\beta^{\tieOp_j}\,=\,\alpha^{\tieOp_j}}}\;v_{\beta} 
\end{align*}
for any $j$ as in Lemma~\ref{lem: downwards lemma}.
This vector is a solution to 
\eqref{eq: multiple sle cartan eigenvalue}~--~\eqref{eq: multiple sle projection conditions} 
for $\alpha$.
\end{prop}
\begin{proof}
Choose a $j$ as in Lemma~\ref{lem: downwards lemma}.
By property (ii), 
the link patterns $\alpha^{\tieOp_j}$ and $\beta$ needed 
in the formula defining $v_\alpha$
satisfy $\alpha^{\tieOp_j} \succ \alpha$ and $\beta \succ \alpha$, and the corresponding vectors $v_{\alpha^{\tieOp_j}}$
and $v_\beta$ are thus assumed given. By assumption, each of these vectors 
satisfies \eqref{eq: multiple sle cartan eigenvalue}~--~\eqref{eq: multiple sle projection conditions}, 
so it readily follows that $v_\alpha$
also satisfies \eqref{eq: multiple sle cartan
eigenvalue} and \eqref{eq: multiple sle highest weight vector}. 
It remains to be shown that for any
$k\in\{1,\ldots,2N-1\}$, the projection $\hat{\pi}_k(v_\alpha)$
gives \eqref{eq: multiple sle projection conditions}.
Note that once we have shown that $v_\alpha$ satisfies
\eqref{eq: multiple sle cartan eigenvalue}~--~\eqref{eq: multiple sle projection conditions},
it follows from the uniqueness argument of Proposition~\ref{prop: uniqueness} that the constructed vector $v_\alpha$
is independent of the choice of $j$, and thus in particular
satisfies \eqref{eq: solution for general pattern}. 

We divide the calculations to three separate cases:
(i): $k=j$, (ii): $\,|k-j|=1\,$ and (iii): $\,|k-j|>1$.

Let us start from the easiest case (i): $k=j$.
To establish \eqref{eq: multiple sle projection conditions} in this case,
we need to show $\hat{\pi}_j(v_\alpha) = 0$.
Since $\pi_j$ is a projection, the first term $(\id-\pi_j)\,(v_{\alpha^{\tieOp_j}})$
in the formula defining $v_\alpha$ 
is annihilated by $\hat{\pi}_j$.
We have $\hat{\pi}_j(v_\beta) = 0$ also for all the other terms,
since $\link{j}{j+1} \notin \beta $ for
$\beta \in \big( \tieOp_j^{-1} \set{\alpha} \big) \setminus \{\alpha,\alpha^{\tieOp_j}\}$.
This concludes the case (i).


Consider then the case (ii): $\,|k-j|=1$.
We need to compute $\hat{\pi}_k(v_\alpha)$ and compare with \eqref{eq: multiple sle projection conditions}.
The application of $\pi_k$ on
the first term of the formula defining $v_\alpha$ 
gives
\begin{align*}
\qnum2\,\pi_k\Big(\id-\pi_j\Big)(v_{\alpha^{\tieOp_j}})\;=\;-\qnum2\,\Big(\pi_k\circ\pi_j\Big)(v_{\alpha^{\tieOp_j}})
\end{align*}
since $\link{k}{k+1}\notin\alpha^{\tieOp_j}$. Now,
\eqref{eq: multiple sle projection conditions} for $\alpha^{\tieOp_j}$
gives $\hat{\pi}_j(v_{\alpha^{\tieOp_j}})=v_{\alpha^{\tieOp_j}\removeLink\link{j}{j+1}}$
and we see that the vector $\pi_j(v_{\alpha^{\tieOp_j}})$ can be obtained
from $v_{\alpha^{\tieOp_j}\removeLink\link{j}{j+1}}\in\Wd_{2}^{\tens2(N-1)}$
by inserting the singlet vector $\Sbas\in\Wd_{2}\tens\Wd_{2}$ into the 
$j$:th and $j+1$:st tensor positions.
Therefore, by Lemma~\ref{lem: projection formulas}(b), we obtain
\begin{align*}
\qnum2\,\hat{\pi}_k\Big(\id-\pi_j\Big)(v_{\alpha^{\tieOp_j}})\;=\;-\qnum2\,\Big(\hat{\pi}_k\circ\pi_j\Big)(v_{\alpha^{\tieOp_j}})\;=\;v_{\alpha^{\tieOp_j}\removeLink\link{j}{j+1}}.
\end{align*}
In the sum in the formula defining $v_\alpha$, only the term corresponding to the link pattern $\beta_0$
of Lemma~\ref{lem: downwards lemma}(a) can survive the projection $\pi_k$, as the others do not contain the link $\link{k}{k+1}$. Therefore,
\begin{align*}
\hat{\pi}_k\Big(\sum_{\substack{\beta\in\LP_N\setminus\{\alpha,\alpha^{\tieOp_j}\}\\\beta^{\tieOp_j}\,=\,\alpha^{\tieOp_j}}}\;v_{\beta}\Big)
\;=\; & \begin{cases}
\hat{\pi}_k(v_{\beta_0})\;=\;v_{\beta_0\removeLink\link{k}{k+1}} & \text{if }\beta_0\neq\alpha\\
0 \qquad & \text{if }\beta_0=\alpha.
\end{cases}
\end{align*}
The patterns $\beta_0\removeLink\link{k}{k+1}$
and $\alpha^{\tieOp_j}\removeLink\link{j}{j+1}$ are the same by Lemma~\ref{lem: downwards lemma}(a), 
and $\link{k}{k+1} \in \alpha$ if and only if $\beta_0 = \alpha$.
We can therefore combine the above observations and conclude that
\begin{align*}
\hat{\pi}_k(v_\alpha) = \begin{cases}
v_{\alpha^{\tieOp_j}\removeLink\link{j}{j+1}}\,-\,v_{\beta_0\removeLink\link{k}{k+1}} \; = \; 0 & \text{if }\beta_0\neq\alpha\\
v_{\alpha^{\tieOp_j}\removeLink\link{j}{j+1}} \qquad\qquad\qquad \; = \; v_{\alpha\removeLink\link{k}{k+1}} \qquad & \text{if }\beta_0=\alpha.
\end{cases}
\end{align*}
This shows \eqref{eq: multiple sle projection conditions} for $\alpha$, and concludes the case (ii).

Finally, consider the case (iii): $\,|k-j|>1$.
We must again calculate $\hat{\pi}_k(v_\alpha)$.
The projections $\pi_j$ and $\pi_k$ now commute, and thus 
the projection $\pi_k$ of the first term in the defining equation of $v_\alpha$ reads
\begin{align}\label{eq: first term}
\qnum2\,\pi_k\Big(\id-\pi_j\Big)(v_{\alpha^{\tieOp_j}})\;=\;\qnum2\,\Big(\id-\pi_j\Big)\pi_k(v_{\alpha^{\tieOp_j}})
\end{align}
and the projection $\hat{\pi}_k$ of the second term reads
\begin{align}\label{eq: second term}
\hat{\pi}_k\Big(\sum_{\substack{\beta\in\LP_N\setminus\{\alpha,\alpha^{\tieOp_j}\}\\\beta^{\tieOp_j}\,=\,\alpha^{\tieOp_j}}}\;v_{\beta}\Big)\;=\;\sum_{\substack{\beta\in\LP_N\setminus\{\alpha,\alpha^{\tieOp_j}\}\\\beta^{\tieOp_j}\,=\,\alpha^{\tieOp_j},\;\link{k}{k+1}\in\beta}}\;v_{\beta\removeLink\link{k}{k+1}}.
\end{align}

Suppose first that $\link{k}{k+1}\notin\alpha$, in which case we need to show $\hat{\pi}_k(v_{\alpha})=0$.
By Lemma~\ref{lem: downwards lemma}(b), in this case also $\link{k}{k+1}\notin\alpha^{\tieOp_j}$,
so $\hat{\pi}_k(v_{\alpha^{\tieOp_j}})=0$ and \eqref{eq: first term} is zero.
Also \eqref{eq: second term}
is zero because the sum on the right hand side of
\eqref{eq: second term} is empty. Thus we indeed have $\hat{\pi}_k(v_{\alpha})=0$.

Suppose then that $\link{k}{k+1}\in\alpha$, in which case we need to
show $\hat{\pi}_k(v_{\alpha})=v_{\alpha\removeLink\link{k}{k+1}}$. 
Using the bijection $\beta\mapsto\hat{\beta}:=\beta\removeLink\link{k}{k+1}$
of Lemma~\ref{lem: downwards lemma}(c), we rewrite
the summation in \eqref{eq: second term} as
\begin{align*}
\hat{\pi}_k\Big(\sum_{\substack{\beta\in\LP_N\setminus\{\alpha,\alpha^{\tieOp_j}\}\\\beta^{\tieOp_j}\,=\,\alpha^{\tieOp_j}}}\;v_{\beta}\Big)\;=\;\sum_{\substack{\beta\in\LP_N\setminus\{\alpha,\alpha^{\tieOp_j}\}\\\beta^{\tieOp_j}\,=\,\alpha^{\tieOp_j},\;\link{k}{k+1}\in\beta}}\;v_{\beta\removeLink\link{k}{k+1}}\,=\,\sum_{\substack{\hat{\beta}\in\LP_{N-1}\setminus\{\hat{\alpha},\hat{\alpha}^{\tieOp_{j'}}\}\\\hat{\beta}^{\tieOp_{j'}}\,=\,\hat{\alpha}^{\tieOp_{j'}}}}\;v_{\hat{\beta}}.
\end{align*}
In \eqref{eq: first term}, we use
$\hat{\pi}_k(v_{\alpha^{\tieOp_j}})=v_{\hat{\alpha}^{\tieOp_{j'}}}$,
and combine to obtain
\begin{align*}
\hat{\pi}_k(v_\alpha)
=&\;\qnum 2\,\Big(\id-\pi_{j'}\Big)\,(v_{\hat{\alpha}^{\tieOp_{j'}}})\;-\;\sum_{\substack{\hat{\beta}\in\LP_{N-1}\setminus\{\hat{\alpha},\hat{\alpha}^{\tieOp_{j'}}\}\\\hat{\beta}^{\tieOp_{j'}}\,=\,\hat{\alpha}^{\tieOp_{j'}}}}\;v_{\hat{\beta}} . 
\end{align*}
By the assumed equality \eqref{eq: solution for general pattern},
this expression is $v_{\hat{\alpha}}$, 
which is what we wanted to show.
This concludes the case (iii) and completes the proof.
\end{proof}

\subsection{\label{subsec: proof of existence}Proof of Theorem~\ref{thm: existence of multiple SLE vectors}}

Uniqueness follows immediately from Proposition~\ref{prop: uniqueness}. Proposition~\ref{prop: solution for nested} gives the solution $v_{\nested_N}$
for each $N\in\bZnn$. From this solution, corresponding to the maximal
element $\nested_N$ of the partially ordered set $\LP_N$,
Proposition~\ref{prop: construction of solutions} allows us to construct
the solutions $v_\alpha$ for 
$\alpha\in\LP_N\setminus\{\nested_N\}$
recursively in any order that refines 
the partial order. $\hfill\qed$

\subsection{\label{subsec: basis of the trivial subrepresentation}Basis of the trivial subrepresentation} 

In this section, consider a fixed $N \in \bN$.
The trivial subrepresentation 
$\HWsp_1=\HWsp_1(\Wd_{2}^{\tens 2N})$,
introduced in \eqref{eq: highest vector space}, 
is exactly the solution space of 
\eqref{eq: multiple sle cartan eigenvalue}~--~\eqref{eq: multiple sle highest weight vector}.
The vectors $v_\alpha$ satisfying the projection 
conditions~\eqref{eq: multiple sle projection conditions} for 
$\alpha\in\LP_N$ in fact constitute a basis of this subspace,
as we will show below in Proposition~\ref{prop: quantum dual elements}(b).
For this purpose, we define certain elements 
$\Quantumdual_\alpha$ 
of the dual space 
$\HWsp_1^* = \set{\psi \colon \HWsp_1 \to \bC \; | \; \psi \text{ is a linear map} }$, 
which will be shown to form the dual basis of $(v_\alpha)_{\alpha\in\LP_N}$.
The definition and construction are analogous to those of a dual basis of the solution space of
\eqref{eq: multiple SLE PDEs}~--~\eqref{eq: multiple SLE Mobius covariance} in \cite{FK-solution_space_for_a_system_of_null_state_PDEs_1},
and we follow some similar terminology and notation.


Consider the link pattern 
\begin{align*}
\alpha=\set{\link{a_1}{b_1},\ldots,\link{a_N}{b_N}}\in\LP_N
\end{align*}
together with an ordering of the links:
$\link{a_1}{b_1},\ldots,\link{a_N}{b_N}$. We use the convention that $a_j < b_j$ for all $j=1,\ldots,N$, i.e.,
$a_j$ is the index of the left endpoint of the $j$:th link.
One possible choice of ordering is by the left endpoints of the links --- 
the ordering 
$\link{a^\circ_1}{b^\circ_1},\ldots,\link{a^\circ_N}{b^\circ_N}$ such that
$a^\circ_1 < a^\circ_2 < \cdots < a^\circ_N$
will be used as a reference to which other choices of orderings are compared:
for some permutation $\sigma \in \SymmGrp_N$ we have $a_j = a^\circ_{\sigma(j)}$, $b_j = b^\circ_{\sigma(j)}$
for all $j=1,\ldots,N$.

If $a_1,b_1$ are consecutive indices, $b_1=a_1+1$, then $\alpha \removeLink \link{a_1}{b_1}$
denotes the link pattern with the first link removed. The indices of the other $N-1$
links $\link{a_2}{b_2}$, \ldots, $\link{a_N}{b_N}$ are
relabeled, so that the links of $\alpha \removeLink \link{a_1}{b_1}$ are $\link{a_2(1)}{b_2(1)}$,
\ldots, $\link{a_N(1)}{b_N(1)}$, respectively.
Iteratively, if $a_{j}(j-1),b_{j}(j-1)$ are consecutive after removal of $j-1$ first links,
i.e., $b_{j}(j-1)=a_{j}(j-1)+1$, then
the $j$:th link can be removed to obtain a link pattern denoted by
$\alpha \removeLink \link{a_{1}}{b_{1}} \removeLink \cdots \removeLink \link{a_{j}}{b_{j}}$.
With relabeling the indices, as in Figure~\ref{fig: removing several links},
the remaining $N-j$ links are denoted by $\link{a_{j+1}(j)}{b_{j+1}(j)}$, \ldots, $\link{a_N(j)}{b_N(j)}$.
The ordering $\sigma$ is said to be allowable for $\alpha$ if all links of $\alpha$ can be removed in the order $\sigma$,
i.e., if we have $b_{j+1}(j)=a_{j+1}(j)+1$ for all $j < N$.


\begin{figure}
\includegraphics[scale=1.1]{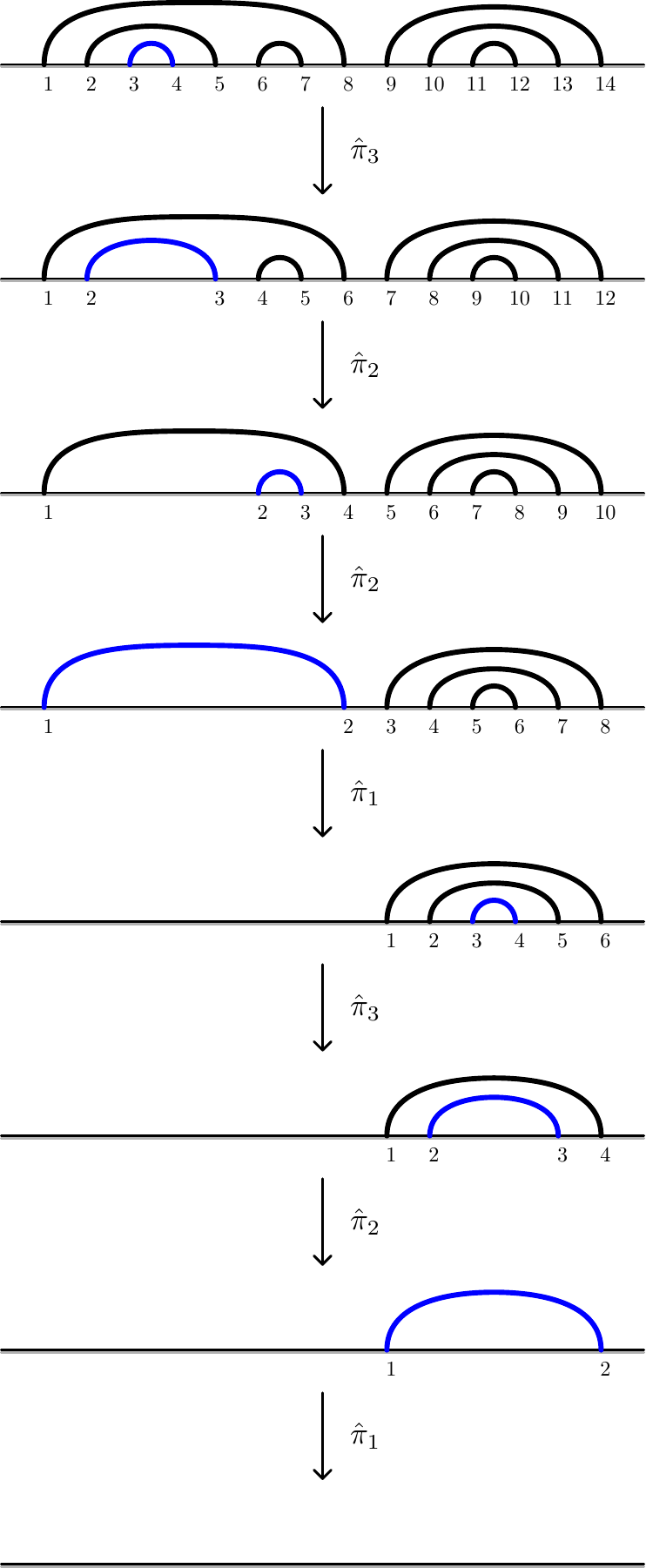}
\caption{\label{fig: removing several links}
The ordering $\set{\link{3}{4},\link{2}{5},\link{6}{7},\link{1}{8},
\link{11}{12},\link{10}{13},\link{9}{14}}$ 
is allowable for $\alpha$ as depicted in the figure. 
The figure also illustrates the iterated projections 
appearing in the definition of the map $\Quantumdual_\alpha$
in~\eqref{eq: quantum dual elements}, with relabeled indices.}
\end{figure}

Let $\alpha\in\LP_N$ and let $\sigma\in\SymmGrp_N$ be 
an allowable ordering for $\alpha$.
We define the linear map
\begin{align}
\Quantumdual_\alpha^{(\sigma)}\colon&\HWsp_1\longrightarrow\bC,\qquad
\Quantumdual_\alpha^{(\sigma)}:=\;\hat{\pi}_{a_N(N-1)} \circ \cdots \circ \hat{\pi}_{a_2(1)}
\circ\hat{\pi}_{a_{1}},
\label{eq: quantum dual elements}
\end{align}
where 
$\hat{\pi}_{a_{j}(j-1)} \colon \Wd_{2}^{\tens 2(N-j+1)}\to\Wd_{2}^{\tens 2(N-j)}$ are projections in the tensor components $a_j(j-1)$
and $a_j(j-1)+1 = b_j(j-1)$, reducing the number of tensorands by
two --- see Figure~\ref{fig: removing several links}.

We next show that $\Quantumdual_\alpha^{(\sigma)}$ 
is in fact independent of the choice of allowable ordering $\sigma$ for $\alpha$, and thus gives
rise to a well defined linear map
\begin{align}\label{eq: dual basis without ordering}
\Quantumdual_\alpha:=\Quantumdual_\alpha^{(\sigma)}
\;\colon\;\HWsp_1\longrightarrow\bC
\end{align}
for any choice of
allowable $\sigma=\sigma(\alpha)\in\SymmGrp_N$, and that
$\left(\Quantumdual_\alpha\right)_{\alpha\in\LP_N}$
is a basis of the dual space $\HWsp_1^*$.

\begin{prop}\label{prop: quantum dual elements}
\
\begin{description}
\item[(a)] Let $\alpha\in\LP_N$. For any two allowable orderings $\sigma,\sigma'\in\SymmGrp_N$ for $\alpha$, we have
\begin{align*}
\Quantumdual_\alpha^{(\sigma)}=\Quantumdual_\alpha^{(\sigma')}.
\end{align*}
Thus, the linear functional
$\Quantumdual_\alpha 
\in\HWsp_1^*$ in \eqref{eq: dual basis without ordering}
is well defined. 

\item[(b)]
For any $\alpha,\beta \in \LP_N$ we have
\begin{align}\label{eq: quantum dual basis}
\Quantumdual_\alpha(v_\beta)\;=\;
\delta_{\alpha,\beta}\;=\; & \begin{cases}
1\quad & \text{if }\beta=\alpha\\
0 & \text{if }\beta\neq\alpha
\end{cases}.
\end{align}
In particular,
$\left(v_{\alpha}\right)_{\alpha\in\LP_N}$
and $\left(\Quantumdual_{\alpha}\right)_{\alpha\in\LP_N}$ 
are bases of $\HWsp_1$ and $\HWsp_1^*$, respectively, and dual to each other.
\end{description}
\end{prop}
\begin{proof}
Let $\alpha,\beta \in \LP_N$, and let $\sigma\in\SymmGrp_N$ be 
any allowable ordering for $\alpha$.
Consider $\Quantumdual_\alpha^{(\sigma)}(v_\beta)$. If $\beta=\alpha$, then
by~\eqref{eq: multiple sle projection conditions} we have
$\hat{\pi}_{a_1}(v_\alpha) = v_{\alpha \removeLink \link{a_1}{b_1}}$,
and recursively,
\[ \Big( \hat{\pi}_{a_j(j-1)} \circ \cdots \circ \hat{\pi}_{a_1} \Big) (v_\alpha)
    = v_{\alpha \removeLink \link{a_1}{b_1} \removeLink \cdots \removeLink \link{a_j}{b_j}} . \]
For $j=N$ this gives $\Quantumdual_\alpha^{(\sigma)}(v_\alpha) = v_\emptyset = 1$.
On the other hand, if $\beta \neq \alpha$, then for some $j$ we have $\link{a_j}{b_j} \notin \beta$,
and by~\eqref{eq: multiple sle projection conditions} we then similarly get $\Quantumdual_\alpha^{(\sigma)}(v_\beta) = 0$.
Summarizing, Equation~\eqref{eq: quantum dual basis} holds: we have
$\Quantumdual_\alpha^{(\sigma)}(v_\beta) = \delta_{\alpha,\beta}$,
independently of the choice of allowable $\sigma$.
In particular, the value of the operator $\Quantumdual_\alpha^{(\sigma)}$ 
in the linear span of the vectors $v_\beta$
is independent of the choice of allowable $\sigma$.
Assertion (a) will follow by showing that this linear span is actually 
the whole space $\HWsp_1$.

In the linear span of the vectors $v_\beta$, we now set
$\Quantumdual_\alpha:=\Quantumdual_\alpha^{(\sigma)}$.
The collection
$\left(\Quantumdual_\alpha\right)_{\alpha\in\LP_N}$ 
of linear functionals on $\spn\{v_\beta\;|\;\beta\in\LP_N\}$
is linearly independent --- indeed,
any linear relation
\begin{align*}
\Quantumdual_\beta=\sum_{\alpha\in\LP_N\setminus\{\beta\}}c_\alpha\Quantumdual_\alpha
\end{align*}
would by~\eqref{eq: quantum dual basis} lead to a contradiction
\begin{align*}
1=\Quantumdual_\beta(v_\beta)=\sum_{\alpha\in\LP_N\setminus\{\beta\}}
c_\alpha \, \Quantumdual_\alpha(v_\beta)=\sum_{\alpha\in\LP_N\setminus\{\beta\}}
c_\alpha\delta_{\alpha,\beta}=0.
\end{align*}
Since we have
$\dmn(\HWsp_1)=\dmn(\HWsp_1^*)=\Catalan_N=\#\LP_N$
by Lemma~\ref{lem: multiplicity of singlet},
it follows that $\left(\Quantumdual_\alpha\right)_{\alpha\in\LP_N}$ 
is a basis of the whole dual space $\HWsp_1^*$, and 
$\left(v_{\alpha}\right)_{\alpha\in\LP_N}$ its dual basis in $\HWsp_1$.
This concludes the proof.
\end{proof}


\bigskip{}

\section{\label{sec: Multiple SLE partition functions}Multiple $\SLE$ partition functions}

In this section, we consider multiple $\SLE$ partition functions constructed
from the vectors $v_\alpha$ 
of Section~\ref{sec: Multiple SLE pure geometries},
using the correspondence of Theorem~\ref{thm: SCCG correspondence special case}.
We give in Section~\ref{sec: Multiple SLE pure partition functions} the
construction and properties of the pure
partition functions $\PartF_\alpha$, i.e., solutions to the system
\eqref{eq: multiple SLE PDEs}~--~\eqref{eq: multiple SLE asymptotics}.
Section~\ref{subsec: Dual elements} contains the only remaining part of the proof
of the properties, namely the injectivity of the correspondence of
Theorem~\ref{thm: SCCG correspondence special case}.
In Section~\ref{sec: Canonical partition functions and entire curve domain Markov property},
we consider partition functions $\PartF^{(N)}$ that are most commonly relevant for statistical
physics models --- these symmetric partition functions are combinations
of the pure partition functions.
Finally, Sections~\ref{subsub: Ising model}~--~\ref{subsub: Percolation} give examples of
symmetric partition functions applicable to a few important lattice models.
These are explicit solutions to the partial differential equations at particular values of $\kappa$,  
and as such, they are relevant to the construction of local multiple $\SLE$s by the results of
Appendix~\ref{app: Local multiple SLEs}. These values of $\kappa$ are rational,
and therefore the solutions are not strictly speaking special cases of our generic
solutions for $\kappa \in (0,8) \setminus \bQ$, but should be obtained via a suitable limiting procedure.

\subsection{\label{sec: Multiple SLE pure partition functions}Pure partition functions}

In this section, we will use the mappings 
\begin{align*}
\sF \colon \HWsp_1(\Wd_2^{\tens 2N}) \to \sC^\infty(\chamber_{2N})
\end{align*}
given by Theorem~\ref{thm: SCCG correspondence special case},
to construct the pure partition functions $(\PartF_\alpha)_{\alpha \in \LP}$
from the vectors $(v_\alpha)_{\alpha \in \LP}$ of Theorem~\ref{thm: existence of multiple SLE vectors}.
Recall from Theorem~\ref{thm: SCCG correspondence special case}
the following important properties 
for any vector $v$ in the 
trivial subrepresentation $\HWsp_1=\HWsp_1(\Wd_2^{\tens 2N})$:
\begin{itemize}
\item (PDE) ensures that the function
$\sF[v]$ is a solution to the PDEs~\eqref{eq: multiple SLE PDEs}.
\item (COV) gives the M\"obius covariance~\eqref{eq: multiple SLE Mobius covariance}.
\item (ASY) enables us to pick solutions of 
\eqref{eq: multiple SLE PDEs}~--~\eqref{eq: multiple SLE Mobius covariance}
with the desired asymptotic properties~\eqref{eq: multiple SLE asymptotics}.
\end{itemize}

We will show in Corollary~\ref{cor: injectivity} that, for all 
$N\in\bN$, the map $\sF$ is injective and thus the basis $\left(v_\alpha\right)_{\alpha\in\LP_N}$
of $\HWsp_1$ given by 
Proposition~\ref{prop: quantum dual elements}(b)
provides a basis for the solution space 
$\sF[\HWsp_1]=\sF[\HWsp_1(\Wd_2^{\tens 2N})]$ of 
the system~\eqref{eq: multiple SLE PDEs}~--~\eqref{eq: multiple SLE Mobius covariance}.
We normalize this basis by 
$B^{-|\alpha|}\,\sF[v_\alpha]$,
where $B=\frac{\Gamma(1-4/\kappa)^2}{\Gamma(2-8/\kappa)}$
--- this choice of normalization is convenient by the
(ASY) part of Theorem~\ref{thm: SCCG correspondence special case}.
Notice that with our assumption $\kappa\notin\bQ$, the normalizing 
constant $B^{-|\alpha|}$ is finite and non-zero.


\begin{thm}\label{thm: pure partition functions}
Let $\kappa \in (0,8)\setminus \bQ$.
The collection $\left(\PartF_{\alpha}\right)_{\alpha\in\LP}$ of functions 
\begin{align*}
\PartF_{\alpha}:=B^{-|\alpha|}\,\sF[v_\alpha]\;\colon\;\chamber_{2|\alpha|}\rightarrow\bC
\end{align*}
satisfies the system of equations~\eqref{eq: multiple SLE PDEs}~--~\eqref{eq: multiple SLE asymptotics} for all $\alpha\in\LP$.
%
For any $N\in\bZnn$, the collection 
$\left(\PartF_{\alpha}\right)_{\alpha\in\LP_N}$ is linearly independent 
and it spans the $\Catalan_N$-dimensional space $\sF[\HWsp_1(\Wd_2^{\tens 2N})]$
of solutions to the
system~\eqref{eq: multiple SLE PDEs}~--~\eqref{eq: multiple SLE Mobius covariance}.
\end{thm}
\begin{proof}
By Theorem~\ref{thm: SCCG correspondence special case}, since the vectors
$v_{\alpha}$ satisfy 
\eqref{eq: multiple sle cartan eigenvalue}~--~\eqref{eq: multiple sle highest weight vector}, the functions
$\PartF_{\alpha}=B^{-|\alpha|}\sF[v_\alpha]$ satisfy 
\eqref{eq: multiple SLE PDEs}~--~\eqref{eq: multiple SLE Mobius covariance}.
The asymptotic conditions~\eqref{eq: multiple SLE asymptotics} follow from the
(ASY) part of Theorem~\ref{thm: SCCG correspondence special case},
by the projection conditions~\eqref{eq: multiple sle projection conditions} 
for the vectors $v_{\alpha}$.
The final assertion of linear independence will be established in 
Proposition~\ref{prop: FK dual elements} 
in the next section.
\end{proof}

\subsection{\label{subsec: Dual elements}Linear independence of the pure partition functions}

To show that $\left(\PartF_{\alpha}\right)_{\alpha\in\LP_N}$
is a basis of the $\Catalan_N$-dimensional solution space 
$\sF[\HWsp_1]$, 
we use 
linear mappings $\FKdual_\alpha$, 
closely related to the maps $\Quantumdual_\alpha$ 
from~\eqref{eq: dual basis without ordering}, defined by iterated limits.
We show in Proposition~\ref{prop: FK dual elements}
that $\left(\FKdual_\alpha\right)_{\alpha\in\LP_N}$  
is a basis of the dual space $\sF[\HWsp_1]^*$
and that the functions $\PartF_{\alpha}=B^{-|\alpha|}\,\sF[v_\alpha]$ 
constitute its dual basis.
These iterated limits were originally
introduced in \cite{FK-solution_space_for_a_system_of_null_state_PDEs_1}, where their
well-definedness was checked with analysis techniques. With the quantum group method
of \cite{KP-conformally_covariant_boundary_correlation_functions_with_a_quantum_group}
and the solutions $(v_\alpha)_{\alpha \in \LP_{N}}$,
the well-definedness on the $\Catalan_N$-dimensional solution space $\sF[\HWsp_1]$
becomes almost immediate.



Suppose the ordering $\sigma\in\SymmGrp_N$ is allowable for $\alpha\in\LP_N$,
and let $\link{a_1}{b_1}, \ldots, \link{a_N}{b_N}$ be the corresponding links
(recall Section~\ref{subsec: basis of the trivial subrepresentation} 
and Figure~\ref{fig: removing several links}). 
By the (ASY) part of
Theorem~\ref{thm: SCCG correspondence special case}, the following
sequence of limits exists for any $\PartF=\sF[v]\in\sF[\HWsp_1]$:
\begin{align}\label{eq: limit operation}
\FKdual_\alpha^{(\sigma)} (\PartF )
:=\;
\lim_{x_{a_N},x_{b_{N}}\to\xi_{N}}
\cdots
\lim_{x_{a_{1}},x_{b_{1}}\to\xi_{1}}
(x_{b_{N}}-x_{a_{N}})^{2h}
\cdots
(x_{b_{1}}-x_{a_{1}})^{2h}
\times\PartF(x_1,\ldots,x_{2N}).
\end{align}
Note that each of the limits in \eqref{eq: limit operation}
is independent of the limit point $\xi_j$,
and \eqref{eq: limit operation} in fact equals
$B^{|\alpha|}\Quantumdual_\alpha (v)$, where 
$\Quantumdual_\alpha\colon\HWsp_1\longrightarrow\bC$ 
is the linear map introduced in~\eqref{eq: dual basis without ordering}.
In particular, the linear map 
\begin{align*}
\FKdual_\alpha:=\FKdual_\alpha^{(\sigma)}\;\colon\;\sF[\HWsp_1]\rightarrow\bC
\end{align*}
is well defined via Equation~\eqref{eq: limit operation},
independently of the choice of allowable ordering $\sigma$ for $\alpha$.

\begin{prop}\label{prop: FK dual elements}
The collection
$\left(\FKdual_\alpha\right)_{\alpha\in\LP_N}$ is a
basis of the dual space $\sF[\HWsp_1]^*$, 
and we have
\begin{align*}
\FKdual_\alpha ( \PartF_\beta ) \;=\;\delta_{\alpha,\beta}\;=\; & \begin{cases}
1\quad & \text{if } \beta = \alpha \\
0 & \text{if }\beta\neq\alpha
\end{cases}.
\end{align*}
In particular, $\left(\PartF_{\alpha}\right)_{\alpha\in\LP_N}$ is a basis of
the $\Catalan_N$-dimensional solution space $\sF[\HWsp_1]$. 
\end{prop}
In \cite{FK-solution_space_for_a_system_of_null_state_PDEs_4},
``connectivity weights'' are defined as the dual basis of the iterated limits
$\FKdual_\alpha$, so by this proposition they coincide with our pure partition functions.
Explicit expressions for the connectivity weights for $N=2,3,4$ were studied further in
\cite{FSK-multiple_SLE_connectivity_weights_for_4_6_8}, and a formula
\cite[Equation~(56)]{FSK-multiple_SLE_connectivity_weights_for_4_6_8} 
for the connectivity
weight of the rainbow pattern (see Proposition~\ref{prop: solution for nested}) was obtained for general $N$.
\begin{proof}[Proof of Proposition~\ref{prop: FK dual elements}]
The assertion follows directly from 
Proposition~\ref{prop: quantum dual elements},
because we have $\PartF_{\alpha}=B^{-|\alpha|}\,\sF[v_\alpha]$ by definition,
and hence, 
$\FKdual_\alpha (\PartF_\beta )
= B^{|\alpha|}\,\Quantumdual_\alpha ( B^{-|\beta|}v_\beta )
= \delta_{\alpha,\beta}$. 
\end{proof}


As a corollary, we get the injectivity of the 
``spin chain~--~Coulomb gas correspondence'' $\sF$.
\begin{cor}\label{cor: injectivity}
For any $N\in\bZpos$, the mapping 
$\sF\colon\HWsp_1(\Wd_2^{\tens 2N})\rightarrow\sC^{\infty}(\chamber_{2N})$ 
is injective.
\end{cor}
\begin{proof}
By Proposition~\ref{prop: FK dual elements}, the images 
$\sF[v_\alpha]$ of the basis vectors 
$v_\alpha$ of $\HWsp_1(\Wd_2^{\tens 2N})$ are 
linearly independent, because $\PartF_\alpha=B^{-|\alpha|}\,\sF[v_\alpha]$ are.
Injectivity of $\sF$ follows by linearity.
\end{proof}

\subsection{\label{sec: Canonical partition functions and entire curve domain Markov property}Symmetric partition functions and entire curve domain Markov property}

In this section, we study partition functions relevant for
statistical mechanics models admitting a cyclic permutation symmetry
of the marked points on the boundary. Examples of such models are critical
percolation, the Ising model at criticality, and the discrete Gaussian
free field, with suitable boundary conditions.

The symmetric partition functions $\PartF^{(N)}$ are 
combinations of the extremal, pure partition
functions $\PartF_\alpha$, and this combination encodes information about the 
crossing probabilities of the model. 
Formulas for $\PartF^{(N)}$ are in fact often easier to find
than those for $\PartF_\alpha$.
Explicit formulas for 
crossing probabilities, however, require the knowledge of
the pure partition functions as well.

The functions $\PartF^{(N)}$ 
should satisfy the conditions of Theorem~\ref{thm: local multiple SLEs}(a),
in particular, Equations \eqref{eq: multiple SLE PDEs}~and
\eqref{eq: multiple SLE Mobius covariance}.
The asymptotics requirement~\eqref{eq: multiple SLE asymptotics}
is replaced by the cascade property
\begin{align}\label{eq: cascade property for canonical partf}
\lim_{x_j , x_{j+1} \to \xi} \frac{\PartF^{(N)}(x_1 , \ldots , x_{2N})}{(x_{j+1} - x_j)^{-2h}}
= \; & \PartF^{(N-1)}(x_1 , \ldots, x_{j-1} , x_{j+2} , \ldots , x_{2N}) 
\end{align}
for any $j = 1, \ldots, 2N-1$.
This expresses the fact that the partition function of the model with $2N$ boundary changes
reduces to that with $2N-2$ boundary changes as any two marked points are merged. In view of
Proposition~\ref{prop: cascade relation}, 
in this limit the other curves do not feel
the merged marked points, and they have the law of the symmetric $(N-1)$-$\SLE$.
Roughly, this should be interpreted as a domain Markov property with respect to
one entire curve of the multiple $\SLE$.

By our correspondence, Theorem~\ref{thm: SCCG correspondence special case}, such symmetric
partition functions $\PartF^{(N)}$ can be constructed as $\PartF^{(N)}\propto\sF[v^{(N)}]$ from vectors
$v^{(N)}\in\Wd_{2}^{\tens 2N}$ satisfying the following system of equations:
\begin{align}
&K.v^{(N)}=v^{(N)}\label{eq: canonical partf cartan eigenvalue}\\
&E.v^{(N)}=0\label{eq: canonical partf highest weight vector}\\
&\hat{\pi}_j(v^{(N)})=v^{(N-1)}\qquad\text{for all}\;j=1,\ldots,2N-1.\label{eq: canonical partf projection conditions}
\end{align}

In the quantum group setting, we have a unique solution for this system when 
the normalization is fixed.

\begin{thm}\label{thm: vectors for canonical partf}
There exists a unique collection $\left(v^{(N)}\right)_{N\in\bZnn}$ of 
vectors in $\Wd_{2}^{\tens 2N}$ 
such that $v^{(0)}=1$ and the system of 
equations~\eqref{eq: canonical partf cartan eigenvalue}~--~\eqref{eq: canonical partf projection conditions} 
hold for all $N\in\bZpos$.
The vectors are given by
\begin{align*}
v^{(N)}=\sum_{\alpha\in\LP_N}v_\alpha.
\end{align*}
\end{thm}
\begin{proof}
Applying Corollary \ref{cor: all projections vanish gives zero} 
to the difference of two solutions gives uniqueness,
as in the proof of Proposition~\ref{prop: uniqueness}. 
It remains to check that the asserted formula satisfies 
\eqref{eq: canonical partf cartan eigenvalue}~--~\eqref{eq: canonical partf projection conditions}. 
Equations~\eqref{eq: canonical partf cartan eigenvalue}~--~\eqref{eq: canonical partf highest weight vector} are satisfied by
\eqref{eq: multiple sle cartan eigenvalue}~--~\eqref{eq: multiple sle highest weight vector}.
For \eqref{eq: canonical partf projection conditions}, 
we use the properties 
\eqref{eq: multiple sle projection conditions} 
of the vectors $v_\alpha$, and the bijection
$\set{\alpha\in\LP_N\;\big|\;\link{j}{j+1}\in\alpha}\rightarrow\LP_{N-1}$
defined by $\alpha\mapsto\hat{\alpha}=\alpha\removeLink\link{j}{j+1}$,
to obtain
\begin{align*}
\hat{\pi}_j(v^{(N)})\;=\;\sum_{\alpha\in\LP_N}\hat{\pi}_j(v_\alpha)\;=\;\sum_{\substack{\alpha\in\LP_N\\\link{j}{j+1}\in\alpha}}v_{\alpha\removeLink\link{j}{j+1}}\;=\;
\sum_{\hat{\alpha}\in\LP_{N-1}}v_{\hat{\alpha}}\;=\;v^{(N-1)}
\end{align*}
for any $j=1,\ldots,2N-1$. This concludes the proof.
\end{proof}

\begin{thm}\label{thm: canonical partf}
Let $\kappa \in (0,8)\setminus \bQ$.
The collection $\left(\PartF^{(N)}\right)_{N\in\bZnn}$
of functions
\begin{align*}
\PartF^{(N)}:=\sum_{\alpha\in\LP_N}\PartF_\alpha\;\colon\;\chamber_{2N}\rightarrow\bC
\end{align*}
satisfies the partial differential equations~\eqref{eq: multiple SLE PDEs},
covariance~\eqref{eq: multiple SLE Mobius covariance},
and cascade property~\eqref{eq: cascade property for canonical partf}.
Moreover, it is uniquely determined by these conditions,
the normalization $\PartF^{(0)} = 1$, 
and the property that for each $N\in\bZnn$, the function $\PartF^{(N)}$
lies in the $\Catalan_N$-dimensional solution space.
\end{thm}
\begin{proof}
By definition, $\PartF^{(N)}=B^{-N}\,\sF[v^{(N)}]$, so 
\eqref{eq: multiple SLE PDEs}~--~\eqref{eq: multiple SLE Mobius covariance}
follow from the (PDE) and (COV) parts of 
Theorem~\ref{thm: SCCG correspondence special case} and the properties
\eqref{eq: canonical partf cartan eigenvalue}~--~\eqref{eq: canonical partf highest weight vector} 
of the vectors $v^{(N)}$. The cascade 
property~\eqref{eq: cascade property for canonical partf}
follows from the (ASY)
part of Theorem~\ref{thm: SCCG correspondence special case}
and the corresponding property
\eqref{eq: canonical partf projection conditions} for 
$\left(v^{(N)}\right)_{N\in\bZnn}$.
Uniqueness follows from the uniqueness in Theorem~\ref{thm: vectors for canonical partf}
and injectivity of $\sF$ given by Corollary~\ref{cor: injectivity}.
\end{proof}


\subsection{\label{subsub: Ising model}Example: Symmetric partition function for the Ising model}
The Ising model was initially introduced as a model of ferromagnetic material,
but its simple and generic interactions make it applicable to a variety of 
phenomena that have positive correlations. 
The two-dimensional Ising model has remarkably
subtle behavior at the critical point, where a transition from ferromagnetic
to paramagnetic phase takes place \cite{MW-2d_Ising_model}. 
Critical Ising model in particular displays
conformal invariance properties in the scaling limit
\cite{Smirnov-towards_conformal_invariance,HS-energy_density,
CHI-conformal_invariance_of_spin_correlations_in_the_planar_Ising_model}.
In recent research, the Ising
model has been studied in terms of interfaces, 
whose scaling limits at criticality
are $\SLE$ type curves with $\kappa=3$
\cite{HK-Ising_interfaces_and_free_boundary_conditions,
CDHKS-convergence_of_Ising_interfaces_to_SLE,
Izyurov-critical_Ising_interfaces_in_multiply_connected_domains}.

Alternating boundary conditions between $2N$ marked points on the
boundary give rise to $N$ random interfaces (see Figure~\ref{fig: Ising connectivities N equals 3}).
The partition function of the critical Ising model with such boundary conditions in the
upper half-plane $\bH = \set{ z \in \bC \; \big| \; \im(z) > 0}$ is the following Pfaffian expression
\begin{align}\label{eq: Ising model partf}
\PartF^{(N)}_{{\Ising}}(x_{1},\ldots,x_{2N}) = 
\sum_\sP 
\sgn(\sP) \Big(\prod_{\set{a,b} \in \sP} \frac{1}{x_{b}-x_{a}} \Big) ,
\end{align}
where the sum is over partitions $\sP = \set{\set{a_1,b_1} , \ldots, \set{a_N, b_N}}$
of the set $\set{1,\ldots,2N}$ into $N$ disjoint two-element subsets $\set{a_k,b_k} \subset \set{1,\ldots,2N}$,
and $\sgn(\sP)$ is the sign of the pair partition
$\sP$ defined as the sign of the product $\prod (a-c) (a-d) (b-c) (b-d)$ over pairs of distinct elements $\set{a,b},\set{c,d} \in \sP$,
and by convention we always use the choice $a<b$ in the products.

We now show that the above Pfaffians
are symmetric partition functions for multiple $\SLE$s
with $\kappa=3$.
\begin{prop}\label{prop: Ising symmetric partition functions}
The functions  $\left(\PartF^{(N)}_{{\Ising}}\right)_{N\in\bZnn}$ 
satisfy \eqref{eq: multiple SLE PDEs}, \eqref{eq: multiple SLE Mobius covariance}
and \eqref{eq: cascade property for canonical partf}, with $\kappa=3$
and $h = \frac{1}{2}$.
\end{prop}
The verification of 
the M\"obius covariance property~\eqref{eq: multiple SLE Mobius covariance} 
on the upper half-plane $\bH$
relies on the following lemma. It will also be used
later on for explicit partition functions for other models.

\begin{lem}\label{lem: Mobius ratio lemma}
Suppose that $\Mob \colon \bH \to \bH$ is conformal.
Then for any $z,w \in \cl{\bH}$ we have
\[ \frac{\Mob(z) - \Mob(w)}{z-w} = \sqrt{\mu'(z)} \sqrt{\mu'(w)} . \]
\end{lem}
\begin{proof}
The conformal self-map $\Mob$ of $\bH$ is a M\"obius transformation, 
$\Mob(z) = \frac{a z + b}{c z + d}$, with $a,b,c,d \in \bR$, and $ad-bc > 0$.
Without loss of generality we take $ad-bc=1$. Then a branch of the square root of the derivative
is defined by $\sqrt{\Mob'(z)} = \frac{1}{c z + d}$ (and the assertion does not depend on the choice of branch).
By a direct calculation, we get that both $\frac{\Mob(z) - \Mob(w)}{z-w}$ and $\sqrt{\mu'(z)} \sqrt{\mu'(w)}$
are equal to $\big((cz+d) (cw+d) \big)^{-1}$.
\end{proof}

\begin{proof}[Proof of Proposition~\ref{prop: Ising symmetric partition functions}]
For $\Mob\colon \bH \to \bH$, Lemma~\ref{lem: Mobius ratio lemma} gives for each term of \eqref{eq: Ising model partf} the equality
\begin{align*}
\prod_{\set{a,b} \in \sP} \frac{1}{x_b - x_a} \; = \; \prod_{i=1}^{2N} \Mob'(x_i)^{1/2} \times \prod_{\set{a,b} \in \sP} \frac{1}{\Mob(x_b) - \Mob(x_a)} ,
\end{align*}
which implies \eqref{eq: multiple SLE Mobius covariance}, i.e., 
$\PartF^{(N)}_{\Ising}(x_1, \ldots,x_{2N}) 
    = \prod_{i=1}^{2N} \Mob'(x_i)^{1/2} \times \PartF^{(N)}_{\Ising}\big(\Mob(x_1), \ldots,\Mob(x_{2N})\big) $.

For the cascade property \eqref{eq: cascade property for canonical partf},
consider the limit $x_j,x_{j+1} \to \xi$ of
$(x_{j+1}-x_j) \times \PartF^{(N)}_{\Ising}$. 
The prefactor $(x_{j+1}-x_j)$ ensures that the terms in 
\eqref{eq: Ising model partf} corresponding to pair partitions $\sP$
that do not contain the pair $\set{j,j+1}$ vanish in the limit.
For the terms for which the pair partition $\sP$ contains $\set{j,j+1}$,
we note that the prefactor $(x_{j+1}-x_j)$ cancels the factor
$\frac{1}{x_{j+1}-x_j}$, and that removing the pair $\set{j,j+1}$ from $\sP$
does not affect $\sgn(\sP)$. We get the desired property
\begin{align*}
\lim_{x_{j},x_{j+1}\to\xi}(x_{j+1}-x_{j})\times\PartF^{(N)}_{{\Ising}}(x_{1},\ldots,x_{2N})=\; & \PartF^{(N-1)}_{{\Ising}}(x_{1},\ldots,x_{j-1},x_{j+2},\ldots,x_{2N}).
\end{align*}

It remains to show that $\PartF^{(N)}_{\Ising}$ satisfies the
partial differential equations \eqref{eq: multiple SLE PDEs},
with $\kappa=3$ and $h=\frac{1}{2}$,
i.e., that
$\sD_2^{(i)}=\frac{3}{2}\pdder{x_{i}}+\sum_{j\neq i}\left(\frac{2}{x_{j}-x_{i}}\pder{x_{j}}-\frac{1}{(x_{j}-x_{i})^{2}}\right)$
annihilate $\PartF^{(N)}_{\Ising}$. 
By symmetry, it suffices to consider $i=1$.
The action of $\sD_2^{(1)}$ on $\PartF^{(N)}_{{\Ising}}$ gives
\begin{align*}
\sum_{\sP}\sgn(\sP)\left\lbrace\frac{3}{(x_1-x_{1'})^2}-
\sum_{j\neq1}\Bigg(\frac{2}{(x_j-x_1)(x_j-x_{j'})}+\frac{1}{(x_j-x_1)^2}
\Bigg)\right\rbrace
\Bigg(\prod_{\set{a,b}\in\sP}\frac{1}{x_{b}-x_{a}}\Bigg),
\end{align*} 
where we denote by $j'=j'(\sP)$ the pair of $j$ in the pair partition $\sP$.
For fixed $\sP$, the term $\frac{3}{(x_1-x_{1'})^2}$
cancels with the term $j=1'$ in the sum.
The other terms in the sum over $j \neq 1$ can be combined pairwise
according to the pairs $\set{c,d}\in\sP\setminus\set{\set{1,1'}}$ as 
\begin{align*}
\Bigg(\frac{2}{(x_c-x_1)(x_c-x_{d})}+\frac{1}{(x_c-x_1)^2}
\Bigg)+\Bigg(\frac{2}{(x_{d}-x_1)(x_{d}-x_{c})}+\frac{1}{(x_{d}-x_1)^2}\Bigg)
=\;\frac{(x_c-x_{d})^2}{(x_c-x_1)^2(x_{d}-x_1)^2}.
\end{align*}
Thus, the claim that $\sD_2^{(1)}\PartF^{(N)}_{{\Ising}}=0$ 
is reduced to the claim that the rational function
\begin{align}\label{eq: rational function}
Q(x_1,\ldots,x_{2N})\;=\;&\sum_{\sP}\sgn(\sP)\Bigg(\prod_{\set{a,b}\in\sP}\frac{1}{x_{b}-x_{a}}\Bigg)
\sum_{\set{c,d}\in\sP\setminus\{\set{1,1'}\}}\frac{(x_c-x_{d})^2}{(x_c-x_1)^2(x_{d}-x_1)^2}
\end{align} 
is identically zero. 
To show this, we proceed by induction on $N$.
For $N=1$, the sum over $\set{j,j'}$ is empty and therefore zero. 
Assume then that $Q(y_1,\ldots,y_{2N-2})\equiv 0$, 
and consider $Q(x_1,\ldots,x_{2N})$. 

We decompose the sum~\eqref{eq: rational function} 
into two sums according to whether 
$\sP$ contains the pair $\set{1,2N}$ or not.

If $\sP$ contains $\set{1,2N}$, we remove it 
and denote by $\sP_\bullet = \sP \setminus \set{\set{1,2N}}$ the resulting pair partition
of $\set{2,\ldots,2N-1}$. Note that $\sgn(\sP_\bullet) = \sgn(\sP)$.
We also extract the corresponding term $\frac{1}{x_{2N}-x_1}$ from 
the product in~\eqref{eq: rational function}.
The sum of the terms in~\eqref{eq: rational function} for which 
$\sP$ contains $\set{1,2N}$ can thus be written as
\begin{align}\label{eq: terms containing (i, 2N)}
\frac{1}{x_{2N}-x_1}\;\sum_{\sP_\bullet}\sgn(\sP_\bullet)\Bigg(\prod_{\set{a,b}\in\sP_\bullet}\frac{1}{x_{b}-x_{a}}\Bigg)
\sum_{\set{c,d}\in\sP_\bullet}\frac{(x_c-x_{d})^2}{(x_c-x_1)^2(x_{d}-x_1)^2} .
\end{align}

If $\sP$ does not contain $\set{1,2N}$, then we have $\set{p,2N}\in\sP$ for some
$p=2,\ldots,2N-1$. 
Denote by $\sP_p = \sP \setminus \set{\set{p,2N}}$ the pair partition of
$\set{1,\ldots,2N-1} \setminus \set{p}$ obtained by removing this pair.
Note that $\sgn(\sP_p) = (-1)^{p-1} \, \sgn(\sP)$.
For each $p$, we write the sum in~\eqref{eq: rational function} over terms $\set{c,d}\neq\set{p,2N}$ as
\begin{align}\label{eq: induction hypothesis}
\frac{(-1)^{p-1}}{x_{2N}-x_p} \; \sum_{\sP_p} \sgn(\sP_p)
\Bigg(\prod_{\set{a,b}\in\sP_p}\frac{1}{x_{b}-x_{a}}\Bigg)
\sum_{\set{c,d}\in\sP_p\setminus\{\set{1,1'}\}}\frac{(x_c-x_{d})^2}{(x_c-x_1)^2(x_{d}-x_1)^2}.
\end{align}
By the induction hypothesis, 
the expression~\eqref{eq: induction hypothesis} is zero.
The remaining terms in~\eqref{eq: rational function} 
have $\set{c,d}=\set{p,2N}$, and they add up to
\begin{align}\label{eq: terms not containing (i, 2N)}
\sum_{p=2}^{2N-1} \frac{(-1)^{p-1}(x_{2N}-x_p)}{(x_p-x_1)^2(x_{2N}-x_1)^2} \;
\sum_{\sP_p}\sgn(\sP_p)\Bigg(\prod_{\set{a,b}\in\sP_p}\frac{1}{x_{b}-x_{a}}\Bigg).
\end{align}
We will finish the proof by showing that~\eqref{eq: terms not containing (i, 2N)} 
cancels~\eqref{eq: terms containing (i, 2N)}. 

For each $p=2,\ldots,2N-1$ there is a bijection $\sP_\bullet \mapsto \sP_p$
from the set of pair partitions of $\set{2,\ldots,2N-1}$ to those of $\set{1,\ldots,2N-1}\setminus\set{p}$
obtained by replacing the pair $\set{p,p'}\in\sP_\bullet$ by 
the pair $\set{1,p'}$. 
Note that $\sgn(\sP_p) = (-1)^p \, \sgn(p'-p) \,\sgn(\sP_\bullet)$.

We can now write \eqref{eq: terms not containing (i, 2N)} as follows (recall that in the products we choose $a<b$):
\begin{align*}
&\sum_{p=2}^{2N-1}
\frac{(-1)^{p-1} \, (x_{2N}-x_p)}{(x_p-x_1)^2 \, (x_{2N}-x_1)^2}
\sum_{\sP_\bullet} (-1)^p \, \sgn(p'-p) \, \sgn(\sP_\bullet)
\Bigg(\prod_{\set{a,b}\in\sP_\bullet}\frac{1}{x_{b}-x_{a}}\Bigg)
\frac{(x_{p'}-x_p) \times \sgn(p'-p)}{x_{p'}-x_1} \\
=&\;\frac{-1}{(x_{2N}-x_1)^2} \; \sum_{\sP_\bullet} \sgn(\sP_\bullet)
\Bigg(\prod_{\set{a,b}\in\sP_\bullet}\frac{1}{x_{b}-x_{a}}\Bigg)
\sum_{p=2}^{2N-1} \frac{(x_{2N}-x_p) \, (x_{p'}-x_{p})}{(x_p-x_1)^2 \, (x_{p'}-x_{1})} .
\end{align*}
We combine the terms $p=c$ and $p=d$, for $\set{c,d}\in\sP_\bullet$,
to simplify the last sum over $p$ as
\begin{align*}
&\sum_{p=2}^{2N-1}
\frac{(x_{2N}-x_p) \, (x_{p'}-x_{p})}{(x_p-x_1)^2 \, (x_{p'}-x_{1})}
=(x_{2N}-x_1)\times\sum_{\set{c,d} \in \sP_\bullet}
\frac{(x_c-x_{d})^2}{(x_{c}-x_1)^2(x_{d}-x_1)^2} .
\end{align*}
Plugging this in the previous formula, we see that \eqref{eq: terms not containing (i, 2N)}
equals $-1$ times \eqref{eq: terms containing (i, 2N)}.
This shows that $\PartF^{(N)}_{{\Ising}}$ satisfies the
PDEs~\eqref{eq: multiple SLE PDEs}, 
and concludes the proof.
\end{proof}

\subsection{\label{subsub: Discrete Gaussian free field}Example: Symmetric partition function for the Gaussian free field}
The Gaussian free field (GFF) is the probabilistic equivalent of free massless
boson in quantum field theory. It is defined, roughly, as the Gaussian process in
the domain, whose mean is the harmonic interpolation of the boundary values of the field,
and whose covariance is 
Green's function for the Laplacian. For more details, see
\cite{Sheffield-GFF_for_mathematicians,Werner-topics_on_GFF}.
The level lines of the GFF, appropriately defined, are $\SLE$ type
curves with $\kappa=4$, see
\cite{Dubedat-SLE_and_free_field_partition_functions_and_couplings, MS-imaginary_geometry_1,
SS-contour_line_of_continuum_GFF, IK-Hadamards_formula_and_couplings_of_SLE_with_GFF}.
The corresponding lattice model is the discrete Gaussian free field, and
its level line converges to $\SLE_4$ in the scaling limit \cite{SS-contour_lines_of_discrete_GFF}.

The level lines of the discrete GFF 
in fact tend to discontinuity lines
of the continuum GFF, with a specific discontinuity $2\lambda$,
see \cite{SS-contour_lines_of_discrete_GFF} for details.
Very natural boundary conditions for the GFF, which give rise to $N$ such curves, are obtained by
alternating $+\lambda$ and $-\lambda$ on boundary segments between $2N$ marked points.

The (regularized) GFF partition function with these boundary 
conditions in the upper half-plane $\bH$ is
\begin{align*}
\PartF^{(N)}_{{\GFF}}(x_{1},\ldots,x_{2N})=
    \prod_{1\leq k<l\leq 2N}(x_l-x_k)^{\frac{1}{2} (-1)^{l-k}} .
\end{align*}
This formula also appears in~\cite{KW-boundary_partitions_in_trees_and_dimers}
as a scaling limit of double-dimer partition functions.
We show below that these are symmetric partition functions for multiple $\SLE$s
with $\kappa=4$.
\begin{prop}
The functions $\left(\PartF^{(N)}_{{\GFF}}\right)_{N\in\bZnn}$  
satisfy \eqref{eq: multiple SLE PDEs}, \eqref{eq: multiple SLE Mobius covariance}
and \eqref{eq: cascade property for canonical partf}, with $\kappa=4$
and $h = \frac{1}{4}$.
\end{prop}
\begin{proof}
The M\"obius covariance property \eqref{eq: multiple SLE Mobius covariance} of
$\PartF^{(N)}_{\GFF}$ is shown using Lemma~\ref{lem: Mobius ratio lemma} ---
we calculate
\[ \frac{\PartF^{(N)}_{\GFF}(\Mob(x_1), \ldots,\Mob(x_{2N}))}{\PartF_{\GFF}^{(N)}(x_1, \ldots,x_{2N})}
    = \prod_{1 \leq k < l \leq 2N} \Big(\frac{\Mob(x_l) - \Mob(x_k)}{x_l-x_k}\Big)^{\frac{1}{2}(-1)^{l-k}}
    = \prod_{1 \leq k < l \leq 2N} \big( \Mob'(x_l) \Mob'(x_k) \big)^{\frac{1}{4}(-1)^{l-k}} . \]
For each $j=1,\ldots,2N$, the product
has $2N-1$ factors which contain the variable $x_j$, of which $N$ are raised
to power $-\frac{1}{4}$ and $N-1$ to power $+\frac{1}{4}$, so the correct factor
$\Mob'(x_j)^{-1/4}$ 
remains after cancellations.

For the cascade property \eqref{eq: cascade property for canonical partf}, 
consider the limit $x_j,x_{j+1} \to \xi$ of
$(x_{j+1}-x_j)^{\frac{1}{2}} \times \PartF^{(N)}_{\GFF}$. 
The prefactor $(x_{j+1}-x_j)^{\frac{1}{2}}$ directly cancels one factor in the product, and 
the factors ${|x_{j}-x_i|}^{\frac{1}{2} (-1)^{j-i}}$ and ${|x_{j+1}-x_i|}^{\frac{1}{2} (-1)^{j+1-i}}$
cancel in the limit. We get the desired property
\begin{align*}
\lim_{x_{j},x_{j+1}\to\xi}(x_{j+1}-x_{j})^{\frac{1}{2}}\times\PartF^{(N)}_{{\GFF}}(x_{1},\ldots,x_{2N})=\; & \PartF^{(N-1)}_{{\GFF}}(x_{1},\ldots,x_{j-1},x_{j+2},\ldots,x_{2N}).
\end{align*}

It remains to show that $\PartF^{(N)}_{\GFF}$ satisfies the
partial differential equations \eqref{eq: multiple SLE PDEs},
with $\kappa=4$ and $h=\frac{1}{4}$,
i.e., that $2\pdder{x_{i}}+\sum_{j\neq i}\left(\frac{2}{x_{j}-x_{i}}\pder{x_{j}}-\frac{1}{2(x_{j}-x_{i})^{2}}\right)$
annihilate $\PartF^{(N)}_{\GFF}$. The terms with derivatives read
\begin{align}\label{eq: GFF second derivatives}
2 \frac{\pdder{x_{i}} \PartF^{(N)}_{\GFF}}{\PartF^{(N)}_{\GFF}}
    = \; & \sum_{j \neq i} \frac{\frac{1}{2}-(-1)^{j-i}}{(x_j-x_i)^2} + \sum_{\substack{j \neq i, \\ k \neq i,j}} \frac{\frac{1}{2} (-1)^{k-j}}{(x_j-x_i)(x_k-x_i)} ,\\
\label{eq: GFF first derivatives}
\sum_{j \neq i} \frac{2}{x_j - x_i} \frac{\pder{x_{j}} \PartF^{(N)}_{\GFF}}{\PartF^{(N)}_{\GFF}}
    = \; & \sum_{\substack{j \neq i, \\ k \neq j}} \frac{(-1)^{k-j}}{(x_j-x_i)(x_j-x_k)} .
\end{align}
The first term of \eqref{eq: GFF second derivatives} is canceled by the
case $k=i$ in \eqref{eq: GFF first derivatives} together with the
term $\sum_{j \neq i} \frac{-1}{2(x_{j}-x_{i})^{2}}$ without derivatives.
In the case $k \neq i$ in \eqref{eq: GFF first derivatives}, combine
the terms where $j$ and $k$ are interchanged, as
\begin{align*}
\frac{1}{(x_j-x_i)(x_j-x_k)} + \frac{1}{(x_k-x_i)(x_k-x_j)} = \frac{-1}{(x_j - x_i) (x_k - x_i)} .
\end{align*}
These exactly cancel the second term of \eqref{eq: GFF second derivatives}. This concludes the proof.
%
\end{proof}

\subsection{\label{subsub: Percolation}Example: Symmetric partition function for percolation}
Percolation is a simple model of statistical mechanics, where
different spacial locations are declared open or closed, independently, and one studies
connectivity along open locations, see e.g. \cite{Grimmett-percolation}.
There is a phase transition in the parameter $p$ which determines the probability
for locations to be declared open. At the critical
point $p=p_c$ where the phase transition takes place, a conformally invariant scaling limit
is expected. Conformal invariance of crossing probabilities, 
originally predicted by the
celebrated Cardy's formula \cite{Cardy-conformal_invariance_and_statistical_mechanics},
was established in \cite{Smirnov-critical_percolation}
for critical site percolation on the triangular lattice.

Connectivities can be formulated in terms of an exploration process, which is a
curve bounding a connected component of open locations \cite{Schramm-LERW_and_UST}.
At criticality, this curve should tend to $\SLEk$ with $\kappa=6$. This was also
proven for the triangular lattice site percolation in
\cite{Smirnov-critical_percolation, CN-critical_percolation_exploration_path}.

The partition function for percolation, 
even with boundary conditions, is trivial,
\begin{align*}
\PartF^{(N)}_{{\perco}}(x_{1},\ldots,x_{2N}) = 1 .
\end{align*}
These constant functions are also symmetric partition
functions for multiple $\SLE$s with $\kappa=6$.
\begin{prop}
The functions $\left(\PartF^{(N)}_{{\perco}}\right)_{N\in\bZnn}$ 
satisfy \eqref{eq: multiple SLE PDEs}, \eqref{eq: multiple SLE Mobius covariance}
and \eqref{eq: cascade property for canonical partf}, with $\kappa=6$
and  $h = 0$.
\end{prop}
\begin{proof}
All of the asserted properties are very easy to check.
\end{proof}
Despite the fact that the symmetric partition functions are trivial (constant functions),
there are interesting and difficult questions about the partition functions of multiple $\SLE$s
at $\kappa=6$ which are relevant for percolation. For the case $N=2$, the pure partition
functions are given by Cardy's formula, and for higher $N$ they encode more general
and complicated crossing probabilities
\cite{Dubedat-Euler_integrals, FZS-percolation_crossing_probabilities_in_hexagons}.

\bigskip{}

\section{\label{sec: SLE boundary visits}$\SLE$ boundary visits}

In this section, we show how the results of the previous sections can be used to
construct solutions to another problem, related to chordal $\SLE$ boundary visit amplitudes,
considered in \cite{JJK-SLE_boundary_visits}. The boundary visit amplitudes
are functions $\Ampl_\omega$, indexed by the order $\omega$ of visits,
and these functions are constructed by the
``spin chain~--~Coulomb gas correspondence'' from vectors $\Bdryvec_\omega$,
in a manner similar to how the pure partition functions $\PartF_\alpha$ are
constructed from the vectors $v_\alpha$ in Section~\ref{sec: Multiple SLE partition functions}.
The desired properties of the functions $\Ampl_\omega$
(given in Figures~\ref{fig: boundary visits}~--~\ref{fig: first visited point})
follow by requiring certain properties of the
vectors $\Bdryvec_\omega$ --- see
\cite{KP-conformally_covariant_boundary_correlation_functions_with_a_quantum_group,
JJK-SLE_boundary_visits}
for details.
The properties required of the vectors $\Bdryvec_\omega$ are given below
in \eqref{eq: bdry visit cartan eigenvalue}~--~\eqref{eq: bdry visit doublet projection conditions}.
They are known to uniquely specify 
$\Bdryvec_\omega$. The solution
to these properties, however, has not previously been shown to exist in general.
The main result of this section is a constructive proof of existence,
starting from the solution to the multiple $\SLE$ pure partition function problem.

\subsection{\label{subsec: Quantum group solution}Quantum group solution for the boundary visit amplitudes}

For the total number $N'~\in~\bN$ 
of points to be visited by the chordal $\SLE$,
an order of visits is a sequence
\[ \omega = (\omega_1,\ldots,\omega_{N'}) \in \set{-,+}^{N'} \]
of $N'$ $\pm$-symbols, where the 
symbol $\omega_j = -$ or $\omega_j = +$ indicates that
the $j$:th point to be visited is on 
the left or right of the starting point, respectively.
Denote by $L = L(\omega) = \#\set{j \; \big| \; \omega_j = -}$
and $R = R(\omega) = \#\set{j \; \big| \; \omega_j = +}$
the total numbers of visits on the left and right. 
The set of all orders of visits to a fixed number $N'$ of points is denoted by
$\Orders_{N'} = \set{-,+}^{N'}$, and the set of all visit orders with any number
of points by $\Orders = \bigsqcup_{N'\in\bN} \Orders_{N'}$.

\begin{figure}
\includegraphics[scale=.7]{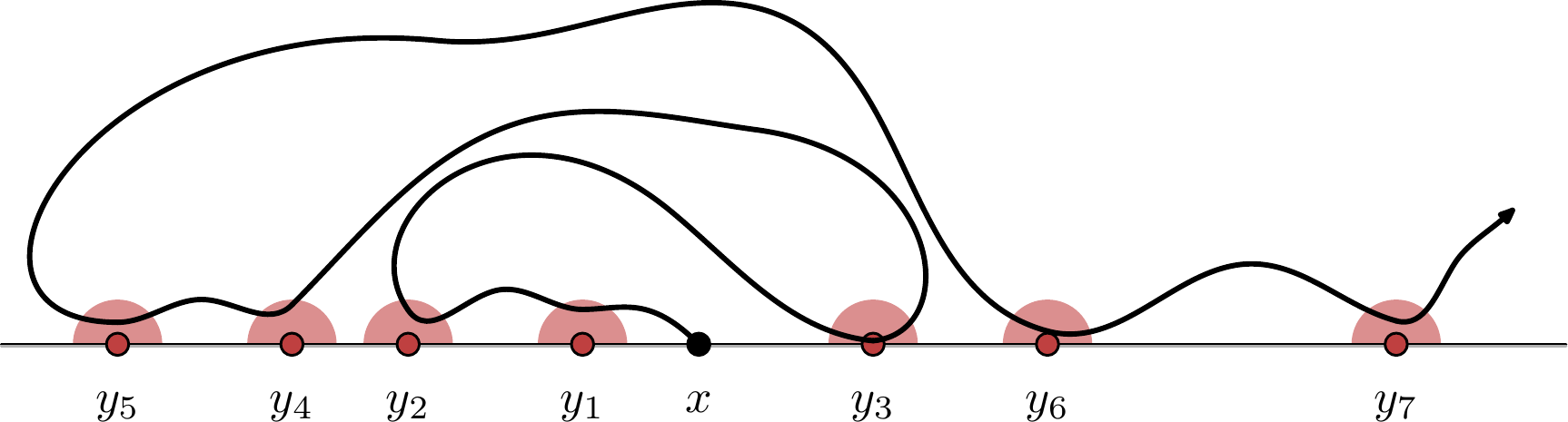}
\caption{\label{fig: boundary visits}
A schematic illustration of the chordal $\SLE$ curve in the upper half-plane,
visiting neighborhoods of given points on the boundary.
The curve starts at $x$, and
the visited points $y_1, \ldots,y_{N'}$ are numbered in the order of the
visits, as depicted in the figure. 
The boundary visit amplitude 
$\Ampl_\omega=\Ampl_\omega(x;y_1,y_2,\ldots,y_{N'})$
is a function of these $N' + 1$ variables.
It satisfies the covariance
\[ \Ampl_\omega(x; y_1 , \ldots , y_{N'})
= \lambda^{N \frac{8-\kappa}{\kappa}} \; \Ampl_\omega(\lambda x + \sigma; \lambda y_1 + \sigma,\ldots,\lambda y_{N'} + \sigma) \]
for all $\lambda>0$, $\sigma \in \bR$, and the
following partial differential equations:
\[ \left[\pdder{x}-\frac{4}{\kappa}\sL_{-2}^{(0)}\right]\Ampl_\omega(x;y_1,\ldots,y_{N'})=0\label{eq: bdry visit PDEs} \]
\[ \left[\pddder{y_j}-\frac{16}{\kappa}\sL_{-2}^{(j)}\pder{y_j}+\frac{8(8-\kappa)}{\kappa^2}\sL_{-3}^{(j)}\right]\Ampl_\omega(x;y_1,\ldots,y_{N'})=0 \] 
for $j=1,\ldots,N'$, 
where 
\[ \sL_{-2}^{(0)}=\sum_{k=1}^{N'}\left(\frac{-1}{y_k-x}\pder{y_k}+\frac{8/\kappa-1}{(y_k-x)^2}\right) \]
\[ \sL_{-2}^{(j)}=\frac{-1}{x-y_j}\pder{x}+\frac{3/\kappa-1/2}{(x-y_j)^2}+\sum_{k\neq j}\left(\frac{-1}{y_k-y_j}\pder{y_k}+\frac{8/\kappa-1}{(y_k-y_j)^2}\right) \]
\[ \sL_{-3}^{(j)}=\frac{-1}{(x-y_j)^2}\pder{x}+\frac{6/\kappa-1}{(x-y_j)^3}+\sum_{k\neq j}\left(\frac{-1}{(y_k-y_j)^2}\pder{y_k}+\frac{16/\kappa-2}{(y_k-y_j)^3}\right) . \]
}
\end{figure}

We seek vectors $\Bdryvec_\omega$ in the tensor product representation of $\Uqsltwo$,
\begin{align}\label{eq: bdry visit tensor product}
\Wd_{3}^{\tens R(\omega)}\tens\Wd_2\tens\Wd_{3}^{\tens L(\omega)},
\end{align}
where $\Wd_2$ and $\Wd_3$ are the two and three dimensional
irreducible representations defined in Section~\ref{subsec: quantum group}.

The requirements for $\Bdryvec_\omega$ are expressed in terms of
projections to subrepresentations.
Recall from Lemma~\ref{lem: tensor product representations of quantum sl2}
that $\Wd_3\tens\Wd_2 \isom \Wd_2 \oplus \Wd_4$ and
$\Wd_2\tens\Wd_3 \isom \Wd_2 \oplus \Wd_4$.
We consider the projections to the two dimensional subrepresentations.
To identify the images with $\Wd_2$, we use
$\Tbas_{0}^{(2;2,3)}$ and 
$\Tbas_{0}^{(2;3,2)}$~from~\eqref{eq: tensor product hwv}
as the highest weight vectors,
and thus define projections composed with this identification as
\begin{align*}
\hat{\pi}^{(2;2,3)} \colon \; & \Wd_3\tens\Wd_2 \rightarrow \Wd_2,
    \qquad \qquad \hat{\pi}^{(2;2,3)}(\Tbas_{l}^{(2;2,3)}) = \Wbas_l^{(2)} ,\\
\hat{\pi}^{(2;3,2)} \colon \; & \Wd_2\tens\Wd_3 \rightarrow \Wd_2,
    \qquad \qquad \hat{\pi}^{(2;3,2)}(\Tbas_{l}^{(2;3,2)}) = \Wbas_l^{(2)} \qquad \text{for } l=0,1 .
\end{align*}
We let these projections act at the natural positions in the tensor
product \eqref{eq: bdry visit tensor product}, and define
\begin{align*}
& \hat{\pi}^{(2)}_+\colon\Wd_{3}^{\tens R}\tens\Wd_2\tens\Wd_{3}^{\tens L}\rightarrow\Wd_{3}^{\tens R-1}\tens\Wd_2\tens\Wd_{3}^{\tens L} ,
    & & \hat{\pi}^{(2)}_+ = \id^{\tens R-1} \tens \hat{\pi}^{(2;2,3)} \tens \id^{\tens L}, \\
& \hat{\pi}^{(2)}_-\colon\Wd_{3}^{\tens R}\tens\Wd_2\tens\Wd_{3}^{\tens L}\rightarrow\Wd_{3}^{\tens R}\tens\Wd_2\tens\Wd_{3}^{\tens L-1} ,
    & & \hat{\pi}^{(2)}_- = \id^{\tens R} \tens \hat{\pi}^{(2;3,2)} \tens \id^{\tens L-1} .
\end{align*}
Also, by Lemma~\ref{lem: tensor product representations of quantum sl2}, we have
$\Wd_3\tens\Wd_3 \isom \Wd_1 \oplus \Wd_3 \oplus \Wd_5$, 
and we consider the projections
to one and three dimensional subrepresentations.
To identify the images with $\Wd_1 \isom \bC$ and $\Wd_3$, we use
$\Tbas_{0}^{(1;3,3)}$ and 
$\Tbas_{0}^{(3;3,3)}$~from~\eqref{eq: tensor product hwv}
as the 
highest weight vectors, and define
\begin{align*}
\hat{\pi}^{(1)} \colon \; & \Wd_3\tens\Wd_3 \rightarrow \bC,
    & & \hat{\pi}^{(1)}(\Tbas_{0}^{(1;3,3)}) = 1 , \\
\hat{\pi}^{(3)} \colon \; & \Wd_3\tens\Wd_3 \rightarrow \Wd_3 ,
    & & \hat{\pi}^{(3)}(\Tbas_{l}^{(3;3,3)}) = \Wbas_l^{(3)} \qquad \text{for } l=0,1,2 .
\end{align*}
We let these projections act at various positions in the tensor
product \eqref{eq: bdry visit tensor product}, and define
\begin{align*}
&\hat{\pi}^{(3)}_{+;m}\colon\Wd_{3}^{\tens R}\tens\Wd_2\tens\Wd_{3}^{\tens L}\rightarrow\Wd_{3}^{\tens R-1}\tens\Wd_2\tens\Wd_{3}^{\tens L} ,
    & & \hat{\pi}^{(3)}_{+;m} = \id^{\tens R-m-1} \tens \hat{\pi}^{(3)} \tens \id^{\tens m-1} \tens \id \tens \id^{\tens L} , \\
&\hat{\pi}^{(3)}_{-;m}\colon\Wd_{3}^{\tens R}\tens\Wd_2\tens\Wd_{3}^{\tens L}\rightarrow\Wd_{3}^{\tens R}\tens\Wd_2\tens\Wd_{3}^{\tens L-1} ,
    & & \hat{\pi}^{(3)}_{-;m} = \id^{\tens R} \tens \id \tens \id^{\tens m-1} \tens \hat{\pi}^{(3)} \tens \id^{\tens L-m-1} , \\
&\hat{\pi}^{(1)}_{+;m}\colon\Wd_{3}^{\tens R}\tens\Wd_2\tens\Wd_{3}^{\tens L}\rightarrow\Wd_{3}^{\tens R-2}\tens\Wd_2\tens\Wd_{3}^{\tens L} ,
    & & \hat{\pi}^{(1)}_{+;m} = \id^{\tens R-m-1} \tens \hat{\pi}^{(3)} \tens \id^{\tens m-1} \tens \id \tens \id^{\tens L} , \\
&\hat{\pi}^{(1)}_{-;m}\colon\Wd_{3}^{\tens R}\tens\Wd_2\tens\Wd_{3}^{\tens L}\rightarrow\Wd_{3}^{\tens R}\tens\Wd_2\tens\Wd_{3}^{\tens L-2} ,
    & & \hat{\pi}^{(1)}_{-;m} = \id^{\tens R} \tens \id \tens \id^{\tens m-1} \tens \hat{\pi}^{(1)} \tens \id^{\tens L-m-1} .
\end{align*}

For any visiting order $\omega \in \Orders_{N'}$ with given $L$ and $R$, the vector
$\Bdryvec_\omega$ is required to be a highest weight vector of
a two dimensional subrepresentation of the tensor product \eqref{eq: bdry visit tensor product},
i.e., to lie in the subspace
\begin{align*}
\HWsp_2\left(\Wd_3^{\tens R} \tens \Wd_2 \tens \Wd_3^{\tens L}\right)
\; = \; \set{v\in\Wd_3^{\tens R} \tens \Wd_2 \tens \Wd_3^{\tens L} \; \Big| \;
    E.v=0 , \; K.v = q\,v } .
\end{align*}
The other conditions depend on the order $\omega$.
We say that the $m$:th and $m+1$:st points on the right 
are successively visited if the $m$:th and $m+1$:st $+\,$-symbols 
in the sequence $\omega = (\omega_1 , \ldots, \omega_{N'})$
are not separated by any $-\,$-symbols. 
More formally, this means that there exists an index $j$ such that 
$\omega_j = \omega_{j+1} = +$ 
and $\#\set{i \in \set{1,\ldots,j} \, \Big| \, \omega_i = +} = m$. 
The visiting order obtained from $\omega$ by collapsing these successive
visits is denoted below by $\hat{\omega} = (\omega_1 , \ldots, \omega_{j-1},\omega_{j+1}, \ldots, \omega_{N'})$.
We define 
successive visits on the left similarly.
The requirements for $\Bdryvec_\omega$ are the following:
\begin{itemize}
\item The vector $\Bdryvec_\omega$ is a highest weight vector of a doublet subrepresentation, 
\begin{align}
&K.\Bdryvec_\omega = q \,\Bdryvec_\omega\label{eq: bdry visit cartan eigenvalue} , \qquad
E.\Bdryvec_\omega = 0 . 
\end{align}
\item
Depending on whether the $m$:th and $m+1$:st points on the right 
are successively visited or not,
we have, for $\epsilon = +$,
\begin{align}
&\hat{\pi}^{(1)}_{\epsilon;m}(\Bdryvec_\omega) = \; 0\label{eq: bdry visit singlet projection conditions}\\
&\hat{\pi}^{(3)}_{\epsilon;m}(\Bdryvec_\omega) = \; \label{eq: bdry visit triplet projection conditions}
\begin{cases} 0 & \mbox{in the case of non-successive visits}\\
    C_3 \times \Bdryvec_{\hat{\omega}} & \mbox{in the case of successive visits,} \end{cases}
\end{align}
where $\hat{\omega}$ 
is the order obtained from $\omega$ by collapsing these successive visits,
and $C_3=\frac{\qnum{2}^2}{q^2+q^{-2}}$ is a non-zero constant. 
For the $m$:th and $m+1$:st points on the left, we 
require~\eqref{eq: bdry visit singlet projection conditions}~--~\eqref{eq: bdry visit triplet projection conditions} 
for $\epsilon = -$. See also Figure~\ref{fig: two consecutive points}.
\item Let $\omega_1 = \pm$ denote the side of the first visit, and
$\mp = -\omega_1$ the opposite side. Then we have
\begin{align}
&\hat{\pi}^{(2)}_{\pm}(\Bdryvec_\omega) = C_2 \times \Bdryvec_{\hat{\omega}}\label{eq: bdry visit doublet projection conditions}\\
&\hat{\pi}^{(2)}_{\mp}(\Bdryvec_\omega) = 0,\nonumber
\end{align}
where $\hat{\omega}=(\omega_2,\omega_3,\ldots,\omega_{N'})$ is the order obtained from $\omega$ by collapsing the first visit,
and $C_2 = \frac{\qnum{2}^2}{\qnum{3}}$ is a non-zero constant. See also Figure~\ref{fig: first visited point}.
\end{itemize}
\begin{figure}
\includegraphics[scale=.7]{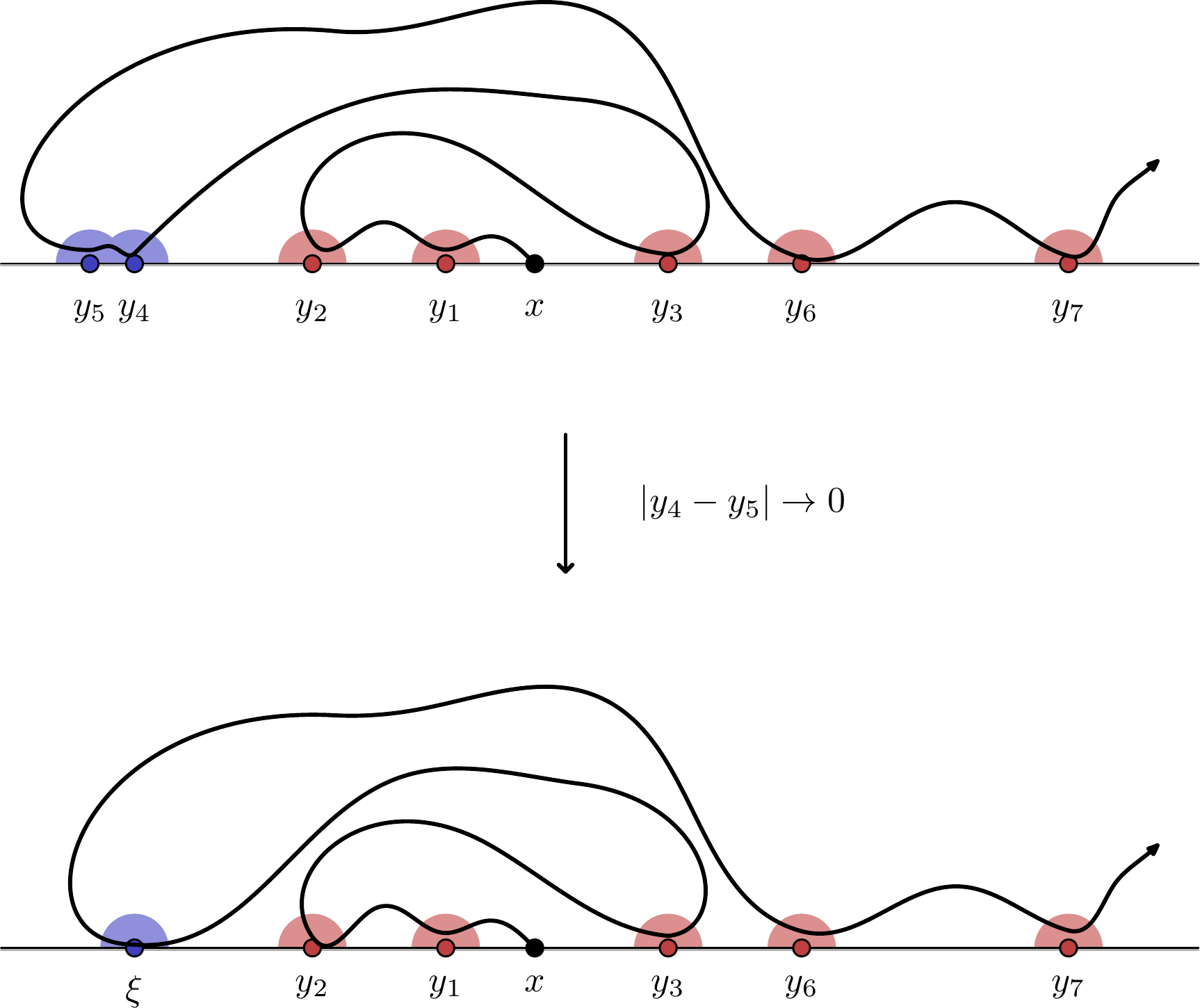}
\caption{\label{fig: two consecutive points}
Collapsing one of two close-by successively visited points
$y_j,y_{j+1}$ on the same side.
This figure illustrates the 
conditions~\eqref{eq: bdry visit singlet projection conditions}~--~\eqref{eq: bdry visit triplet projection conditions} 
for the vector $\Bdryvec_\omega$
in the case of successive visits,
which guarantee the following asymptotics of the 
boundary visit amplitude $\Ampl_\omega$:
$$\hspace{-2cm} \lim_{y_j,y_{j+1}\to\xi}|y_{j+1}-y_j|^{\frac{8-\kappa}{\kappa}}\Ampl_\omega(x;y_1,\ldots,y_{N'})
=C_3'\times\Ampl_{\hat{\omega}}(x,y_1,\ldots,y_{j-1},\xi,y_{j+2},\ldots,y_{N'})$$
for a non-zero constant $C_3'$.
On the other hand, for non-successively visited consecutive points on the same side, 
the corresponding limit is zero,
as follows from~\eqref{eq: bdry visit singlet projection conditions}~--~\eqref{eq: bdry visit triplet projection conditions}.
In~the example depicted in this figure, 
the collapsed visits $y_4,y_5$ are 
the third and fourth visits on the left, i.e., in the notation of 
\eqref{eq: bdry visit singlet projection conditions}~--~\eqref{eq: bdry visit triplet projection conditions}
we have $j=4$ and $m=3$, $\epsilon=-1$.
}
\end{figure}

\begin{figure}
\includegraphics[scale=.7]{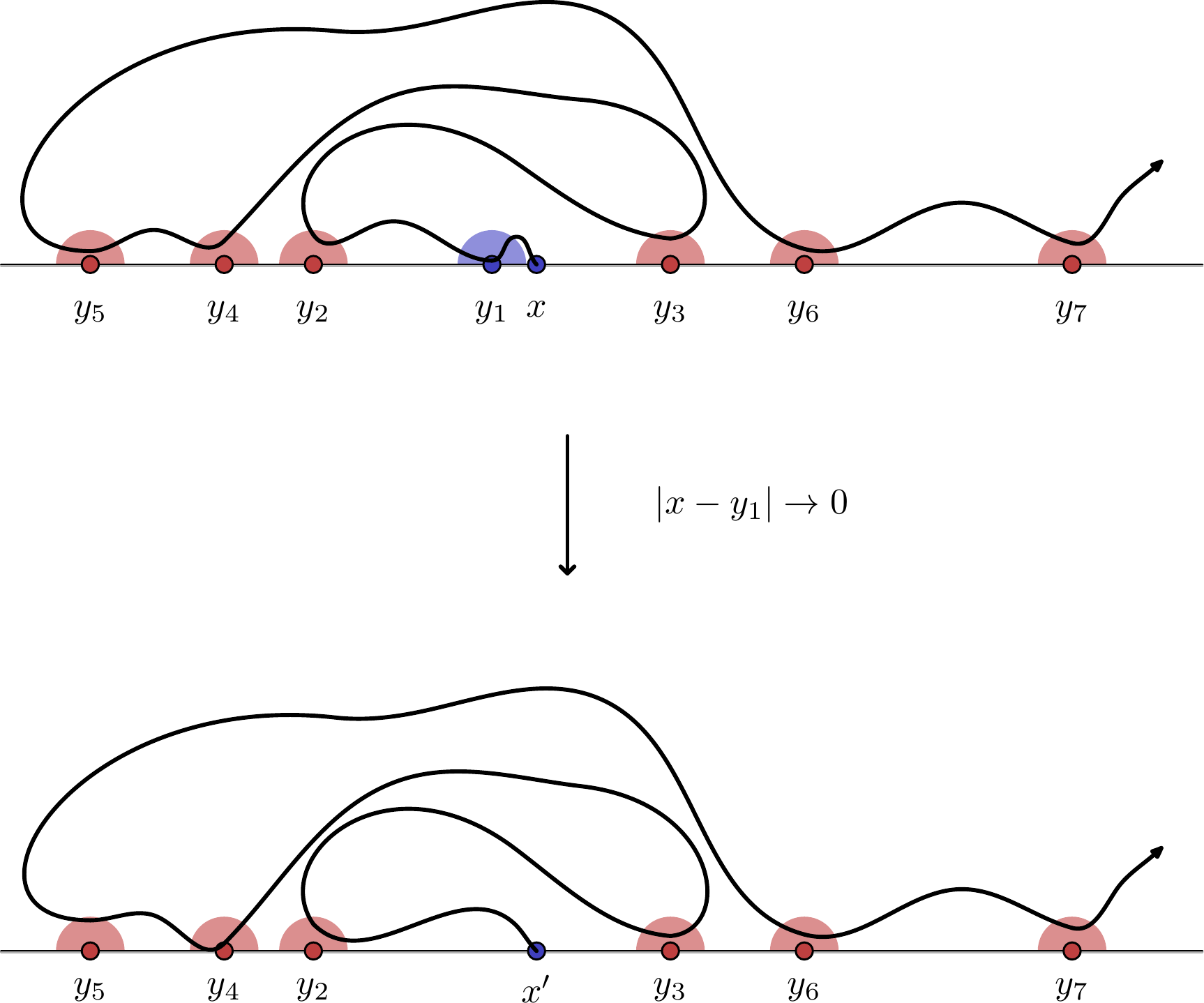}
\caption{\label{fig: first visited point}
Collapsing the first visited point $y_1$.
This figure illustrates the 
conditions~\eqref{eq: bdry visit doublet projection conditions}
for the vector $\Bdryvec_\omega$,
which guarantee the following asymptotics of the 
boundary visit amplitude $\Ampl_\omega$:
for the first visited point $y_1$, we have
$$
\lim_{y_1,x\to x'}|y_1-x|^{\frac{8-\kappa}{\kappa}}
        \Ampl_\omega(x;y_1,\ldots,y_{N'})
    = C_2' \times \Ampl_{\hat{\omega}}(x',y_2,\ldots,y_{N'})
    $$
for a non-zero constant $C_2'$, and
for the first point $y_k$ on the opposite side,
$$\hspace{-2cm}
\lim_{y_k,x\to x'}|y_k-x|^{\frac{8-\kappa}{\kappa}}\Ampl_\omega(x;y_1,\ldots,y_{N'})=0.$$
In the example depicted in this figure, the first visit takes place on the left side, i.e., in the notation
of~\eqref{eq: bdry visit doublet projection conditions}, we have $\omega_1=-$.
Also, the first point on the right is the third one visited, i.e.,
in the limit equation above one should take $k=3$.}
\end{figure}

Our main result in the quantum group setup 
of the boundary visit problem
is the existence of solutions to the 
system~\eqref{eq: bdry visit cartan eigenvalue}~--~\eqref{eq: bdry visit doublet projection conditions} 
and their uniqueness up to normalization.
The proof is based on a number of observations made in Section~\ref{subsub: bdry visit construction},
which are combined in Section~\ref{sub: bdry visit existence}.
\begin{thm}\label{thm: bdry visit vectors}
There exists a unique collection 
$\left(\Bdryvec_{\omega}\right)_{\omega\in\Orders}$ of vectors 
$\Bdryvec_{\omega}\in\Wd_{3}^{\tens R(\omega)}\tens\Wd_2\tens\Wd_{3}^{\tens L(\omega)}$ 
such that the system of 
equations~\eqref{eq: bdry visit cartan eigenvalue}~--~\eqref{eq: bdry visit doublet projection conditions} 
holds for all $\omega\in\Orders$, with the normalization 
$\Bdryvec_{|N'=0}=\Wbas_0\in\Wd_2$.
\end{thm}

As a corollary, we obtain the existence of solutions to $\SLE$ boundary
visit amplitudes for any number of visited points.
For this one applies a slightly different ``spin chain~--~Coulomb gas correspondence''
$\sF$ from $H_2(\Wd_3^{\tens R} \tens \Wd_2 \tens \Wd_3^{\tens L})$
to functions of $L+R+1$ variables 
--- see \cite{KP-conformally_covariant_boundary_correlation_functions_with_a_quantum_group,
JJK-SLE_boundary_visits} for details.
\begin{thm}\label{thm: bdry visit amplitudes}
The collection $\left(\Ampl_{\omega}\right)_{\omega\in\Orders}$ of functions 
$\Ampl_{\omega}=\sF[\Bdryvec_{\omega}]$ satisfies 
the system of 
partial differential equations, covariance, 
and boundary conditions required in \cite{JJK-SLE_boundary_visits},
that is, the equations given in Figures~\ref{fig: boundary visits},
\ref{fig: two consecutive points}, and \ref{fig: first visited point}. 
\end{thm}
\begin{proof}
Since the vectors $\Bdryvec_{\omega}$ satisfy
\eqref{eq: bdry visit cartan eigenvalue},
the functions
$\Ampl_{\omega}=\sF[\Bdryvec_{\omega}]$ satisfy the PDEs and covariance
given in Figure~\ref{fig: boundary visits},
by the (PDE) and (COV) parts of 
\cite[Theorem~4.17]{KP-conformally_covariant_boundary_correlation_functions_with_a_quantum_group}.
The asymptotic conditions of Figures~\ref{fig: two consecutive points} and \ref{fig: first visited point}
follow from the projection conditions~\eqref{eq: bdry visit singlet projection conditions}~--~\eqref{eq: bdry visit doublet projection conditions}
for the vectors $\Bdryvec_{\omega}$, by the (ASY) part of
\cite[Theorem~4.17]{KP-conformally_covariant_boundary_correlation_functions_with_a_quantum_group}.
\end{proof}


\subsection{\label{subsub: Associating visiting orders to link patterns}Link patterns associated to visiting orders}

In Section~\ref{subsub: bdry visit construction}, we construct solutions
$\Bdryvec_{\omega}$ to the
system~\eqref{eq: bdry visit cartan eigenvalue}~--~\eqref{eq: bdry visit doublet projection conditions}.
For a given visiting order $\omega$, the vector $\Bdryvec_{\omega}$ 
will be built from
a vector $v_{\alpha}$ of Theorem~\ref{thm: existence of multiple SLE vectors}, 
with an appropriately chosen link pattern $\alpha = \alpha(\omega)$.
The mapping 
\begin{align}\label{eq: visiting order to link pattern}
\qquad \omega\;\mapsto\;\alpha(\omega),
\qquad \Orders_{N'}\;\rightarrow\;\LP_N,
\end{align}
where $N=N'+1=R(\omega)+L(\omega)+1$,
associates to each visiting order 
$\omega=(\omega_1,\ldots,\omega_{N'})\in\Orders_{N'}$ 
a link pattern $\alpha(\omega)\in\LP_N$ as illustrated and explained
in Figure~\ref{fig: boundary visits to link patterns},
and defined in detail below.

Let $\omega \in \Orders_{N'}$
and let $N = N'+1$ and $L=L(\omega)$. 
Then $\alpha = \alpha(\omega)$ contains the $N$ links
$\link{a_1}{b_1}, \ldots, \link{a_N}{b_N}$, whose indices are defined recursively as follows.
The index $a_1 = 2 L + 1$ corresponds to the starting point $x$, 
and $b_1 = a_1 + \omega_1$ corresponds to entering the first visited point $y_1$.
For $j=2,\ldots,N'+1$, the index $a_j = b_{j-1} + \omega_{j-1}$
corresponds to exiting the point $y_{j-1}$. 
Also, for $j=2,\ldots,N'$, the index $b_j$ corresponds to entering
the point $y_j$: 
if $\omega_j = +$ then $b_j = \max \set{a_1, a_2, \ldots, a_{j-1}} + 1$,
and if $\omega_j = -$ then $b_j = \min \set{a_1, a_2, \ldots, a_{j-1}} - 1$.
Finally, we also set $b_{N}=2N$, which corresponds to entering 
an auxiliary point $y_\infty$ --- see 
Figure~\ref{fig: boundary visits to link patterns}.
It is straightforward to check that this defines a link pattern $\alpha=\alpha(\omega)$.

\begin{rem}\label{rem: collapsing visits corresponds to removing links}
\emph{Recall that the projection 
conditions~\eqref{eq: bdry visit singlet projection conditions}~--~\eqref{eq: bdry visit doublet projection conditions}
are written in terms of a visiting 
order $\hat{\omega}$ obtained from $\omega$ by
collapsing two successive visits into one, or collapsing the first visit with the starting point.
From the definition of the map $\alpha \mapsto \alpha(\omega)$ of \eqref{eq: visiting order to link pattern},
it is easy to see that $\hat{\alpha} = \alpha(\hat{\omega})$ is obtained
from $\alpha=\alpha(\omega)$ by removing one link. More precisely,
in the notation used in the above definition, the two cases are the following.
For the case \eqref{eq: bdry visit doublet projection conditions} of
collapsing the first visit, we have
\begin{align*}
\hat{\omega} & \; = (\omega_2,\omega_3,\ldots,\omega_{N'}) \qquad \qquad \qquad \text{ and} & 
\alpha(\hat{\omega}) & \; = \alpha(\omega) \removeLink \link{a_1}{b_1} ,
\end{align*}
see also Figure~\ref{fig: removing first visited point}.
For the case \eqref{eq: bdry visit singlet projection conditions} of
collapsing the $m$:th and $m+1$:st visit on the right, 
we have
\begin{align*}
\hat{\omega} & \; = (\omega_1,\ldots,\omega_{j-1},\omega_{j+1},\ldots,\omega_{N'}) \qquad \qquad \text{ and} & 
\alpha(\hat{\omega}) & \; = \alpha(\omega) \removeLink \link{a_{j+1}}{b_{j+1}} ,
\end{align*}
where the index $j$ is such that 
$\omega_j = \omega_{j+1} = +$ 
and $\#\set{i \in \set{1,\ldots,j} \, \Big| \, \omega_i = +} = m$.
For the case of $m$:th and $m+1$:st visit on the left instead, the choice of $j$
is modified accordingly.
See also Figure~\ref{fig: removing one of consecutively visited points}.}
\end{rem}

\begin{figure}
\includegraphics[width=.8\textwidth]{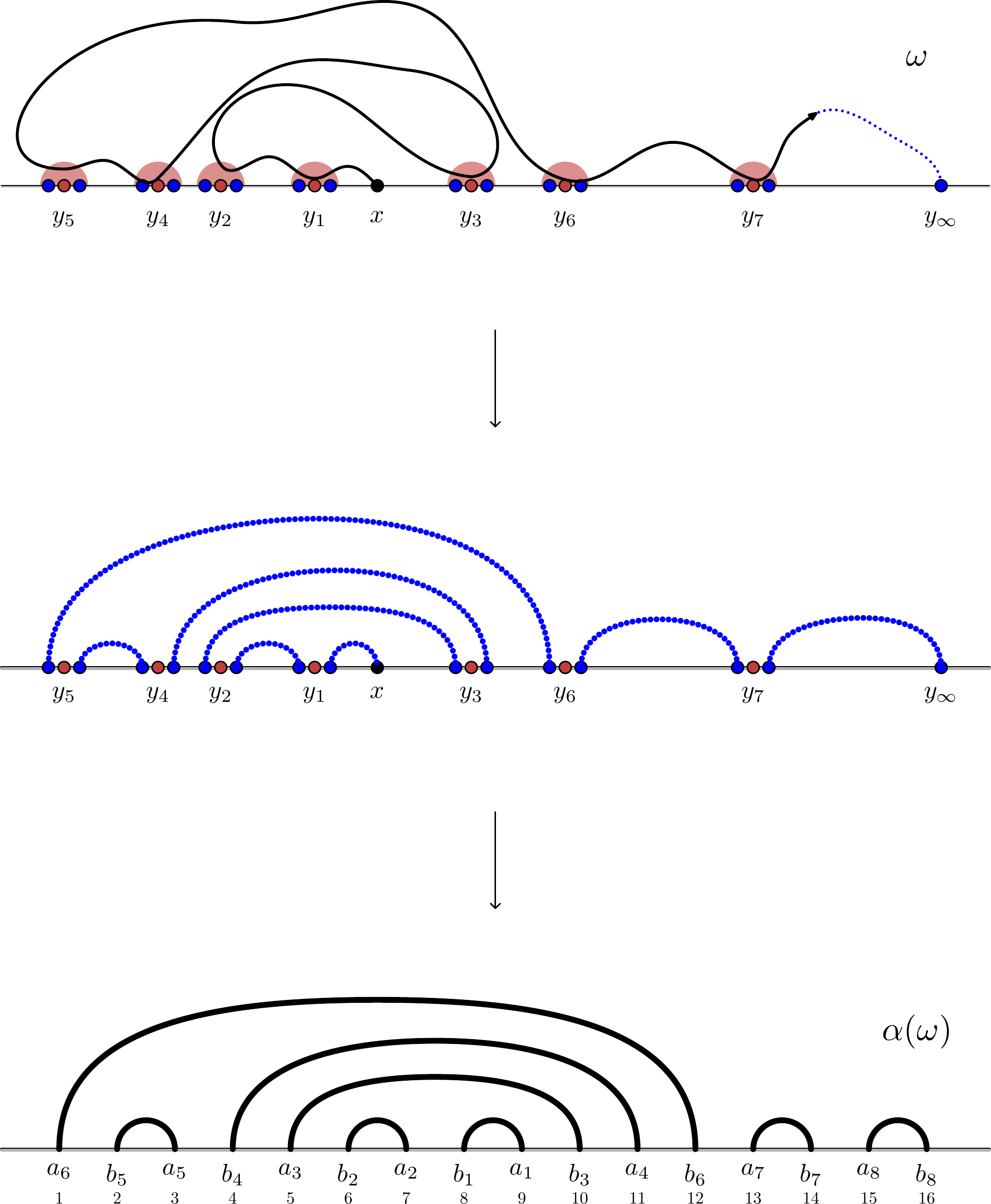}
\caption{\label{fig: boundary visits to link patterns}
The mapping $\omega \mapsto \alpha(\omega)$
of \eqref{eq: visiting order to link pattern} can be described as follows.
Think of the visiting orders as planar connectivities of
$N'+1$ points, where small neighborhoods of the visited points $y_j$ have two lines 
attached to them, the starting point $x$ has one line attached, 
and every line is connected to another line so that
no lines of the same point are connected (no loops).
%
We open up the connectivity corresponding to the
visiting order $\omega$: replace the points $y_j$
having two lines attached by two points, each having just 
one line attached. After this procedure, there will be one leftover line,
attached to the last visited point (i.e. one of the two points by 
which we replaced the last visited point). 
We add a point $y_\infty$ to the right side of all other points, and
connect the leftover line to this point.
Finally, we label the endpoints of the links from left to right
appropriately, to correspond with the labels of the
endpoints of a link pattern $\alpha(\omega)$ in $\LP_N$. 
}
\end{figure}

\begin{figure}
\includegraphics[width=.8\textwidth]{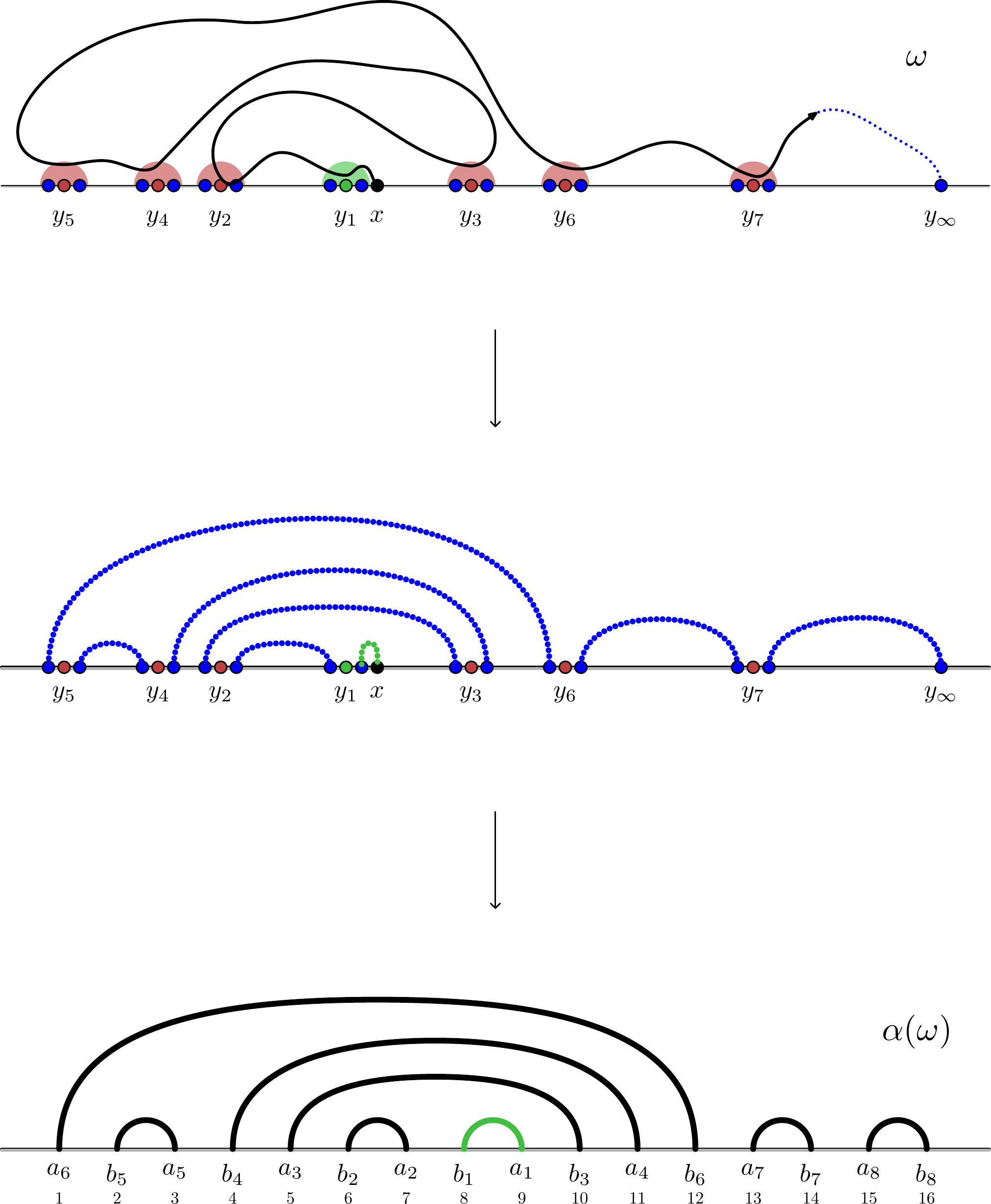}
\caption{\label{fig: removing first visited point} 
Collapsing the first visited point:
the effect on the link pattern $\alpha(\omega)$. 
In the example depicted in this figure, we have 
$\omega_1=-$, and the green link 
$\link{b_1}{a_1}=\link{8}{9}$ is to be removed.
}
\end{figure}

\begin{figure}
\includegraphics[width=.8\textwidth]{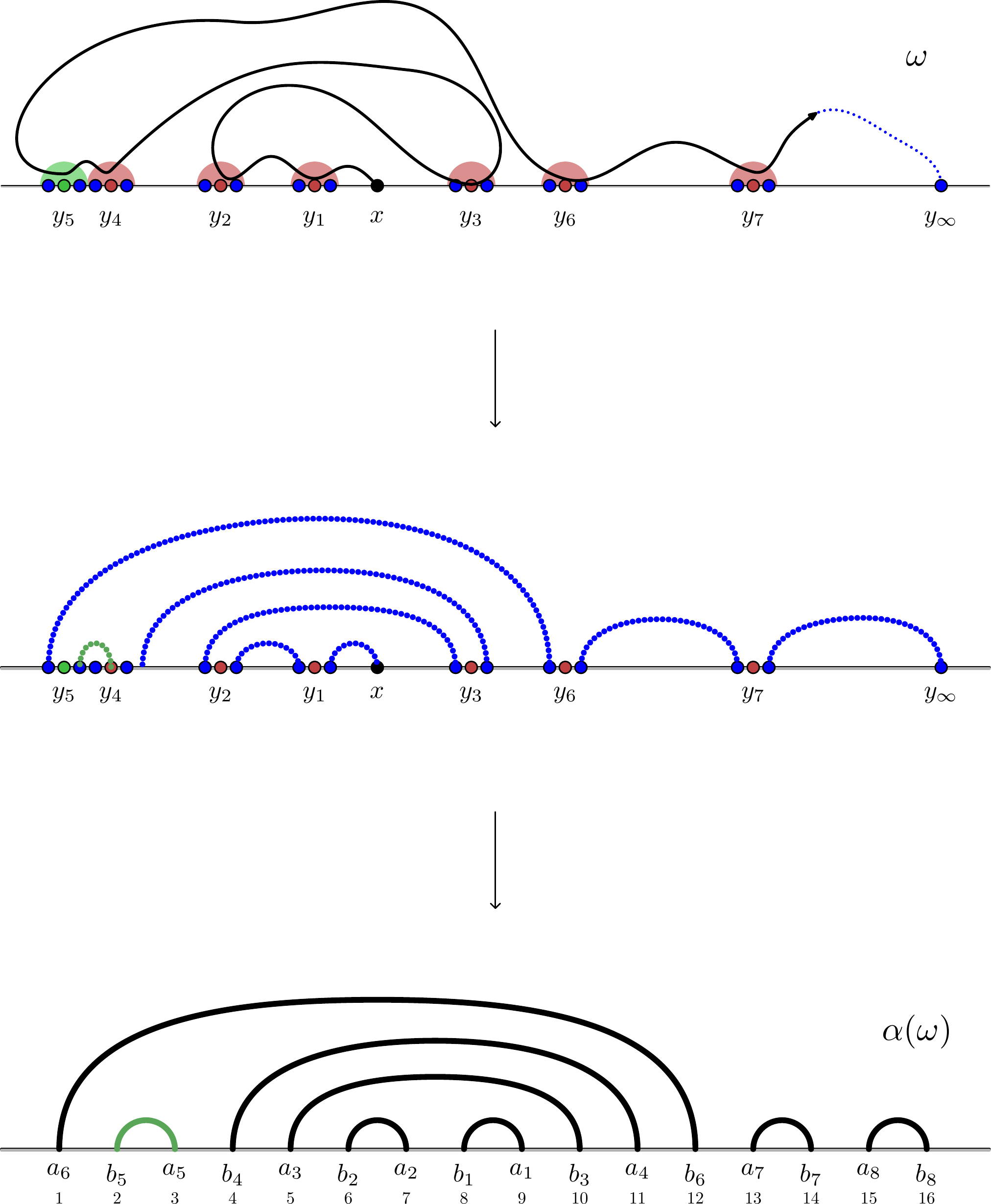}
\caption{\label{fig: removing one of consecutively visited points}
Collapsing two successive visits on the same side:
the effect on the link pattern $\alpha(\omega)$. 
In the example depicted in this figure, 
the collapsed visits $y_4,y_5$ are 
the third and fourth visits on the left,
and the green link 
$\link{b_5}{a_5}=\link{2}{3}$ is to be removed.
}
\end{figure}

\subsection{\label{subsub: bdry visit construction}Construction of solutions}


Recall from Lemma~\ref{lem: tensor product representations of quantum sl2}
that $\Wd_2\tens\Wd_2 \isom \Wd_1 \oplus \Wd_3$. In this section, we use
projections to these two irreducible subrepresentations. We identify the subrepresentations
with $\Wd_3$ and $\Wd_1 \isom \bC$ by using the highest weight vectors
$\Tbas_{0}^{(3;2,2)}$ and 
$\Tbas_{0}^{(1;2,2)}$~from~\eqref{eq: tensor product hwv}.
We thus define projections composed with the identifications as
\begin{align*}
\hat{\pi}^{(3)} \colon \; & \Wd_2\tens\Wd_2 \rightarrow \Wd_3,
    & \hat{\pi}^{(3)}(\Tbas_{l}^{(3;2,2)}) & \; = \Wbas_l^{(3)} \qquad \text{for } l=0,1,2 , 
\qquad \\
\hat{\pi}^{(1)} \colon \; & \Wd_2\tens\Wd_2 \rightarrow \bC,
    & \hat{\pi}^{(1)}(\Tbas_{0}^{(1;2,2)}) & \; = 1 .
\end{align*}
We trust that no confusion arises, although the notation $\hat{\pi}^{(3)}$
coincides with that introduced in Section~\ref{subsec: Quantum group solution},
since the two projections are defined on different spaces. We also denote by
\begin{align*}
\iota^{(3)}\colon\Wd_3\hookrightarrow\Wd_2\tens\Wd_2
\end{align*}
the embedding of the three dimensional subrepresentation into the tensor
product $\Wd_2\tens\Wd_2$ such that 
$\iota^{(3)}(\Wbas_l^{(3)})=\Tbas_{l}^{(3;2,2)}$, and thus,
$\hat{\pi}^{(3)}\circ\iota^{(3)}=\id$.

Given a visiting order $\omega \in \Orders_{N'}$ with $L(\omega)=L$ and $R(\omega)=R$,
the vector $\Bdryvec_{\omega}$
satisfying~\eqref{eq: bdry visit cartan eigenvalue}~--~\eqref{eq: bdry visit doublet projection conditions}
will be constructed from the vector 
$v_{\alpha(\omega)}\in\HWsp_1(\Wd_2^{\tens 2N})$. 
The construction, $\Bdryvec_{\omega}=R_+ \big( \Projection(v_{\alpha(\omega)}) \big)$,
is summarized in the following diagram:
\begin{align*}
\xymatrix{
  \HWsp_1\left(\Wd_2^{\tens 2N}\right) \; \; \ar@<1ex>[r]^{\hspace{-1.2cm}\Projection}
  \ar@{<-^{)}}@<-2ex>[r]_{\hspace{-1.2cm}\Embedding} & \; \;  \HWsp_1\left(\Wd_2 \tens \Wd_3^{\tens R} \tens \Wd_2 \tens \Wd_3^{\tens L}\right) \ar[r]^{\hspace{.5cm}R_+} & \; \;  \HWsp_2\left(\Wd_3^{\tens R} \tens \Wd_2 \tens \Wd_3^{\tens L}\right) \\
  v_{\alpha(\omega)}  \; \; \; \; \; \; 
  \ar@{|->}@<1ex>[r]^{\hspace{-1.2cm}\Projection} \ar@{<-^{)}}@<-2ex>[r]_{\hspace{-1.2cm}\Embedding} & \; \; \; \; \; \; \Bdryvec_{\omega}^{\infty} \; \; \; \; \; \;\ar@{|->}[r]^{\hspace{.5cm}R_+} & \; \;  \; \; \; \;\Bdryvec_{\omega},
}
\end{align*}
where the notations 
are defined below.
The projection 
$\Projection$ and embedding $\Embedding$ are 
defined by
\begin{align*}
\Projection = & \; \id\tens(\hat{\pi}^{(3)})^{\tens R}\tens\id\tens(\hat{\pi}^{(3)})^{\tens L} , &
\Embedding = & \;\id\tens(\iota^{(3)})^{\tens R}\tens\id\tens(\iota^{(3)})^{\tens L} , \\
\Projection \colon & \; \Wd_2^{\tens 2N} \longrightarrow \Wd_2 \tens \Wd_3^{\tens R} \tens \Wd_2 \tens \Wd_3^{\tens L} , &
\Embedding \colon & \; \Wd_2 \tens \Wd_3^{\tens R} \tens \Wd_2 \tens \Wd_3^{\tens L} \longrightarrow \Wd_2^{\tens 2N} ,
\end{align*}
and we denote $\Projection(v_{\alpha(\omega)})=\Bdryvec_{\omega}^{\infty}$.
The notation
\begin{align*}
\HWsp_1\left(\Wd_2 \tens \Wd_3^{\tens R} \tens \Wd_2 \tens \Wd_3^{\tens L}\right) \; = \;
\set{ v \in \Wd_2 \tens \Wd_3^{\tens R} \tens \Wd_2 \tens \Wd_3^{\tens L} \; \big| \; E.v = 0, \; K.v = v}
\end{align*}
is used for the trivial subrepresentation.
Finally, it can be shown, see \cite[Lemma~5.3]{KP-conformally_covariant_boundary_correlation_functions_with_a_quantum_group},
that for any vector $v \in \HWsp_1\left(\Wd_2 \tens \Wd_3^{\tens R} \tens \Wd_2 \tens \Wd_3^{\tens L}\right)$
there exists a unique vector $\tau_0^+ \in \HWsp_2(\Wd_3^{\tens R} \tens \Wd_2 \tens \Wd_3^{\tens L})$ such that
\begin{align*}
v= -q\;\Wbas_{0}^{(2)}\tens F.\tau_0^+ +\Wbas_{1}^{(2)}\tens\tau_0^+ .
\end{align*}
This defines a linear isomorphism $R_+ \colon v \mapsto \tau_0^+$
by \cite[Lemma~5.3]{KP-conformally_covariant_boundary_correlation_functions_with_a_quantum_group}.

\begin{rem}\label{rem: doublet conditions}
\emph{By construction, the vector $\Bdryvec_{\omega} = R_+ \big( \Projection(v_{\alpha(\omega)}) \big)$ 
satisfies the  conditions~\eqref{eq: bdry visit cartan eigenvalue}.
It remains to check that  also the
conditions~\eqref{eq: bdry visit singlet projection conditions}~--~\eqref{eq: bdry visit doublet projection conditions} are 
satisfied.}
\end{rem}

\begin{lem}\label{lem: pure geometries lie in projected space}
The image of $\HWsp_1\left(\Wd_2 \tens \Wd_3^{\tens R} \tens \Wd_2 \tens \Wd_3^{\tens L}\right)$
under the embedding $\Embedding$ is the space
\begin{align}\label{eq: reasonable subspace}
\set{ v \in \HWsp_1(\Wd_2^{\tens 2N}) \; \Big| \; 
\hat{\pi}^{(1)}_{2L+1-2m}(v) = 0 \text{ for } m = 1,\ldots,L
\; \text{ and } \; \hat{\pi}^{(1)}_{2L+2m'}(v) = 0 \text{ for } m' = 1,\ldots,R } .
\end{align}
The projection $\Projection$ restricted to this space is a bijection onto
$\HWsp_1\left(\Wd_2 \tens \Wd_3^{\tens R} \tens \Wd_2 \tens \Wd_3^{\tens L}\right)$,
and its inverse is $\Embedding$.
The vector $v_{\alpha(\omega)}$ lies in~
\eqref{eq: reasonable subspace}, and
in particular, 
$\Embedding\big(\Projection(v_{\alpha(\omega)})\big)=
\Embedding(\Bdryvec_{\omega}^{\infty})=v_{\alpha(\omega)}$.
\end{lem}
\begin{proof}
\label{rem: pure geometries lie in projected space}
The 
statements about the image of $\Embedding$ and restriction of $\Projection$
are clear from the decomposition $\Wd_2 \tens \Wd_2 \isom \Wd_1 \oplus \Wd_3$.
By definition of the mapping $\omega\mapsto\alpha(\omega)$, 
the link pattern $\alpha(\omega)$ does not contain links 
of type $\link{j}{j+1}$ where $j$ and $j+1$ correspond to the 
same visited point, i.e., $j=2L+1-2m$ or $j=2L+2m'$.
By the projection 
conditions~\eqref{eq: multiple sle projection conditions} for 
$v_{\alpha(\omega)}$, we have $\hat{\pi}_j(v_{\alpha(\omega)})=0$
for all such $j$.
\end{proof}

We next show how the projections appearing in 
the conditions~\eqref{eq: bdry visit singlet projection conditions}~--~\eqref{eq: bdry visit doublet projection conditions}, 
acting on the vector $\Bdryvec_\omega$, can be expressed in terms of 
singlet projections $\hat{\pi}^{(1)}_j$ acting on the vector 
$v_{\alpha(\omega)}$. This will be done in three separate cases, 
corresponding to the three projection conditions 
\eqref{eq: bdry visit singlet projection conditions}~--~\eqref{eq: bdry visit doublet projection conditions}, 
in the form of commutative diagrams in 
Lemmas~\ref{lem: doublet diagram}, \ref{lem: triplet diagram} and 
\ref{lem: singlet diagram}.
Using these commutative diagrams, we then deduce the desired properties 
of the vector $\Bdryvec_\omega$ from the projection properties
\eqref{eq: multiple sle projection conditions} of $v_{\alpha(\omega)}$,
in Corollaries~\ref{cor: doublet diagram}, \ref{cor: triplet diagram}
and \ref{cor: singlet diagram}.

Let $C_2 = \frac{\qnum{2}^2}{\qnum{3}}$, 
and note that $C_2 \neq 0$, since $q$ is not a root of unity.
\begin{lem}\label{lem: doublet diagram}
For any $v \in \Wd_3 \tens \Wd_2$, we have $\hat{\pi}^{(2)}(v) = C_2 \times \big( (\id \tens \hat{\pi}^{(1)}) \circ (\iota^{(3)}\tens\id) \big) (v)$,
and for any $v \in \Wd_2 \tens \Wd_3$, we have $\hat{\pi}^{(2)}(v) = C_2 \times \big( (\hat{\pi}^{(1)} \tens \id) \circ (\id\tens\iota^{(3)} \big) (v)$.
In other words, the following diagrams commute, up to 
the non-zero multiplicative constant $C_2$:
\begin{align*}\xymatrix{
    \Wd_{2} \tens \Wd_{2} \tens \Wd_{2} \; \; \ar@{<-^{)}}[rr]^{\hspace{.9cm}\iota^{(3)} \, \tens\, \id}\ar[d]_{\id\,\tens \,\hat{\pi}^{(1)} }
  & & \; \; \Wd_{3} \tens \Wd_{2} \ar[d]^{\hat{\pi}^{(2)}}
  & & \; \; \\
    \Wd_{2} \; \; \ar@{<->}[rr]_{\isom}
  & & \; \;  \Wd_{2}
}
\xymatrix{
    \Wd_{2} \tens \Wd_{2} \tens \Wd_{2} \; \; \ar@{<-^{)}}[rr]^{\hspace{.9cm}  \id\, \tens\,\iota^{(3)}}\ar[d]_{\hat{\pi}^{(1)}\,\tens \, \id}
  & & \; \; \Wd_{2} \tens \Wd_{3} \ar[d]^{\hat{\pi}^{(2)}}
  & & \; \; \\
    \Wd_{2} \; \; \ar@{<->}[rr]_{\isom}
  & & \; \;  \Wd_{2}  \; \; .
} \end{align*}
\end{lem}
\begin{proof}
We show the commutativity of the left diagram --- the right is similar.
Since the multiplicity of the subrepresentation $\Wd_2$ in $\Wd_3\tens\Wd_2$ is one, 
by Schur's lemma it suffices to show that the map
$(\id\,\tens \,\hat{\pi}^{(1)})\circ(\iota^{(3)}\,\tens\,\id)$
is non-zero. The vector
\begin{align*}
\Tbas_{0}^{(2;2,3)}=\frac{q^2}{1-q^2}\,\Wbas_0^{(3)}\tens\Wbas_1^{(2)}\,+\,\frac{q^2}{q^4-1}\,\Wbas_1^{(3)}\tens\Wbas_0^{(2)}
\end{align*}
satisfies $\hat{\pi}^{(2)}(\Tbas_{0}^{(2;2,3)})=\Wbas_0^{(2)}\in\Wd_2$,
by definition.
Using Lemma~\ref{lem: projection formulas}(a) 
and the basis \eqref{eq: triplet basis vectors}, we calculate
\begin{align*}
(\id\,\tens \,\hat{\pi}^{(1)})\circ(\iota^{(3)}\,\tens\,\id)
(\Tbas_{0}^{(2;2,3)})=\,&(\id\,\tens \,\hat{\pi}^{(1)})\left(\frac{q^2}{1-q^2}\,
\Tbas^{(3;2,2)}_0\tens\Wbas_1^{(2)}\,+\,\frac{q^2}{q^4-1}\,
\Tbas^{(3;2,2)}_1\tens\Wbas_0^{(2)}\right)\\
=\,&\frac{\qnum{3}}{\qnum{2}^2}\times\Wbas_0^{(2)}
= \frac{1}{C_2}\times\Wbas_0^{(2)} \neq 0 .
\end{align*}
\end{proof}

\begin{cor}\label{cor: doublet diagram}
The vector $\Bdryvec_{\omega} = R_+ \big( \Projection(v_{\alpha(\omega)}) \big)$ 
satisfies the conditions~\eqref{eq: bdry visit doublet projection conditions}.
\end{cor}
\begin{proof}
The conditions~\eqref{eq: bdry visit doublet projection conditions}
for the vector $\Bdryvec_{\omega}$
concern projections $\hat{\pi}^{(2)}_\pm$ acting on $\Wd_3\tens\Wd_2$ and 
$\Wd_2\tens\Wd_3$ in the middle of the tensor 
product~\eqref{eq: bdry visit tensor product}. The linear isomorphism $R_+$
commutes with these projections, by
\cite[Eq.~(5.2)]{KP-conformally_covariant_boundary_correlation_functions_with_a_quantum_group}.
Therefore, it suffices to consider the corresponding 
projection conditions for $\Bdryvec_{\omega}^{\infty}$.
By Lemma~\ref{lem: pure geometries lie in projected space}, we have
$v_{\alpha(\omega)}=\Embedding(\Bdryvec_{\omega}^{\infty})$.
The projections $\hat{\pi}^{(2)}_\pm$ acting on $\Bdryvec_{\omega}^{\infty}$
can be calculated using the right columns of the commutative diagrams in 
Lemma~\ref{lem: doublet diagram}, 
and the projections $\hat{\pi}^{(1)}_{2L}$ and $\hat{\pi}^{(1)}_{2L+1}$,
acting on 
$v_{\alpha(\omega)}$, by the left columns. 
The assertion follows by observing that, by
Remark~\ref{rem: collapsing visits corresponds to removing links}
and \eqref{eq: multiple sle projection conditions}, we have
\begin{align*}
\hat{\pi}^{(1)}_{2L}(v_{\alpha(\omega)})=\; &\begin{cases}
0\quad & \text{if }\omega_1=+\\
v_{\alpha(\hat{\omega})} & \text{if }\omega_1=-
\end{cases}\qquad\qquad\Rightarrow\qquad\qquad\hat{\pi}^{(2)}_-(\Bdryvec_{\omega})=\; &\begin{cases}
0\quad & \text{if }\omega_1=+\\
C_2\times\Bdryvec_{\hat{\omega}} & \text{if }\omega_1=-
\end{cases}\\
\hat{\pi}^{(1)}_{2L+1}(v_{\alpha(\omega)})=\; &\begin{cases}
v_{\alpha(\hat{\omega})} & \text{if }\omega_1=+\\
0\quad & \text{if }\omega_1=-
\end{cases}\qquad\qquad\Rightarrow\qquad\qquad\hat{\pi}^{(2)}_+(\Bdryvec_{\omega})=\; &\begin{cases}
C_2\times\Bdryvec_{\hat{\omega}} & \text{if }\omega_1=+\\
0\quad & \text{if }\omega_1=-,
\end{cases}
\end{align*}
where $\hat{\omega}=(\omega_2,\omega_3,\ldots\omega_{N'})$ is
the order obtained from $\omega$ by collapsing the first visit.
See also Figure~\ref{fig: removing first visited point}.
\end{proof}


Let $C_3 = \frac{\qnum{2}^2}{q^2 + q^{-2}}$, 
and note that $C_3 \neq 0$, since $q$ is not a root of unity.
\begin{lem}\label{lem: triplet diagram}
For any $v \in \Wd_3 \tens \Wd_3$, we have
$\hat{\pi}^{(3)}(v) = C_3 \times \big( \hat{\pi}^{(3)} \circ (\id \tens \hat{\pi}^{(1)} \tens \id) \circ (\iota^{(3)}\tens\iota^{(3)}) \big) (v)$.
In other words, the following diagram commutes, up to 
the non-zero multiplicative constant $C_3$:
\begin{align*}\xymatrix{
   \Wd_{2} \tens \Wd_{2} \tens \Wd_{2} \tens \Wd_{2} \; \; \ar@{<-^{)}}[rr]^{\hspace{.9cm}\iota^{(3)} \, \tens\, \iota^{(3)}}\ar[d]_{\id\,\tens \,\hat{\pi}^{(1)}  \tens\, \id}
  & & \; \; \Wd_{3} \tens \Wd_{3} \ar[dd]^{\hat{\pi}^{(3)}} \\
    \Wd_{2} \tens \Wd_{2} \; \; \ar[d]_{\hat{\pi}^{(3)}}
  & & \; \; \\
    \Wd_{3} \; \; \ar@{<->}[rr]_{\isom}
  & & \; \;  \Wd_{3}  \; \; .
} \end{align*}
\end{lem}
\begin{proof}
The proof is similar to the proof of 
Lemma~\ref{lem: doublet diagram}. One uses Schur's lemma and concludes by
calculating that the vector $\Tbas_{0}^{(3;3,3)}$ maps to a non-zero multiple
of $\Wbas_0^{(3)}\in\Wd_3$ in the various mappings.
\end{proof}

\begin{cor}\label{cor: triplet diagram}
The vector $\Bdryvec_{\omega} = R_+ \big( \Projection(v_{\alpha(\omega)}) \big)$ 
satisfies the conditions~\eqref{eq: bdry visit triplet projection conditions}.
\end{cor}
\begin{proof}
The conditions~\eqref{eq: bdry visit triplet projection conditions}
for the vector $\Bdryvec_{\omega}$
concern projections $\hat{\pi}^{(3)}_{\epsilon;m}$ for the $m$:th and
$m+1$:st points on the right ($\epsilon=+$) or left ($\epsilon=-$),
acting on an appropriate pair $\Wd_3\tens\Wd_3$ of consecutive 
tensor components of \eqref{eq: bdry visit tensor product}.
Again, it suffices to prove the corresponding conditions for $\Bdryvec_{\omega}^{\infty}$.
The projections $\hat{\pi}^{(3)}_{\epsilon;m}$ acting on $\Bdryvec_\omega^{\infty}$
can be calculated using the right column of the commutative diagram in 
Lemma~\ref{lem: triplet diagram},
and the corresponding projections $\hat{\pi}^{(1)}_{k(\epsilon;m)}$,
acting on $v_{\alpha(\omega)}$,
using the left column
(the index $k(\epsilon;m)$ is determined by 
$\epsilon$ and $m$).
The assertion follows by observing that, by
Remark~\ref{rem: collapsing visits corresponds to removing links}
and \eqref{eq: multiple sle projection conditions}, we have
\begin{align*}
\hat{\pi}^{(1)}_{k(\epsilon;m)}(v_{\alpha(\omega)})=\; &\begin{cases}
0 & \mbox{in the case of non-successive visits }\\ v_{\alpha(\hat{\omega})} & \mbox{in the case of successive visits }\end{cases}\\
\qquad\qquad\Rightarrow\qquad
\hat{\pi}^{(3)}_{\epsilon;m}(\Bdryvec_\omega)=\; &\begin{cases}
0 & \mbox{in the case of non-successive visits}\\ C_3\times\Bdryvec_{\hat{\omega}} & \mbox{in the case of successive visits}\end{cases}
\end{align*}
where $\hat{\omega}$ is the order obtained from $\omega$ by collapsing 
these successive visits.
See also Figure~\ref{fig: removing one of consecutively visited points}.
\end{proof}


\begin{lem}\label{lem: singlet diagram}
For any $v \in \Wd_3 \tens \Wd_3$, we have
$\hat{\pi}^{(1)}(v) = C \times \big( \hat{\pi}^{(1)} \circ (\id \tens \hat{\pi}^{(1)} \tens \id) \circ (\iota^{(3)}\tens\iota^{(3)}) \big) (v)$.
In other words, the following diagram commutes, up to 
the non-zero multiplicative constant $C = \frac{\qnum{2}^3}{\qnum{3}}$:
\begin{align*}\xymatrix{
   \Wd_{2} \tens \Wd_{2} \tens \Wd_{2} \tens \Wd_{2} \; \; \ar@{<-^{)}}[rr]^{\hspace{.9cm}\iota^{(3)} \, \tens\, \iota^{(3)}}\ar[d]_{\id\,\tens \,\hat{\pi}^{(1)}  \tens\, \id}
  & & \; \; \Wd_{3} \tens \Wd_{3} \ar[dd]^{\hat{\pi}^{(1)}} \\
    \Wd_{2} \tens \Wd_{2} \; \; \ar[d]_{\hat{\pi}^{(1)}}
  & & \; \; \\
    \bC \; \; \ar@{<->}[rr]_{\isom}
  & & \; \;  \bC \; \; .
} \end{align*}
\end{lem}
\begin{proof}
The proof is similar to the proof of 
Lemma~\ref{lem: doublet diagram}. One uses Schur's lemma and concludes by
calculating that the vector $\Tbas_{0}^{(1;3,3)}$ maps to a non-zero multiple
of $1\in\bC\isom\Wd_1$ in the various mappings.
\end{proof}

\begin{cor}\label{cor: singlet diagram}
The vector $\Bdryvec_{\omega} = R_+ \big( \Projection(v_{\alpha(\omega)}) \big)$ 
satisfies the conditions~\eqref{eq: bdry visit singlet projection conditions}.
\end{cor}
\begin{proof}
The conditions~\eqref{eq: bdry visit singlet projection conditions}
for the vector $\Bdryvec_{\omega}$
concern projections $\hat{\pi}^{(1)}_{\epsilon;m}$ for the $m$:th and
$m+1$:st points on the right ($\epsilon=+$) or left ($\epsilon=-$),
acting on an appropriate pair $\Wd_3\tens\Wd_3$ of consecutive 
tensor components of \eqref{eq: bdry visit tensor product}. 
Again, it suffices to prove the corresponding conditions for $\Bdryvec_{\omega}^{\infty}$.
The projections $\hat{\pi}^{(1)}_{\epsilon;m}$ acting on $\Bdryvec_\omega^{\infty}$
can be calculated using the right column of the commutative diagram in 
Lemma~\ref{lem: singlet diagram},
and the corresponding projections $\hat{\pi}^{(1)}_{k'(\epsilon;m)}\circ\hat{\pi}^{(1)}_{k(\epsilon;m)}$,
acting on $v_{\alpha(\omega)}$,
using the left column 
(the indices $k(\epsilon;m)$ and $k'(\epsilon;m)$ are determined 
by the indices $\epsilon$ and $m$).
The link pattern $\alpha(\omega)$
cannot contain the two nested links $\link{k}{k+1}$
and $\link{k-1}{k+2}$ with $k=k(\epsilon;m)$.
Hence, by \eqref{eq: multiple sle projection conditions}, we have
\begin{align*}
\hat{\pi}^{(1)}_{k'(\epsilon;m)}
(\hat{\pi}^{(1)}_{k(\epsilon;m)}(v_{\alpha(\omega)}))=0.
\end{align*}
\end{proof}

\subsection{\label{subsub: bdry visit uniqueness}Uniqueness of solutions}

The proof of the uniqueness of solutions of the 
system~\eqref{eq: bdry visit cartan eigenvalue}~--~\eqref{eq: bdry visit doublet projection conditions}
is similar to the multiple $\SLE$ case
(Proposition~\ref{prop: uniqueness}) --- 
the homogeneous system 
only admits the trivial solution.
We make use of the following lemma, which is 
a generalization of Corollary~\ref{cor: all projections vanish gives zero}.
\begin{lem}\label{lem: all projections vanish gives zero for bdry visits}
Let $L,R\in\bZnn$, $L+R\geq1$, and assume that the vector 
$v\in\Wd_{3}^{\tens R}\tens\Wd_2\tens\Wd_{3}^{\tens L}$ satisfies 
$\,E.v=0\,$, $\,K.v=q\,v\,$, and 
$\,\hat{\pi}^{(2)}_\pm(v)=0\,$, and
$\,\hat{\pi}^{(3)}_{\epsilon;m}(v)=0\,$,
$\,\hat{\pi}^{(1)}_{\epsilon;m}(v)=0\,$
for all indices $m$ and $\epsilon=\pm$. Then we have $v=0$.  
\end{lem}
\begin{proof}
The conditions $E.v=0$, $K.v=q\,v$ show that 
$v\in\HWsp_2(\Wd_{3}^{\tens R}\tens\Wd_2\tens\Wd_{3}^{\tens L})$.
Denote by 
\begin{align*}
\Projection'=\;&(\hat{\pi}^{(3)})^{\tens R}\tens\id\tens(\hat{\pi}^{(3)})^{\tens L}\;\colon\;\Wd_2^{\tens (2N-1)}\longrightarrow\Wd_3^{\tens R} \tens \Wd_2 \tens \Wd_3^{\tens L},\\
\Embedding'=\;&(\iota^{(3)})^{\tens R}\tens\id\tens(\iota^{(3)})^{\tens L}\;\colon\;\Wd_3^{\tens R} \tens \Wd_2 \tens \Wd_3^{\tens L}\hookrightarrow\Wd_2^{\tens (2N-1)},\\
v'=\;&\Embedding'(v) \; \in \; \HWsp_2(\Wd_{2}^{\tens (2N-1)})
    := \set{v \in \Wd_{2}^{\tens (2N-1)} \; \big| \; E.v =0 , \; K.v = q v } .
\end{align*}
The assumptions $\,\hat{\pi}^{(2)}_\pm(v)=0\,$,
$\,\hat{\pi}^{(3)}_{\epsilon;m}(v)=0\,$,
$\,\hat{\pi}^{(1)}_{\epsilon;m}(v)=0\,$
for all indices $m$ and $\epsilon$ imply that
in the direct sum decomposition of 
any two consecutive tensorands
($\Wd_3\tens\Wd_3$, $\Wd_2\tens\Wd_3$ or $\Wd_3\tens\Wd_2$) of
the tensor product~\eqref{eq: bdry visit tensor product}
into irreducibles, the vector $v$ lies in the
highest dimensional subrepresentation.
An application of Lemma~\ref{lem: all projections vanish}
to each pair of two consecutive tensorands embedded in a tensor product 
of two dimensional representations
($\Wd_3\tens\Wd_3$ via $\iota^{(3)}\tens\iota^{(3)}$, $\Wd_2\tens\Wd_3$ 
via $\id\tens\iota^{(3)}$, and $\Wd_3\tens\Wd_2$ via $\iota^{(3)}\tens\id$) 
shows that $\hat{\pi}^{(1)}_j(v')=0$ 
for all $1\leq j\leq 2N-2$.
Therefore, by Lemma~\ref{lem: all projections vanish}
applied to $\Wd_{2}^{\tens (2N-1)}$, 
the vector $v'$ belongs to the highest dimensional subrepresentation
$\Wd_{2N}\subset\Wd_2^{\tens(2N-1)}$.
Hence, we have $v'\in\Wd_{2N}\cap\Wd_2=\{0\}$, as $N=R+L+1\geq 2$.
We conclude by $v=\Projection'(v')=0$.
\end{proof}


\subsection{\label{sub: bdry visit existence}Proof of Theorem~\ref{thm: bdry visit vectors}}

Solutions to~\eqref{eq: bdry visit cartan eigenvalue}~--~\eqref{eq: bdry visit doublet projection conditions} were constructed in 
Section~\ref{subsub: bdry visit construction} --- see 
Remark~\ref{rem: doublet conditions} and 
Corollaries~\ref{cor: doublet diagram}, \ref{cor: triplet diagram} and
\ref{cor: singlet diagram}.
Uniqueness of normalized solutions follows from 
Lemma~\ref{lem: all projections vanish gives zero for bdry visits}
similarly as in the proof of Proposition~\ref{prop: uniqueness}:
the difference of any two solutions of 
\eqref{eq: bdry visit cartan eigenvalue}~--~\eqref{eq: bdry visit doublet projection conditions}
vanishes as a solution to the homogeneous system,
which appears in the statement of the lemma.
$\hfill\qed$

\bigskip{}


\appendix

\section{\label{app: Local multiple SLEs}Local multiple $\SLE$s}

The key technique to construct $\SLE$ type curves is their description as growth processes
\cite{Schramm-LERW_and_UST}. Growth processes for multiple curves, however,
straightforwardly only allow to construct initial segments of the curves.
In this article we restrict our attention to such local multiple $\SLE$s.
In extending the definition to a probability measure on $N$ globally
defined random curves connecting $2N$ boundary points, one encounters
technical difficulties similar to the challenges in proving the reversibility property of
a single chordal $\SLE$ curve \cite{Zhan-reversibility,MS-imaginary_geometry_3}.

\subsection{\label{subsec: Schramm-Loewner evolutions}Chordal Schramm-Loewner evolution}

The simplest $\SLE$ variant is the chordal $\SLEk$.
Let $\domain \subset \bC$ be an open simply connected domain,
with two distinct boundary points $\bdrypt,\bdryptb \in \bdry \domain$ (prime ends).
The chordal $\SLEk$ in $\domain$ from $\bdrypt$ to $\bdryptb$
is a random curve --- more precisely, a probability measure $\SLEmeasure^{(\domain;\bdrypt, \bdryptb)}_\chordal$ on
oriented but unparametrized non-self-crossing curves in $\cl{\domain}$ from $\bdrypt$ to $\bdryptb$
(the space of such curves is equipped with a natural metric inherited
from a uniform norm on parametrized curves).
We often choose some parametrized curve $\gamma \colon [0,1] \to \cl{\domain}$,
to be interpreted as its equivalence class under increasing
reparametrizations.
The chordal $\SLEk$ itself is the family $\big(\SLEmeasure^{(\domain;\bdrypt, \bdryptb)}_\chordal \big)_{\domain,\bdrypt, \bdryptb}$
of these probability measures, indexed by the domain and marked boundary points.
Schramm's observation was that, up to the value of the parameter $\kappa>0$, the family is
characterized by the following two assumptions.
\begin{itemize}
\item {\em Conformal invariance}: If $\confmap \colon \domain \to \domain'$ is
a conformal map, 
and the curve $\gamma$ in $\domain$ has the law $\SLEmeasure^{(\domain;\bdrypt, \bdryptb)}_\chordal$,
then the image $\confmap \circ \gamma$ has the law $\SLEmeasure_\chordal^{(\confmap(\domain);\confmap(\bdrypt),\confmap(\bdryptb))}$.
A concise way to state this is that the measures are related by pushforwards,
$\confmap_* \SLEmeasure_\chordal^{(\domain;\bdrypt, \bdryptb)} =
\SLEmeasure_\chordal^{(\confmap(\domain);\confmap(\bdrypt),\confmap(\bdryptb))}$.
\item {\em Domain Markov property}: Conditionally, given an
initial segment $\gamma |_{[0,\tau]}$ of the curve $\gamma$ with law $\SLEmeasure_\chordal^{(\domain;\bdrypt, \bdryptb)}$,
the remaining part $\gamma |_{[\tau,1]}$ 
has the law $\SLEmeasure_\chordal^{(\domain';\gamma(\tau),\bdryptb)}$, where
$\domain'$ is the component of $\domain \setminus \gamma[0,\tau]$
containing $\bdryptb$ on its boundary.
\end{itemize}

By Riemann mapping theorem, between any triples
$(\domain;\bdrypt, \bdryptb)$ and $(\domain';\bdrypt',\bdryptb')$
there exists a conformal map 
$\confmap \colon \domain \to \domain'$ such that $\confmap(\bdrypt)=\bdrypt'$, $\confmap(\bdryptb) = \bdryptb'$.
By conformal invariance, it therefore suffices to construct the chordal $\SLEk$ in
one reference domain, e.g., the upper half-plane
\[ \bH = \set{ z \in \bC \; \big| \; \im(z) > 0} . \]

The chordal $\SLEk$ in $(\domain;\bdrypt, \bdryptb)=(\bH;0,\infty)$ is constructed by a
growth process encoded in a Loewner chain $(g_t)_{t \in [0,\infty)}$ as follows
--- see \cite{RS-basic_properties_of_SLE} for details.
Let $(B_t)_{t \in[0,\infty)}$
be a standard Brownian motion on the real line, and for $z \in \bH$,
consider the solution to the Loewner differential equation
\begin{align}\label{eq: Loewner equation}
\frac{\ud}{\ud t} g_t(z) = \frac{2}{g_t(z) - \driving_t},\qquad g_0(z)=z,
\end{align}
with the driving function $\driving_t=\sqrt{\kappa}B_t$.
The solution is defined up to a (possibly infinite) explosion time $T_z$. 
Let $K_t$ be the closure of the set $\set{z\in\bH\;|\;T_z<t}$.
The growth process $(K_t)_{t\in[0,\infty)}$ is 
generated by a random curve $\gamma \colon [0,\infty) \to \overline{\bH}$ 
in the sense that $\bH\setminus K_t$ is the 
unbounded component of $\bH\setminus\gamma[0,t]$.
This curve is (a parametrization of) the chordal $\SLEk$. 
For fixed $t \in [0,\infty)$, the solution to \eqref{eq: Loewner equation}
viewed as a function of the initial condition $z$ 
gives the unique conformal map
$g_t\colon\bH\setminus K_t\rightarrow\bH$
normalized so that $g_t(z)=z+\oo(1)$ as $z\to\infty$.

\subsection{\label{subsec: Definition of local multiple SLEs}Local multiple $\SLE$s}


Let $\domain\subsetneq\bC$ be a simply connected domain, and let
$\bdrypt_1, \ldots, \bdrypt_{2N} \in \bdry \domain$ be $2N$ distinct boundary points
appearing in counterclockwise order along $\bdry \domain$.
Morally, we would like to associate to the domain 
and boundary points 
a probability measure $\SLEmeasure^{(\domain;\bdrypt_1,\ldots,\bdrypt_{2N})}$ on a collection of $N$ curves,
connecting the $2N$ marked boundary points. Instead, a local multiple $\SLE$ will describe $2N$
initial segments $\gamma^{(j)}$ of curves starting from the points $\bdrypt_j$,
up to exiting some neighborhoods $U_j \ni \bdrypt_j$. 
The localization neighborhoods $U_1,\ldots,U_{2N}$ are assumed to be
closed subsets of $\cl{\domain}$ such that 
$\domain \setminus U_j$ are simply connected and $U_j \cap U_k = \emptyset$ for $j \neq k$. See Figure~\ref{fig: local multiple SLE} for an illustration of 
the localization.

The local $N$-$\SLEk$ in $\domain$, started from $(\bdrypt_1, \ldots, \bdrypt_{2N})$ and localized in $(U_1, \ldots, U_{2N})$,
is a probability measure on $2N$-tuples of oriented unparametrized 
curves with parametrized representatives 
$(\gamma^{(1)},\ldots,\gamma^{(2N)})$,
such that for each $j$, the curve $\gamma^{(j)} \colon [0,1] \to U_j$ starts
at $\gamma^{(j)}(0)=\bdrypt_j$ and ends at
$\gamma^{(j)}(1) \in \bdry (\domain \setminus U_j)$ on the boundary of the localization neighborhood. The local $N$-$\SLEk$ itself
is the indexed collection
\begin{align*}
\SLEmeasure 
= \left(\SLEmeasure^{(\domain;\bdrypt_1,\ldots,\bdrypt_{2N})}_{(U_1,\ldots,U_{2N})}\right)_{\domain;\bdrypt_1,\ldots,\bdrypt_{2N};U_1,\ldots,U_{2N}} .
\end{align*}
This collection
of probability measures is required to satisfy 
conformal invariance, domain Markov property, and absolute continuity
of marginals with respect to the chordal $\SLEk$:
\begin{description}
\medskip
\item [(CI)] 
If $\confmap:\domain\rightarrow \domain'$ is a conformal map, then the measures are related by pushforward,
\begin{align*}
\confmap_*\SLEmeasure^{(\domain;\bdrypt_1,\ldots,\bdrypt_{2N})}_{(U_1,\ldots,U_{2N})} 
= \SLEmeasure^{(\confmap(\domain);\confmap(\bdrypt_1),\ldots,\confmap(\bdrypt_{2N}))}_{(\confmap(U_1),\ldots,\confmap(U_{2N}))} .
\end{align*}
\item [(DMP)]
Conditionally, given 
initial segments $\gamma^{(j)} |_{[0,\tau_j]}$ of the curves $(\gamma^{(1)},\ldots,\gamma^{(2N)})$
with law $\SLEmeasure^{(\domain;\bdrypt_1,\ldots,\bdrypt_{2N})}_{(U_1,\ldots,U_{2N})}$,
the remaining parts $(\gamma^{(1)}|_{[\tau_1,1]},\ldots,\gamma^{(2N)} |_{[\tau_{2N},1]})$ 
have the law $\SLEmeasure^{(\domain';\bdrypt_1',\ldots,\bdrypt_{2N}')}_{(U'_1,\ldots,U'_{2N})}$, where 
$\domain'$ is the component of $\domain \setminus \bigcup_{j=1}^{2N} \gamma^{(j)}[0,\tau_j]$
containing all tips $\bdrypt_j'=\gamma^{(j)}(\tau_j)$ on its boundary, and $U'_j = U_j \cap \domain'$.
\medskip
\item [(MARG)]
There exist smooth functions $b^{(j)}:\chamber_{2N}\to\bR$, for $j=1,\ldots,2N$, such that
for the domain $\domain = \bH$, boundary points $x_1 < \ldots < x_{2N}$, and their localization neighborhoods $U_1, \ldots,U_{2N}$,
the marginal on the $j$:th curve $\gamma^{(j)}$ under 
$\SLEmeasure^{(\bH;x_1,\ldots,x_{2N})}_{(U_1,\ldots,U_{2N})}$
is the following. 
Consider the Loewner equation~\eqref{eq: Loewner equation}
with driving process $(X_t)_{t \in [0,\sigma)}$ that solves the system of It\^o SDEs 
\begin{align}
\ud X_t=&\;\sqrt{\kappa}\,\ud B_t\,+\,b^{(j)}(X^{(1)}_t,\ldots,X^{(j-1)}_t,X_t,X^{(j+1)}_t\ldots,X^{(2N)}_t)\,\ud t\label{eq: SDE}\\
\ud X^{(i)}_t=&\;\frac{2\,\ud t}{X^{(i)}_t-X_t}\qquad\text{for}\;i\neq j\nonumber ,
\end{align}
where 
$X_0 = x_j$ and $X^{(i)}_0=x_i$ for $i \neq j$. 
As in the case of the chordal $\SLE{}_\kappa$, the solution ${g_t \colon \bH \setminus K_t \to \bH}$ to \eqref{eq: Loewner equation}
is a conformal map, and the growth process $(K_t)_{t\in[0,\sigma)}$ is generated by a curve
$\gamma \colon [0,\sigma) \to \cl{\bH}$. The processes \eqref{eq: SDE} are defined at least up to the
stopping time $\sigma_j = \inf\set{t \geq 0 \; \big| \; \gamma \in \bdry (\bH \setminus U_j)}$.
The marginal law of $\gamma^{(j)}$ is that of the random curve $\gamma |_{[0,\sigma_j]}$.
\end{description}
\begin{figure}
\includegraphics[scale=.7]{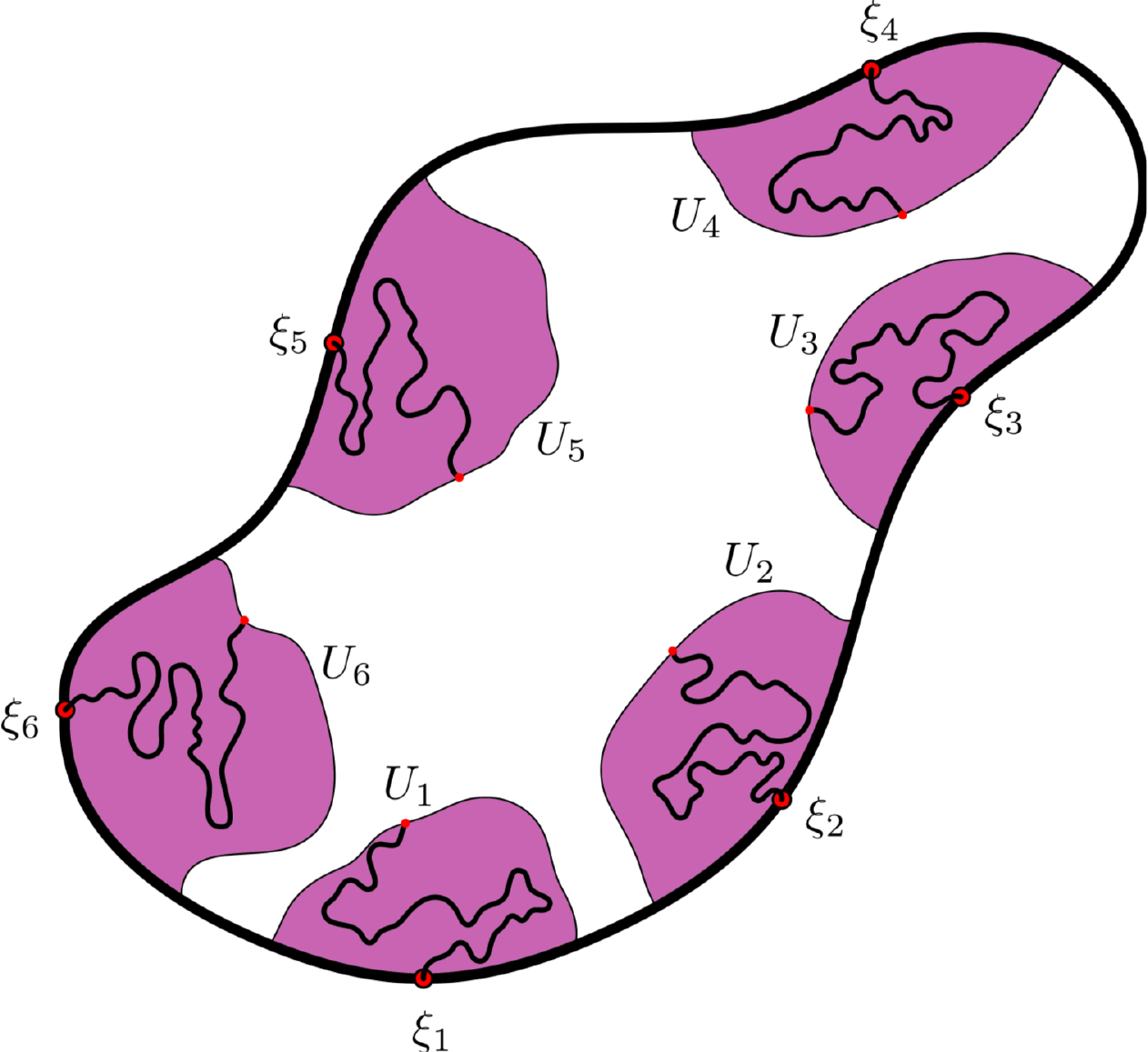}
\caption{\label{fig: local multiple SLE}
Schematic illustration of a local multiple $\SLE$.
}
\end{figure}

We will use the following result of Dub\'edat.
\begin{prop}[{\cite{Dubedat-commutation}}]\label{prop: Dubedats partition function}
If $\SLEmeasure = \big( \SLEmeasure^{(\domain;\bdrypt_1,\ldots,\bdrypt_{2N})}_{(U_1,\ldots,U_{2N})} \big)$ 
satisfies the conditions (CI), (DMP), and (MARG), then there exists a function $\PartF \colon \chamber_{2N} \to \bRpos$
satisfying the partial differential equations~\eqref{eq: multiple SLE PDEs}, such that the drift functions 
in (MARG) take the form 
\[ b^{(j)} = \kappa\,\frac{\partial_j\PartF}{\PartF} \qquad\text{for}\;j=1,\ldots,2N. \]
\end{prop}
\begin{proof}
Given the local multiple $\SLEk$ $\SLEmeasure$, we may take $\domain=\bH$ and
any $(x_1,\ldots,x_{2N})\in\chamber_{2N}$.  It follows from \cite[Theorem~7 and the remark after it]{Dubedat-commutation}
that there exists a function $\PartF$ satisfying~\eqref{eq: multiple SLE PDEs}
such that the drifts in~\eqref{eq: SDE} are
$b^{(j)}=\kappa\,\frac{\partial_j\PartF}{\PartF}$. 
The solution $\PartF$ is a priori defined in a neighborhood of the 
starting point $(x_1,\ldots,x_{2N})$, and determined only up to a 
multiplicative constant. By considering different starting points,
we see that such a function  $\PartF\colon\chamber_{2N}\rightarrow\bRpos$ exists on the entire chamber.
\end{proof}

We note that the localizations satisfy the following
consistency property under restriction to smaller localization neighborhoods.
This consistency allows always continuing the curves by a little,
but it is not sufficient for the definition of a global multiple $\SLE$.
\begin{prop}\label{prop: restriction consistency}
Suppose that both $(U_1,\ldots,U_{2N})$ and $(V_1,\ldots,V_{2N})$ are localization neighborhoods for $(\domain;\bdrypt_1,\ldots,\bdrypt_{2N})$,
and that $V_j \subset U_j$ for each $j$. Let $(\gamma^{(1)},\ldots,\gamma^{(2N)})$ be representatives of curves
with law $\SLEmeasure^{(\domain;\bdrypt_1,\ldots,\bdrypt_{2N})}_{(U_1,\ldots,U_{2N})}$, and
let $\sigma_j = \inf\set{t \geq 0 \; \big| \; \gamma^{(j)}(t) \in \bdry (\domain \setminus V_j)}$ be their exit times from the
smaller neighborhoods. Then $(\gamma^{(1)}|_{[0,\sigma_1]},\ldots,\gamma^{(2N)}|_{[0,\sigma_{2N}]})$ are representatives of curves
with law $\SLEmeasure^{(\domain;\bdrypt_1,\ldots,\bdrypt_{2N})}_{(V_1,\ldots,V_{2N})}$.
\end{prop}
\begin{proof}
The marginal on the $j$:th curve described in (MARG) clearly satisfies such a restriction consistency.
The consistency for all curves follows inductively by the domain Markov property (DMP).
\end{proof}

\subsection{\label{subsec: Sampling one by one from marginals}Sampling one by one from marginals}
Property (MARG) describes the marginal law of one of the $2N$ curves, and property (DMP)
gives the conditional law of the others, given one. One thus obtains a procedure for sampling a local
multiple $\SLEk$. In this section, we first describe the marginal law of $\gamma^{(j)}$ by 
its Radon-Nikodym derivative with respect to the chordal $\SLEk$, explicitly expressed in terms of $\PartF$.
We then formalize the sampling procedure, which is now meaningful with just the knowledge of $\PartF$.

Recall that by Proposition~\ref{prop: Dubedats partition function}, the drift functions
$b^{(j)}$ are necessarily of the form $b^{(j)} = \kappa\,\frac{\partial_j\PartF}{\PartF}$.
Let $(X_t)_{t \in [0,\tau]}$ and $\big(X^{(i)}_t\big)_{t\in[0,\tau]}$ for $i \neq j$ 
solve the SDEs~\eqref{eq: SDE} with these drifts, and let $\gamma\colon[0,\tau]\rightarrow\cl{\bH}$ be 
the random curve with the Loewner driving function $X_t$ as in (MARG),
and $\tau$ any stopping time such that, for some $\varepsilon>0$ and all $0\leq t\leq\tau$, we have
\begin{align}\label{eq: condition for stopping time}
|X^{(i)}_t-X^{(k)}_t| \geq \varepsilon\qquad\text{for all}\quad i\neq k\qquad \text{ and } \qquad
|X^{(i)}_t-X_t| \geq \varepsilon\qquad\text{for all}\quad i .
\end{align}
Then, the law $\SLEmeasure_\gamma$ of the curve $\gamma$ is absolutely continuous with respect to
an initial segment of the chordal $\SLEk$, with Radon-Nikodym derivative
\begin{align}\label{eq: marginals}
\frac{\ud\SLEmeasure_\gamma}{\ud\SLEmeasure_{\chordal}^{(\bH;X_0,\infty)}}
=\,\prod_{i\neq j}(g'(X^{(i)}_0))^h\;\times\;
    \frac{\PartF \big( g(X^{(1)}_0),\ldots,g(X^{(j-1)}_0),g(\gamma(\tau)),g(X^{(j+1)}_0),\ldots,g(X^{(2N)}_0) \big)}{\PartF\big( X^{(1)}_0,\ldots,X^{(j-1)}_0,X_0,X^{(j+1)}_0,\ldots,X^{(2N)}_0 \big)},
\end{align}
where $g:\bH\setminus K\to\bH$ is the 
conformal map such that $g(z) = z + \oo(1)$ as $z\to\infty$, $K$ is the hull of $\gamma$
and $h=\frac{6-\kappa}{2\kappa}$. 
In fact, $\gamma$ is obtained by Girsanov re-weighting of the chordal $\SLEk$ by the martingale
\begin{align}
t\;\mapsto&\;M_t\label{eq: martingale}\\
=&\prod_{i\neq j}(g_t'(X^{(i)}_0))^h\;\times\;
    \PartF \Big(g_t(X^{(1)}_0),\ldots,g_t(X^{(j-1)}_0),\sqrt{\kappa}B_t+X_0,g_t(X^{(j+1)}_0),\ldots,g_t(X^{(2N)}_0) \Big)\nonumber
\end{align}
where $g_t$ is the solution to the Loewner equation~\eqref{eq: Loewner equation}
with the driving function $\driving_t = \sqrt{\kappa}B_t+X_0$.

In view of the above, given 
points $X^{(1)}_0 < \cdots < X^{(j-1)}_0 < X_0 < X^{(j+1)}_0 < \cdots < X^{(2N)}_0$,
a neighborhood $U$ of $X_0$ not containing the points $X^{(i)}_0$,
and a positive function $\sZ$, we construct a random curve $\gamma$ by
the weighting \eqref{eq: marginals} of an initial segment of the chordal $\SLEk$ up
to the stopping time
\begin{align*}
\tau=\sigma^{(U)}=\inf\set{t>0\;\big|\;\gamma(t)\in\bdry (\bH \setminus U) }.
\end{align*}
Note that this stopping time 
satisfies the condition~\eqref{eq: condition for stopping time}, by standard harmonic measure estimates.

\begin{proc}\label{proc: one by one}
Given a positive function $\PartF\colon\chamber_{2N}\rightarrow\bRpos$, 
a random $2N$-tuple $(\gamma^{(1)},\ldots,\gamma^{(2N)})$ of curves 
in $\bH$ starting from $(x_1,\ldots,x_{2N})$ with localization neighborhoods 
$(U_1,\ldots,U_{2N})$ 
can be sampled according to the following procedure.
\begin{itemize}
\item Select an order in which the curves will be sampled, encoded in a permutation $p \in \SymmGrp_{2N}$.
\item Sample the first curve in the selected order, $\gamma^{(j)}$ for $j=p(1)$,
from the measure~\eqref{eq: marginals} with $X_0=x_j$ and $X_0^{(i)} = x_i$ for $i \neq j$, and with
the stopping time $\tau = \sigma^{(U_j)}$.
\item Suppose the first $k-1$ curves 
$\gamma^{(p(1))}, \ldots, \gamma^{(p(k-1))}$
in the selected order have been sampled, and let $j=p(k)$.
Let $K$ be the hull of $\gamma^{(p(1))} \cup \cdots \cup \gamma^{(p(k-1))}$, 
and $G \colon \bH \setminus K \to \bH$
the conformal map such that $G(z) = z + \oo(1)$ as $z \to \infty$. 
Let $X_0 = G(x_j)$, $X^{(p(i))}_0 = G(x_i)$ for $i = k+1, \ldots, 2N$,
and let $X^{(p(l))}_0$ be the image of the tip of $\gamma^{(p(l))}$ 
for $l = 1, \ldots, k-1$.
Construct the $k$:th curve as $\gamma^{(j)} = G^{-1} \circ \gamma$ using
the curve $\gamma$ sampled from the measure~\eqref{eq: marginals}
 with the stopping time $\tau = \sigma^{(G(U_j))}$.
\end{itemize}
\end{proc}

By the local commutation of Dub\'edat, \cite{Dubedat-commutation},
this procedure results in the same law of $(\gamma^{(1)}, \ldots, \gamma^{(2N)})$ independently of the
sampling order $p$, provided that $\PartF$ is a solution to the system~\eqref{eq: multiple SLE PDEs}.

\subsection{\label{subsec: Partition function and PDEs}Partition function and PDEs}

We will next state the main theorem towards a construction of multiple $\SLE$
processes. The most profound parts of the proof rely on Dub\'edat's commutation of $\SLE$s \cite{Dubedat-commutation}.
In summary, the theorem states that 
multiple $\SLEk$ partition functions $\PartF$ 
correspond to local multiple $\SLEk$ processes, via the sampling procedure~\ref{proc: one by one},
and two partition functions correspond to the same local multiple $\SLEk$ if and only if they are proportional to each other.
More precisely, the set
\begin{align*}
\set{\PartF(x_1,\ldots,x_{2N})\;|\;\PartF\colon\chamber_{2N}\rightarrow\bRpos\;\text{is a positive solution to \eqref{eq: multiple SLE PDEs}~--~\eqref{eq: multiple SLE Mobius covariance}}\;}\;\big/\;\bRpos
\end{align*}
corresponds one-to-one with the set of local multiple $\SLE$s $\PR$.
Moreover, convex combinations of partition functions correspond to convex combinations of
localizations of multiple $\SLE$s as probability measures.

\begin{thm}\label{thm: local multiple SLEs}
\
\begin{description}
\item[(a)] Suppose $\PartF\colon\chamber_{2N}\rightarrow\bRpos$ is 
a positive solution to the system~\eqref{eq: multiple SLE PDEs}~--~\eqref{eq: multiple SLE Mobius covariance}. 
Then the random collection of curves obtained by 
Procedure~\ref{proc: one by one} is a local multiple $\SLEk$. 
Two functions $\PartF,\tilde{\PartF}$ give rise to the same local
multiple $\SLEk$ if and only if $\PartF=\const\times\tilde{\PartF}$.

\medskip

\item[(b)] Suppose $\SLEmeasure$ is a local multiple $\SLEk$. Then 
there exists a positive solution 
$\PartF\colon\chamber_{2N}\rightarrow\bRpos$ to the 
system~\eqref{eq: multiple SLE PDEs}~--~\eqref{eq: multiple SLE Mobius covariance}, 
such that for any $j=1,\ldots,2N$,
the drift in~\eqref{eq: SDE} is given by 
$b^{(j)}=\kappa\,\frac{\partial_j\PartF}{\PartF}$. 
Such a function $\PartF$ is determined up to a multiplicative constant.

\medskip

\item[(c)] Suppose that $\PartF_1$ and $\PartF_2$ are positive solutions to the 
system~\eqref{eq: multiple SLE PDEs}~--~\eqref{eq: multiple SLE Mobius covariance} and
\begin{align*}
\PartF=r\,\PartF_1+(1-r)\,\PartF_2,
\end{align*}
with $0\leq r\leq1$. 
Denote by $\;\SLEmeasure=\left(\SLEmeasure^{(\domain;\bdrypt_1,\ldots,\bdrypt_{2N})}_{(U_1,\ldots,U_{2N})}\right)$, 
$\SLEmeasure_1=\left((\SLEmeasure_1)^{(\domain;\bdrypt_1,\ldots,\bdrypt_{2N})}_{(U_1,\ldots,U_{2N})}\right)$ and 
$\SLEmeasure_2=\left((\SLEmeasure_2)^{(\domain;\bdrypt_1,\ldots,\bdrypt_{2N})}_{(U_1,\ldots,U_{2N})}\right)$ the 
local multiple $\SLE$s associated to $\PartF$, $\PartF_1$ and $\PartF_2$, 
respectively. Fix the domain $\domain$ and the marked points 
$(\bdrypt_1,\ldots,\bdrypt_{2N})$, and let 
$(U_1,\ldots,U_{2N})$ 
be any localization neighborhoods. Then the probability measure 
associated to $\PartF$ is obtained as the following convex combination
\begin{align}
\SLEmeasure^{(\domain;\bdrypt_1,\ldots,\bdrypt_{2N})}_{(U_1,\ldots,U_{2N})}\;=&\;r\,\frac{\PartF_1(\confmap(\bdrypt_1),\ldots,\confmap(\bdrypt_{2N}))}{\PartF(\confmap(\bdrypt_1),\ldots,\confmap(\bdrypt_{2N}))}\;(\SLEmeasure_1)^{(\domain;\bdrypt_1,\ldots,\bdrypt_{2N})}_{(U_1,\ldots,U_{2N})}\,\label{eq: convex combination}\\
&+\,(1-r)\,\frac{\PartF_2(\confmap(\bdrypt_1),\ldots,\confmap(\bdrypt_{2N}))}{\PartF(\confmap(\bdrypt_1),\ldots,\confmap(\bdrypt_{2N}))}\;(\SLEmeasure_2)^{(\domain;\bdrypt_1,\ldots,\bdrypt_{2N})}_{(U_1,\ldots,U_{2N})},\nonumber
\end{align}
where $\confmap\colon \domain\rightarrow\bH$ is any conformal map such that 
$\phi^{-1}(\infty)$ belongs to the positively oriented boundary segment 
between $\bdrypt_{2N}$ and $\bdrypt_1$.
\end{description}
\end{thm}
\begin{proof}
\
\begin{description}
\item[(a)] In $\bH$, for any marked points $(x_1,\ldots,x_{2N})$ and localization neighborhoods $(U_1,\ldots,U_{2N})$, 
Procedure~\ref{proc: one by one} defines a random $2N$-tuple 
$(\gamma^{(1)},\ldots,\gamma^{(2N)})$ of curves. From the assumption~\eqref{eq: multiple SLE PDEs}
and the local commutation  of~\cite[Lemma~6 \& Theorem~7]{Dubedat-commutation} it follows that
the law of $(\gamma^{(1)},\ldots,\gamma^{(2N)})$ is independent of the chosen sampling order $p$. 
From the assumption~\eqref{eq: multiple SLE Mobius covariance} and a standard coordinate change~\cite[Proposition~6.1]{Kytola-local_mgales}
it follows that the result of the sampling is M\"obius invariant, that is, (CI) holds for any $\confmap\colon\bH\rightarrow\bH$ which preserves the order
of the marked points.
Thus, we may use (CI) to define the laws of the random $2N$-tuples of curves in any other simply connected domain.
%
%
The condition (MARG) follows directly from Procedure~\ref{proc: one by one}, by choosing the $j$:th curve to be sampled first
to obtain~\eqref{eq: SDE}. Finally, the condition (DMP) follows from the local commutation 
of~\cite[Proposition~5 \& Lemma~6]{Dubedat-commutation}.

\medskip

\item[(b)]
Given the local multiple $\SLEk$ $\SLEmeasure$, the infinitesimal commutation
of \cite{Dubedat-commutation} summarized in Proposition~\ref{prop: Dubedats partition function}
implies the existence of a positive solution $\sZ$ 
to~\eqref{eq: multiple SLE PDEs} which determines
the drifts $b^{(j)}=\kappa\,\frac{\partial_j\PartF}{\PartF}$ in~\eqref{eq: SDE}.
The function $\sZ$ is determined up to a multiplicative constant.
By (CI) for M\"obius transformations $\confmap\colon\bH\rightarrow\bH$ and a standard $\SLE$ calculation
--- see \cite[Proposition~6.1]{Kytola-local_mgales} and \cite{Graham-multiple_SLEs} ---
the function $\PartF$ also satisfies
M\"obius covariance~\eqref{eq: multiple SLE Mobius covariance}.

\medskip

\item[(c)] It suffices to show 
the convex combination property~\eqref{eq: convex combination} for 
the marginal of a single curve $\gamma^{(j)}$, by (DMP), (CI), and the 
independence of Procedure~\ref{proc: one by one} of the sampling order.
To describe that marginal, we 
compare the conformal image $\gamma=\confmap\circ\gamma^{(j)}$ with the
chordal $\SLEk$ in $\bH$. The law of $\gamma$ under the localization 
$\SLEmeasure^{(\domain;\bdrypt_1,\ldots,\bdrypt_{2N})}_{(U_1,\ldots,U_{2N})}$ of $\SLEmeasure$
has Radon-Nikodym derivative of type~\eqref{eq: marginals} with 
respect to an initial segment of the chordal $\SLEk$. 

Denote by $t\mapsto M_t$, $M_t^{(1)}$ and $M_t^{(2)}$ the chordal $\SLEk$
martingales~\eqref{eq: martingale} associated to the partition functions
$\PartF$, $\PartF_1$ and $\PartF_2$, respectively, where
the starting points are $X_0 = \confmap(\bdrypt_j)$ and $X^{(i)}_0=\confmap(\bdrypt_i)$, for $i\neq j$.
The Radon-Nikodym derivative~\eqref{eq: marginals}
of the law of $\gamma$ can be written in terms of 
the martingales at a certain stopping time $\tau$ as
\begin{align*}
\frac{M_\tau}{M_0}=r\,\frac{\PartF_1(X^{(1)}_0,\ldots,X^{(2N)}_0)}{\PartF(X^{(1)}_0,\ldots,X^{(2N)}_0)}\times\frac{M^{(1)}_\tau}{M^{(1)}_0}+(1-r)\,\frac{\PartF_2(X^{(1)}_0,\ldots,X^{(2N)}_0)}{\PartF(X^{(1)}_0,\ldots,X^{(2N)}_0)}\times\frac{M^{(2)}_\tau}{M^{(2)}_0}.
\end{align*}
The ratios $M^{(1)}_\tau / M^{(1)}_0$ and $M^{(2)}_\tau / M^{(2)}_0$ are the corresponding Radon-Nikodym derivatives of $\gamma$ under the
localizations of $\SLEmeasure_1$ and $\SLEmeasure_2$, respectively.
The convex combination property~\eqref{eq: convex combination} follows.
\end{description}
\end{proof}

\subsection{\label{subsec: asymptotics of the partition functions}Asymptotics of the partition functions}
Theorem~\ref{thm: local multiple SLEs} explains
the requirements \eqref{eq: multiple SLE PDEs} and \eqref{eq: multiple SLE Mobius covariance}
for the multiple $\SLE$ partition functions $\sZ$.
Our objective is to construct the partition functions corresponding to
the extremal multiple $\SLE$s with deterministic connectivities described by link
patterns $\alpha$. All $\sZ_\alpha$, for $\alpha \in \LP_N$, are required to satisfy
the same partial differential equations \eqref{eq: multiple SLE PDEs} and covariance
\eqref{eq: multiple SLE Mobius covariance}, but the boundary conditions depend on $\alpha$,
as formulated in the asymptotics requirements~\eqref{eq: multiple SLE asymptotics}.
This section pertains to the probabilistic justification of these asymptotics requirements.

For the pure geometries of multiple $\SLE$s, we want the curves to meet pairwise according to a given connectivity $\alpha$.
For a local multiple $\SLE$, meeting of curves is not meaningful, but it suggests clear
requirements for the processes. Indeed, in terms of the
processes $X^{(1)}_t, \ldots, X^{(j-1)}_t,X_t,X^{(j+1)}_t, \ldots, X^{(2N)}_t$,
the differences $|X_t - X^{(i)}_t|$ quantify suitable conformal distances of the tip $\gamma^{(j)}(t)$ of the $j$:th curve to the 
marked points $X^{(i)}_0$ --- the difference is a renormalized limit of
the harmonic measure of the boundary segment between the tip and the marked point, seen from infinity.
Therefore, the connectivity $\alpha$ should determine whether or not it is possible for the difference
process $|X_t - X^{(i)}_t|$ to hit zero. As the drift $b^{(j)}$ of the process $X_t$ depends on $\sZ$
(Proposition~\ref{prop: Dubedats partition function}),
the possibility of hitting zero is in fact encoded in the asymptotics of $\sZ$.

Concerning the possible asymptotics, recall from Theorem~\ref{thm: SCCG correspondence special case} that
for the solutions $\sZ = \sF[v]$, the limit
\begin{align}\label{eq: generic asymptotics limit}
\lim_{x_j,x_{j+1} \to \xi} \frac{\sZ(x_1, \ldots, x_{2N})}{(x_{j+1} - x_j)^{\frac{\kappa-6}{\kappa}}}
\end{align}
always exists. If $\hat{\pi}_j(v) \neq 0$, then this limit is non-zero.
If $\hat{\pi}_j(v) = 0$, then the above limit vanishes, but a slightly more precise formulation of the
``spin chain~--~Coulomb gas correspondence'' of \cite{KP-conformally_covariant_boundary_correlation_functions_with_a_quantum_group}
shows that in that case the limit
\begin{align}\label{eq: subleading asymptotics limit}
\lim_{x_j,x_{j+1} \to \xi} \frac{\sZ(x_1, \ldots, x_{2N})}{(x_{j+1} - x_j)^{\frac{2}{\kappa}}}
\end{align}
exists. In fact, if $v \neq 0$, then exactly one of the limits above exists and is non-zero.
This property can also be derived directly from the partial differential equations \eqref{eq: multiple SLE PDEs},
see \cite[Theorem~2, part~1]{FK-solution_space_for_a_system_of_null_state_PDEs_4}.

Morally, in the two possible cases, the difference $|X_t - X^{(j+1)}_t|$
is locally absolutely continuous with respect to a Bessel process.
The dimension of the Bessel process is different depending on which of the limits
\eqref{eq: generic asymptotics limit} or \eqref{eq: subleading asymptotics limit} is non-zero,
so that hitting zero is possible in the former case and impossible in the latter --- see
Proposition~\ref{prop: local abs continuity wrt Bessel processes} below.
Conditions~\eqref{eq: multiple SLE asymptotics} express that the former case should happen
if $\link{j}{j+1} \in \alpha$, and the latter if $\link{j}{j+1} \notin \alpha$.
Moreover, in the former case the conditions~\eqref{eq: multiple SLE asymptotics}
specify the form of the limit~\eqref{eq: generic asymptotics limit}.
To motivate this requirement, in Proposition~\ref{prop: cascade relation} we show that it implies a cascade property for the
behavior of the other curves in the limit of collapsing the marked points $x_j$ and $x_{j+1}$.


Consider now the SDE defining the $j$:th curve,
and the process of the distance of its tip to the next marked points $X^{(j-1)}_0$
and $X^{(j+1)}_0$ on the left and right, respectively.
\begin{prop}\label{prop: local abs continuity wrt Bessel processes}
Fix $j, j \pm 1 \in \set{1,\ldots,2N}$, and
suppose that $\PartF$ is a positive solution 
to~\eqref{eq: multiple SLE PDEs} 
for which the limit of $|x_{j}-x_{j \pm 1}|^{-\Delta} \times \PartF(x_1,\ldots,x_{2N})$ as $x_j , x_{j\pm1} \to \xi$
exists and is non-zero, for all choices of 
$( x_k )_{k \neq j , j \pm 1}$
and 
$\xi \in (x_{j - \frac{3}{2} \pm \half} , x_{j + \frac{3}{2} \pm \half})$.
Let $\left( X^{(1)}_t, \ldots, X^{(j-1)}_t,X_t,X^{(j+1)}_t, \ldots, X^{(2N)}_t \right)_{t\in[0,\tau]}$
be a solution to the SDE~\eqref{eq: SDE}, where 
$b^{(j)}=\kappa\,\frac{\partial_j\PartF}{\PartF}$,
and $\tau$ is any stopping time such that for some $\varepsilon>0, R>0$ and
for all $t\leq\tau$, we have
$|X^{(k)}_t - X^{(l)}_t| \geq \varepsilon$ for all $k\neq l$, and
$|X^{(k)}_t - X_t| \geq \varepsilon$ for all $k\neq j \pm 1$,
and $|X^{(k)}_t| \leq R$ for all $k$, and $|X_t| \leq R$.
Then the law of the difference process $Y_t=|X_t-X^{(j\pm1)}_t|$, $t \in [0,\tau]$, 
is absolutely continuous with respect to the law of a linear time change of a Bessel process 
of dimension $\delta=1+2\Delta+\frac{4}{\kappa}$. According to the two possible cases, we have the following:
\begin{itemize}
\item If the limit \eqref{eq: generic asymptotics limit} 
is non-zero, then $\Delta=\frac{\kappa-6}{\kappa}$ and $\delta = 3 - \frac{8}{\kappa} < 2$,
and solutions to the SDE~\eqref{eq: SDE} exist
up to stopping times $\tau$ at which we have $Y_\tau = 0$ with positive probability.
\item  If the limit \eqref{eq: generic asymptotics limit} vanishes, then the limit
\eqref{eq: subleading asymptotics limit} exists and is non-zero, and we have
$\Delta=\frac{2}{\kappa}$ and $\delta = 1 + \frac{8}{\kappa} > 2$, so the distance $Y_t$ remains positive for all $t \in [0,\tau]$, almost surely.
\end{itemize}
\end{prop}
\begin{proof}
We prove the case where the sign $\pm$ is $-\,$; the other is similar.
We consider three probability measures $\SLEmeasure$, $\SLEmeasure_\chordal$, 
and $\tilde{\SLEmeasure}$ on $2N$-component stochastic processes
$\left( X^{(1)}_t, \ldots, X^{(j-1)}_t,X_t,X^{(j+1)}_t, \ldots, X^{(2N)}_t \right)_{t\in[0,\tau]}$. 
For each of the three, the path of the $j$:th component up to time $t$ 
is taken to determine the other components at time $t$ by $X^{(i)}_t=g_t(X^{(i)}_0)$,
where $(g_s)_{s \in [0,t]}$ is the solution to the Loewner equation \eqref{eq: Loewner equation}
with driving process $\driving_s = X_s$. The measures $\SLEmeasure$, $\SLEmeasure_\chordal$ 
and $\tilde{\SLEmeasure}$ thus essentially only differ by the law they assign to $(X_t)_{t \in [0,\tau]}$.

The statement of the proposition concerns the measure $\SLEmeasure$ 
defined by the SDE~\eqref{eq: SDE}. The second measure $\SLEmeasure_\chordal$ is taken to be the 
chordal $\SLEk$ measure, under which the driving process is $X_t=\sqrt{\kappa} \, B_t + X_0$.
Note that there exists a sequence $(\tau_n)_{n \in \bN}$ of stopping times increasing to $\tau$,
such that up to $\tau_n$ the measure 
$\SLEmeasure$ is absolutely continuous with respect to 
$\SLEmeasure_\chordal$, with Radon-Nikodym derivative $M_{\tau_n} / M_0$ 
given in terms of the
$\SLEmeasure_\chordal$-martingale
\begin{align*}
M_t=\prod_{i\neq j}(g_t'(X^{(i)}_0))^{\frac{6-\kappa}{2\kappa}}\;\times\;\PartF(X^{(1)}_t,\ldots,X^{(2N)}_t).
\end{align*}
Finally, the third measure $\tilde{\SLEmeasure}$ is defined up to time $\tau_n$ by its
Radon-Nikodym derivative $\tilde{M}_{\tau_n} / \tilde{M}_0$ with respect to 
$\SLEmeasure_\chordal$, where
\begin{align*}
\tilde{M}_t=(g_t'(X^{(j-1)}_0))^{\frac{6-\kappa}{2\kappa}}\;\times\;(X_t-X^{(j-1)}_t)^\Delta.
\end{align*}
By the assumption on $\PartF$, the ratio $M_t / \tilde{M}_t$ is bounded away
from $0$ and $\infty$ for all $t \leq \tau$, so we see that the probability measures $\SLEmeasure$
and $\tilde{\SLEmeasure}$ are mutually absolutely continuous. 
Under $\tilde{\SLEmeasure}$, we have, by Girsanov's theorem,
\begin{align*}
\ud X_t=\sqrt{\kappa}\,\ud \tilde{B}_t\,+\frac{\kappa\Delta}{X_t-X^{(j-1)}_t}\,\ud t,
\end{align*}
where $\tilde{B}_t$ is a $\tilde{\SLEmeasure}$-Brownian motion.
Consequently, the difference $Y_t=X_t-X^{(j-1)}_t$ satisfies
the SDE
\begin{align*}
\ud Y_t=\sqrt{\kappa}\,\ud \tilde{B}_t\,+\frac{\kappa\Delta+2}{Y_t}\,\ud t.
\end{align*}
By this SDE, under the measure $\tilde{\SLEmeasure}$, which is absolutely
continuous with respect to the original measure $\SLEmeasure$,
the process $Y_t$ is a linear time change of a Bessel process of dimension
$\delta=1+2\Delta+\frac{4}{\kappa}$. The last part of the statement
is essentially routine in view of
the fact that if $\Delta=\frac{2}{\kappa}$, then the dimension is
$\delta=1+\frac{8}{\kappa}>2$, and if $\Delta=\frac{\kappa-6}{\kappa}$,
then $\delta=3-\frac{8}{\kappa}<2$.
\end{proof}

%

\begin{prop}\label{prop: cascade relation}
Suppose $\PartF$ is a positive solution 
to~\eqref{eq: multiple SLE PDEs} 
such that the limit 
\begin{align}\label{eq: limit of partition functions}
\lim_{x_j,x_{j+1}\to\xi}\frac{\PartF(x_1,\ldots,x_{2N})}{|x_{j+1}-x_j|^{\frac{\kappa-6}{\kappa}}}=:\tilde{\PartF}(x_1,\ldots,x_{j-1},x_{j+2},\ldots,x_{2N})
\end{align}
exists for any $\xi \in (x_{j-1},x_{j+2})$, and assume that the limit function $\tilde{\PartF}$ is continuous and positive. 
Fix the points $(x_1,\ldots,x_{j-1},x_{j+2},\ldots,x_{2N})$ and the localization
neighborhoods $(U_1,\ldots,U_{j-1},U_{j+2},\ldots,U_{2N})$ so that they do not contain
a chosen point ${\xi \in (x_{j-1},x_{j+2})}$.
Then, as $x_{j+1},x_j\to\xi$,
the marginal law of the curves $(\gamma^{(1)}, \ldots, \gamma^{(j-1)},\gamma^{(j+2)}, \ldots, \gamma^{(2N)})$
under $\PR^{(\bH;x_1,\ldots,x_{2N})}_{(U_1,\ldots,U_{2N})}$
converges weakly to the measure $\PR^{(\bH;x_1,\ldots,x_{j-1},x_{j+2},\ldots,x_{2N})}_{(U_1,\ldots,U_{j-1},U_{j+2},\ldots,U_{2N})}$
obtained by the sampling procedure~\ref{proc: one by one} with the function $\tilde{\sZ}$.
\end{prop}
\begin{proof}
First note that the limit~\eqref{eq: limit of partition functions} is uniform on compact subsets of the positions
of the variables $(x_1, \ldots, x_{j-1}, \xi, x_{j+2},\ldots,x_{2N})$.
Namely, let $\chamber_{2N} \subset \bR^{2N}$ be the domain of definition of $\PartF$, and
interpret a compact subset $K$ of the positions of these variables as a subset of its
boundary $K \subset \bdry \chamber_{2N}$, by using $\xi$ as the value for both variables $x_j , x_{j+1}$.
In the closure $\cl{\chamber_{2N}}$ we can take
a compact $K'$ which contains a small open neighborhood $G$ of $K$, i.e. $K \subset G \subset K'$.
Now because $\tilde{\PartF}$ is continuous,
the limit~\eqref{eq: limit of partition functions} can be used to extend the ratio
$\PartF \, \big/ \, |x_{j+1}-x_j|^{\frac{\kappa-6}{\kappa}}$ continuously to $\bdry \chamber_{2N} \cap K'$.
This extension is uniformly continuous on $K$, so the limit~\eqref{eq: limit of partition functions} is
uniform over all choices of $(x_1, \ldots, x_{j-1}, \xi, x_{j+2},\ldots,x_{2N}) \in K$.

Consider the law of the $k$:th curve $\gamma^{(k)}$, $k \neq j, j+1$.
Its Radon-Nikodym derivative with respect to the
chordal $\SLEk$ is given by \eqref{eq: marginals}. But the expression~\eqref{eq: marginals}
is uniformly close to the corresponding one with $\tilde{\PartF}$, by the uniformity on compacts of
the limits $ g'(x_j)^{h} \, g'(x_{j+1})^{h} \; \big( \frac{x_{j+1}-x_j}{g(x_{j+1})-g(x_j)} \big)^{2h} \to 1$
and \eqref{eq: limit of partition functions}.
Since the Radon-Nikodym derivatives are uniformly close, the marginal laws of
the $k$:th curve are close in the topology of weak convergence.
To handle the the joint law of all curves other than $\gamma^{(j)}$, $\gamma^{(j+1)}$, apply the
same argument also in further steps of the sampling procedure~\ref{proc: one by one}.
\end{proof}


\bibliographystyle{annotate}

\newcommand{\etalchar}[1]{$^{#1}$}

\end{document}